\documentclass[aos]{imsart}
\input{myPreliminary}

\begin{document}
\begin{frontmatter}
\title{Optimal Difference-based Variance Estimators in Time Series: A General Framework}   
\runtitle{Optimal variance estimator}  

\begin{aug}
\author[A]{\fnms{Kin Wai} \snm{Chan}\ead[label=e1]{kinwaichan@cuhk.edu.hk}}
\address[A]{Department of Statistics, The Chinese University of Hong Kong, \printead{e1}}
\end{aug}

\begin{abstract}
Variance estimation is important for statistical inference. 
It becomes non-trivial when observations are masked by serial dependence structures 
and time-varying mean structures. 
Existing methods either ignore or sub-optimally handle these nuisance structures.
This paper develops a general framework  
for the estimation of the long-run variance for time series with non-constant means.
The building blocks are difference statistics.
The proposed class of estimators is general enough to cover many existing estimators. 
Necessary and sufficient conditions for consistency are investigated. 
The first asymptotically optimal estimator is derived. 
Our proposed estimator is theoretically proven to be invariant to arbitrary mean structures, 
which may include trends and a possibly divergent number of discontinuities. 
\end{abstract}


\end{frontmatter}

\section{Introduction}\label{sec:intro}
\subsection{Motivation and background}\label{sec:bkgd}
Let the observed time series $X_{1:n} = \{ X_1, \ldots, X_n \}$ be generated from the signal-plus-noise model: 
\begin{equation}\label{eqt:modelX}
	X_i = \mu_i + Z_i, \qquad i=1, \ldots, n, 
\end{equation}
where the deterministic signals $\mu_i$
and the zero-mean stationary noises $Z_i$ are not directly observable. 
Many statistics  
designed for inferring $\mu_{1:n} = \{ \mu_{1},\ldots, \mu_n \}$
admit the form $T_n = T_n(\widehat{v})$,
where $\widehat{v}$ is an estimator of the long-run variance (LRV)
$v = \lim_{n\rightarrow \infty} n\Var(\bar{Z}_n)$ of $\bar{Z}_n= \sum_{i=1}^n Z_i/n$.
Deriving a good estimator $\widehat{v}$ is, therefore, important, 
and is the major goal of this article.

Examples of such $T_n(\widehat{v})$ include, but are not restricted to,  
the Kolmogorov--Smirnoff (KS) change point test and its variants 
\citep{horvath1999,Crainiceanu2007,JuhlXiao2009,GHK2018}, 
mean constancy tests \citep{wu2004CP,DallaGiraitisPhillips2015}, 
mass excess tests of relevant mean changes \citep{DetteWu2019},
tests for monotone trends \citep{wu2001}, 
simultaneous confidence bands (SCBs) for trends \citep{wu_zhao_2007}, etc.
Serving as a normalizer in $T_n(\widehat{v})$,  
the estimator $\widehat{v}$ 
measures the significance of the signals $\mu_{i}$ relative to the noises $Z_{i}$. 
Constructing a good $\widehat{v}$ is nevertheless difficult due to 
two nuisance structures.

\begin{enumerate}
	\item \label{nuisanceStructure1}\emph{Nuisance structure 1: variability of $\mu_{1:n}$}.
			The stochastic variability of $Z_{1:n} = \{Z_1, \ldots, Z_n\}$ is masked by 
			the deterministic variability of $\mu_{1:n}$; see Figure~\ref{fig:dTAVC_I006_TSplot}. 
			Disentangling the variabilities of $\mu_{1:n}$ and $Z_{1:n}$ 
			can be challenging. 
			Without the nuisance structure 2 below, 
			this task was studied by, e.g.,  \cite{hall1990}. 
			Similar and extended results include
			\cite{anderson1971}, \cite{rice1984}, and \cite{LevineTecuapetla2019}. 
	\item \label{nuisanceStructure2}\emph{Nuisance structure 2: serial dependence of $Z_{1:n}$}.
			Under regularity conditions,
			$v = \sum_{k\in\mathbb{Z}}\gamma_k$ is a sum of infinitely many unknowns, 
			where $\gamma_k = \E(Z_0Z_{k})$ is the autocovariance function (ACVF).
			So, estimation of $v$ is hard.
			Without the nuisance structure 1 above, 
			this task was studied by, e.g., 
			\citet{OBM1984}, \cite{carlstein86}, \cite{newwy_west_1987}, \cite{kunsch89}, 
			\cite{andrews1991}, \cite{politis2011}, \cite{chanyau2013}, and \cite{chanyau2014_rTACM}.
\end{enumerate}

\begin{figure}[t]
\begin{center}
\includegraphics[width=\textwidth]{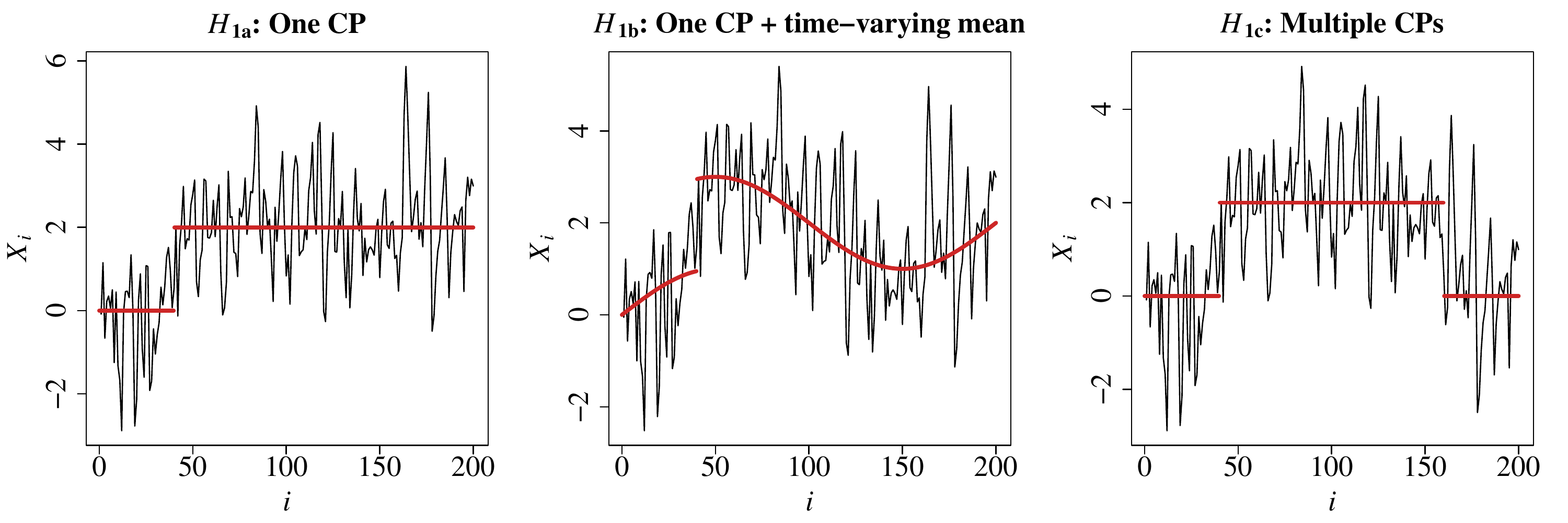} 
\end{center} 
\vspace{-0.6cm}
\caption{The black and red lines denote $X_{1:n}$ and $\mu_{1:n}$, respectively. 
		Here, 
		$H_{1a}$: $\mu_i = \Xi \mathbb{1}(i>2n/10)$; 
		${H}_{1b}$: 
		$\mu_i = \Xi \left\{\mathbb{1}(i>2n/10) + \sin(2\pi i/n)/2\right\}$;
		${H}_{1c}$: 
		$\mu_i = \Xi \mathbb{1}(2n/10<i<8n/10)$, 
		where $\Xi\in\mathbb{R}$ measures the magnitude of jump(s) and/or 
		the amplitude of trend.
		The noises $Z_{1:n}$ are generated from an autoregressive AR(2) model:
		$Z_i =  Z_{i-1}/2+ Z_{i-2}/5 + \varepsilon_i$ for each $i$, where 
		$\varepsilon_i$ follow $\Normal(0,1)$ independently.
		In particular, $n=200$ and $\Xi=2$ are used in the above plots.}
		\label{fig:dTAVC_I006_TSplot}
		\vspace{-0.4cm} 
\end{figure}

In this article, we propose a \emph{general} framework of
estimators of $v$; see Definition~\ref{def:DS} and Equation~(\ref{eqt:diffvarEst}). 
Necessary and sufficient conditions for consistency are derived;
see Theorems~\ref{thm:consistency_finite} and \ref{thm:consistency_infinite}.
They are proven to achieve the \emph{optimal} $\mathcal{L}^2$ rate of convergence 
under various strengths of serial dependence (see Theorems~\ref{thm:bias} and \ref{thm:var})
and are \emph{robust} against a wide class of mean structures (see Theorem~\ref{thm:robustness}).
The optimal $m$th order difference-based variance
estimator $\widehat{v}_{(m)}$ is given in Corollary \ref{coro:optim}, 
where the optimality refers to the best possible difference statistics used in the estimator. 
In particular, one special case (with $m=3$) is given by 
\begin{align*}
	\widehat{v}_{(3)} 
		= \sum_{|k| < \ell} \left\{1-\left(\frac{|k|}{\ell}\right)^2\right\} 
			\frac{1}{n} \sum_{i=mh+|k|+1}^n D_{i} D_{i-|k|}, 
\end{align*}
where $D_i = 0.1942X_i+0.2809X_{i-h}+0.3832X_{i-2h}-0.8582X_{i-3h}$ for each $i$, 
$\ell = O(n^{1/5})$, and $h = 2\ell$.
The proposed estimator with the optimally selected $\ell$ is presented in (\ref{eqt:suggestedEst}).
This estimator outperforms all existing estimators in terms of the mean-squared error (MSE) asymptotically; 
see (\ref{eqt:compareMSEexisting}).
We conclude this subsection with an example to illustrate the importance of this project.

\begin{example}[Change point detection]\label{eg:CP}
Suppose we want to test $H_0: \mu_1 = \ldots = \mu_n$. 
The celebrated KS change point (CP) test statistic (see, e.g., \citet{csorgo1997}) is defined as  
\begin{align}\label{eqt:KS}
	T_n(v) =  \frac{1}{\sqrt{n v}} \max_{k\in\{1,\ldots,n\}} 
				\left\vert \sum_{i=1}^{k} (X_i - \bar{X}_n) \right\vert, 
	\quad \text{where} \quad 
	\bar{X}_n = \frac{1}{n}\sum_{i=1}^n X_i.  
\end{align}
We reject $H_0$ at size $5\%$ if $T_n({v})>1.358$. 
Although this test is designed for a one-CP alternative ($H_{1a}$), 
it is still applicable to more complicated situations, e.g., 
a one-CP alternative in the presence of a smooth trend ($H_{1b}$), and 
a multiple-CP alternative ($H_{1c}$); see Figure~\ref{fig:dTAVC_I006_TSplot}. 
No matter which situation we consider, having a good estimator of $v$ is still necessary.
We compare two estimators: 
the classical Bartlett kernel estimator $\widehat{v}_{(\Andrews)}$ 
with a bandwidth selected by fitting an AR(1) model proposed in \cite{andrews1991}, and 
our proposed estimator $\widehat{v}_{(3)}$ 
with optimally selected parameters to be discussed in (\ref{eqt:suggestedEst}).

Consider the time series defined in Figure~\ref{fig:dTAVC_I006_TSplot}
with different magnitudes of jump $\Xi$.
We compute the power of the classical test $T_n(\widehat{v}_{(\Andrews)})$ 
and the proposed test $T_n(\widehat{v}_{(3)})$ against $\Xi$;
see Figure~\ref{fig:dTAVC_I006_KSpower}. 
Under $H_{1a}$,
the test $T_n(\widehat{v}_{(\Andrews)})$ is valid but less powerful than $T_n(\widehat{v}_{(3)})$ 
because $\widehat{v}_{(\Andrews)}$ is inaccurate for $v$ when $\Xi\neq 0$.
Under ${H}_{1b}$ or $H_{1c}$,
the test $T_n(\widehat{v}_{(\Andrews)})$ even fails to demonstrate monotone power when $\Xi$ increases  
because $\widehat{v}_{(\Andrews)}\rightarrow\infty$ in probability as $\Xi \rightarrow \infty$ in these cases.
So a robust and efficient estimator of $v$ is crucial.\hfill $\blacksquare$
\end{example}

\begin{figure} 
\begin{center}
\includegraphics[width=\textwidth]{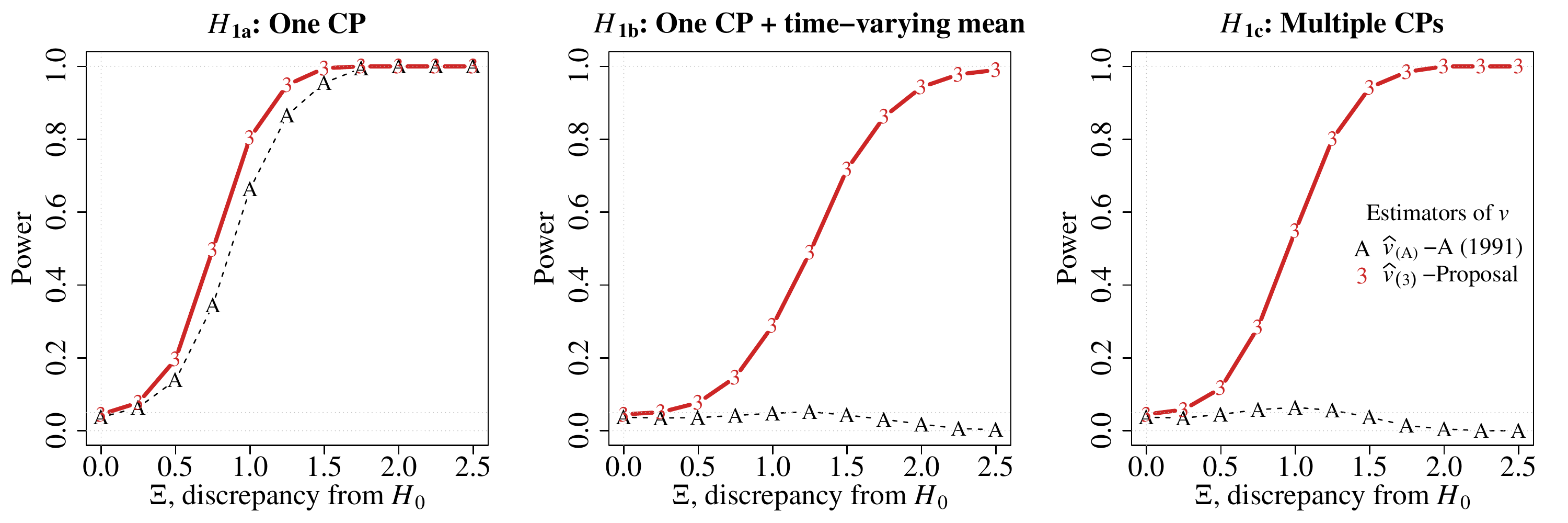} 
\end{center} 
\vspace{-0.6cm}
\caption{Power curves of the KS tests with the classical estimator $\widehat{v}_{(\Andrews)}$
			and the proposed estimator $\widehat{v}_{(3)}$
			in three different types of alternative hypotheses stated in Figure~\ref{fig:dTAVC_I006_TSplot}. }
		\label{fig:dTAVC_I006_KSpower}
		\vspace{-0.4cm} 
\end{figure}

\subsection{Notation and mathematical background}\label{sec:notations}
Let $\mu_i = \mu(i/n)$ for $i=1,\ldots, n$, where $\mu: [0,1]\rightarrow \mathbb{R}$ is a mean function. 
Suppose that $\mu(\cdot)$ consists of a continuous part $c(\cdot)$ and a step-discontinuous part $s(\cdot)$
such that 
\begin{align}\label{eqt:muDef}
	\mu(t) = c(t) + s(t),
	\qquad  
	s(t) = \sum_{j=0}^{\mathcal{J}} \xi_j \mathbb{1}(T_j/n \leq t < T_{j+1}/n),
\end{align}
where $\mathcal{J}$ is the number of discontinuities, 
$1\equiv T_0 < T_1 < \ldots < T_{\mathcal{J}} < T_{\mathcal{J}+1} \equiv n+1$ are the 
times of discontinuities, and
$\xi_0, \ldots, \xi_{\mathcal{J}}$ are the step sizes 
such that $\xi_j \neq \xi_{j-1}$ for each $j$.  
Note that $c(\cdot)$, $\mathcal{J}$, $\xi_0, \ldots, \xi_{\mathcal{J}}$, and $T_1, \ldots, T_{\mathcal{J}}$
are possibly dependent on $n$.
For example, $\mathcal{J}$ and $\xi_0, \ldots, \xi_{\mathcal{J}}$ can be divergent with $n$.
Denote the minimal gap between two consecutive CP times by 
\[
	\mathcal{G} = \min_{0\leq j\leq \mathcal{J}} (T_{j+1} - T_{j}).
\]
We measure the smoothness of $c(\cdot)$ by $\mathcal{C}$,
the maximum step magnitude of $s(\cdot)$ by $\mathcal{S}$, 
and the overall variability of $\mu(\cdot)$ by $\mathcal{V}$, where
\[
	\mathcal{C} = \sup_{0\leq t'<t\leq 1} \left\vert \frac{c(t)-c(t')}{t-t'} \right\vert, \qquad   
	\mathcal{S} = \sup_{0\leq j\leq\mathcal{J}} |\xi_j - \xi_{j+1}| , \qquad   
	\mathcal{V} = \int_0^1 \left\{ \mu(t) - \bar{\mu} \right\}^2 \, \dd t, 
\]
and $\bar{\mu} = \int_0^1 \mu(t)\, \dd t$. 
Clearly, $\mathcal{C}=0$ iff there is no trend effect;
$\mathcal{S}=0$ or $\mathcal{J}=0$ iff there is no discontinuity; and 
$\mathcal{V}=0$ iff the mean function is a constant.

Let $Z_i = g(\mathcal{F}_{i})$ for some measurable function $g$,
where $\mathcal{F}_{i}=(\ldots,\varepsilon_{i-1},\varepsilon_i)$  
and $\{\varepsilon_i\}_{i\in\mathbb{Z}}$ are independent and identically distributed (i.i.d.) innovations.
Let $\varepsilon_j'$ be an i.i.d. copy of $\varepsilon_j$, 
$\mathcal{F}_{i,\{j\}}=(\mathcal{F}_{j-1},\varepsilon_j',\varepsilon_{j+1},\ldots,\varepsilon_i)$,
and $Z_{i,\{j\}}=g(\mathcal{F}_{i,\{j\}})$.
Define $\mathcal{P}_i \,\cdot = \E(\cdot\mid\mathcal{F}_i)- \E(\cdot\mid\mathcal{F}_{i-1})$.
For $p\geq 1$, define
the \emph{physical dependence measure} and its aggregated value by 
\begin{eqnarray}\label{eqt:depMea}
	\theta_{p,i}=\| Z_i - Z_{i,\{0\}} \|_p  
	\qquad \text{and} \qquad
	\Theta_p = \sum_{i=0}^{\infty} \theta_{p,i},  
\end{eqnarray}
respectively, 
where $\left\| \,\cdot\, \right\|_p = \left( \E  \left| \,\cdot\, \right|^p \right)^{1/p}$. 
The finiteness of $\Theta_p$ 
provides a mild and easily verifiable condition for asymptotic theory;
see \citet{wu2005,wu2007,wu2011}. 

\begin{assumption}[Weak dependence]\label{ass:weakDep}
The noise sequence $\{Z_i\}_{i\in\mathbb{Z}}$ is a zero-mean strictly stationary time series
that satisfies $\E(Z_1^r)<\infty$ for some $r>4$, and $\Theta_4 <\infty$.
\end{assumption}

Indeed, Assumption~\ref{ass:weakDep} implies that the ACVFs are absolutely summable, i.e., 
$u_0 := \sum_{k\in\mathbb{Z}} |\gamma_k| <\infty$, which ensures the existence of 
$v = \sum_{k\in\mathbb{Z}} \gamma_k$.  
We remark that there exist other ways of quantifying dependence,
including various types of mixing coefficients \citep{Rosenblatt56,wolkonski1959} and  
near-epoch approach \citep{ibragimov1962}.
They have been widely adopted and studied;
see \citet{TaqquEberlein86} and \citet{bradley2005} for some surveys of results.
It is certainly interesting to develop our theoretical results under these settings, 
however, it is beyond the scope of this paper. We leave it for further study.

The following notation is used. 
Let $\mathbb{N} = \{1, 2, 3, \ldots\}$, $\mathbb{N}_0 = \{0, 1, 2, \ldots\}$, 
and $\mathbb{R}^+ = (0, \infty)$.
For any statement $E$, $\mathbb{1}_{\{E\}}=1$ if $E$ is true, otherwise $\mathbb{1}_{\{E\}}=0$.
For any $a,b\in\mathbb{R}$, $a^+ = \max(a,0)$ and $a\wedge b = \min(a,b)$.
For any $\{a_n\}_{n\in\mathbb{N}}$ and $\{b_n\}_{n\in\mathbb{N}}$ with $a_n, b_n \in\mathbb{R}^+$,   
the relation $a_n\sim b_n$ means $a_n/b_n\rightarrow 1$;
$a_n\asymp b_n$ means there is $C\in\mathbb{R}^+$ such that $1/C \leq a_n/b_n \leq C$ for all large $n$; 
$a_n\ll b_n$ or $a_n = o(b_n)$ means $a_n/b_n\rightarrow 0$; 
$a_n \lesssim b_n$ or $a_n = O(b_n)$ means there is $C>0$ such that $a_n/b_n \leq C$ for all large $n$.
Convergence in probability and convergence in distribution are denoted by ``$\inP$'' and ``$\inD$'',
respectively. 
Write $\|\cdot\|= \|\cdot\|_2$.
For any sequence of random variables $\{Z_n\}_{n\in\mathbb{N}}$,
$Z_n = O_p(a_n)$ means for any $\epsilon>0$ there exist $C\in\mathbb{R}^+$ and $N\in\mathbb{N}$ such that 
$\pr(|Z_n/a_n|>C) < \epsilon$ for all $n>N$;
$Z_n = o_p(a_n)$ means $Z_n/a_n \inP 0$.  
For any estimator $\widehat{\theta}$ of $\theta$, 
denote $\Bias(\widehat{\theta};\theta)=\Bias(\widehat{\theta})=\E(\widehat{\theta})-\theta$ 
and 
$\MSE(\widehat{\theta};\theta)=\MSE(\widehat{\theta})=\E(\widehat{\theta}-\theta)^2$.

In this article,
we propose and study a general framework for estimating 
\begin{align}\label{eqt:def_v}
	v = \lim_{n\rightarrow\infty} n \Var(\bar{Z}_n)  = \sum_{k\in\mathbb{Z}} \gamma_k, 
	\qquad \text{where}\qquad
	\gamma_k =\Cov(Z_0, Z_k),  
\end{align}
by using difference statistics.
This article is structured as follows.
Section~\ref{sec:general} defines the proposed class of estimators.
We show that it covers many existing estimators as special cases. 
Section~\ref{sec:invariance} demonstrates its invariance to mean structures. 
Section~\ref{sec:consistency} derives 
the necessary and sufficient conditions for consistency.
Section~\ref{sec:optimal} shows that the proposed estimator is asymptotically optimal. 
Section~\ref{sec:imp_gen} addresses implementation issues and generalization.
Section~\ref{sec:experiment} presents 
simulation experiments, applications, and real-data examples. 
We conclude the paper with a summary of major contributions and possible future work 
in Section~\ref{sec:conclusion}.
All proofs are deferred to a separate supplementary note.
An R-package \texttt{"dlrv"} is available on the author website.

\section{A general framework for variance estimation}\label{sec:general}
\subsection{Difference-based statistics}
Variance estimators usually require \emph{centering}
to achieve mean invariance.   
For example, 
if $\mu_1 =\ldots = \mu_n$, 
one may \emph{globally} center the data as $D_i' = X_i - \bar{X}_n$; 
if $\mu_{1:n}$ are not constant, one may \emph{locally} center each $X_{i}$ by 
the kernel method and the lag-1 difference:  
\begin{equation}\label{eqt:exampleD}
	D_i'' = X_i - \frac{\sum_{j} H(\frac{i-j}{m/2}) X_j }{ \sum_{j'} H(\frac{i-j'}{m/2})}
	\qquad \text{and}\qquad
	D_i''' = X_i - X_{i-1}, 
\end{equation}
where $H(\cdot)$ is a kernel, and $m/2$ is a bandwidth.
The statistics $D_i'$, $D_i''$, and $D_i'''$ are special cases of the following 
class of general difference statistics.

\begin{definition}[Difference statistics]\label{def:DS}
A real-valued sequence $\{d_j\}_{j=0}^{m}$ is said to be an 
$m$th order \emph{difference sequence} if $d_0 + \ldots + d_m =0$. 
If, in addition, $\delta_0 = d_0^2 + \ldots + d_m^2 =1$, then $\{d_j\}$ is said to be \emph{normalized}. 
For $h\in\mathbb{N}$,  
the $m$th order lag-$h$ \emph{difference statistics} are defined as 
\begin{equation}\label{eqt:diffStat}
	D_{i} = \sum_{j=0}^m d_j X_{i-jh},   
	\qquad 
	i=mh+1, \ldots, n.
\end{equation}
The zeroth order difference statistics are  $D_i = X_i - \bar{X}_n$ for $i=1,\ldots,n$.
Also denote $\delta_s = \sum_{j=|s|}^m d_j d_{j-|s|}$ for $|s|\leq m$ and $\delta_s = 0$ for $|s|>m$.
\end{definition}

The condition $\sum_{j=0}^m d_j=0$ is used to ensure that  
$\E(D_i) \approx 0$  
when $\mu_{i}\approx \mu_{i-h}\approx \ldots \approx \mu_{i-mh}$.
This property is important for deriving asymptotic mean invariance of statistics based on $D_i$;
see Section~\ref{sec:invariance} for a precise and rigorous definition of mean invariance. 
The requirement $\delta_0 =1$ 
is used to regularize $D_i$ such that 
$\Var(D_i) = \Var(X_i)$ when $X_{1:n}$ are serially uncorrelated. 
One can easily normalize $d_j$
by $d_j/\sqrt{\delta_0}$ provided that $\delta_0 \neq 0$.
From now on, we assume the difference sequence $\{d_j\}$ is normalized.  
The lag parameter $h$ is used to control how frequent the observations are
used for constructing one difference statistic. 
When the data are independent, $h=1$ works well.
When the data are serially dependent, 
a larger $h$ can be used to reduce the serial dependence among the observations 
that are used in the same difference statistic. 
Some difference sequences are shown in Example~\ref{eg:dj}.

\begin{example}\label{eg:dj}
Some commonly used difference sequences $\{d_j\}_{j=0}^m$ are listed below. 
\begin{itemize}
	\item Binomial differencing:
			$d_j = {m\choose j}(-1)^j / {2m\choose m}^{1/2}$ for $j=0,\ldots, m$. 
			It gives $\delta_k = (-1)^k (m!)^2/\{ (m+k)!(m-k)!\}$ for $k=0,1,\ldots,m$.
	\item Local differencing: 
			$d_0 = \sqrt{m/(m+1)}$ and $d_j= -1/\sqrt{m^2+m}$ for $j=1, \ldots, m$.
			It gives $\delta_0=1$ and $\delta_k = -k/(m^2+m)$ for $k=1,\ldots,m$.
	\item \cite{hall1990}: 
			Define $\{d_j\}_{j=0}^m$ by minimizing $\sum_{k=1}^m \delta_k^2$; 
			see Table~\ref{tab:optimalSeq} for the solution. 
			It gives $\delta_0=1$ and $\delta_k = -1/(2m)$ for $k=1,\ldots,m$. 
			\hfill $\blacksquare$
\end{itemize}
\end{example}

\begin{table}
\begin{center}
\begin{tabular}{|c|ccccc|}
\hline
$m$ & $d_0$ & $d_1$ & $d_2$ & $d_3$ & $d_4$  \\\hline\hline
$1$ & $0.7071$ & $-0.7071$ &-&-&-\\
$2$ & $0.8090$ & $-0.5000$ & $-0.3090$ &-&-\\ 
$3$ & $0.1942$ & $0.2809$ & $0.3832$ & $-0.8582$ &-\\
$4$ & $0.2708$ & $-0.0142$ & $0.6909$ & $-0.4858$ & $-0.4617$\\
\hline
\end{tabular}
\end{center}
\vspace{-0.2cm}
\caption{\cite{hall1990}'s difference sequence $\{d_j\}_{j=0}^m$ for $m=1,\ldots,4$.}\label{tab:optimalSeq}
\vspace{-0.4cm} 
\end{table}

Note that $m=m_n$ is allowed to diverge with $n$. In this case, we need the following assumption 
to regularize the difference sequence. 

\begin{assumption}\label{ass:summability_dj}
The difference sequence $\{d_j\}$ satisfies 
(i) $\sup_{n\in\mathbb{N}}\sum_{j=0}^m|d_j|<\infty$, and
(ii) $\sup_{n\in\mathbb{N}}\sum_{|s|\leq m}|\delta_s|<\infty$.
\end{assumption}

\subsection{Proposed difference-based variance estimator}
Since $D_{mh+1}, \ldots, D_n$ are approximately centered at zero, 
it motivates us to utilize them as building blocks for 
estimating $v$.
We define the
\emph{$m$th order difference-based estimator} of $v$ by 
\begin{align}\label{eqt:diffvarEst}
	\widehat{v}
		= \sum_{|k| < \ell} K\left(\frac{k}{\ell}\right) \widehat{\gamma}_k^D, 
	\qquad \text{where}\qquad
	\widehat{\gamma}_{k}^D
		= \frac{1}{n} \sum_{i=mh+|k|+1}^n D_{i} D_{i-|k|} .
\end{align}
We may write $\widehat{v}$ as $\widehat{v}_{(m)}$ to emphasize the order $m$.
In (\ref{eqt:diffvarEst}), 
\begin{align}\label{eqt:def_m_l_h}
	m=m_n\in\mathbb{N}_0, \qquad \ell=\ell_n\in\{1, \ldots,n\}, \qquad h=h_n\in\{1, \ldots,n\}
\end{align}
are the order of differencing, bandwidth parameter, and lag parameter of the estimator $\widehat{v}$, respectively.
The function 
$K:\mathbb{R}\rightarrow\mathbb{R}$ is called a kernel, which satisfies that 
$K(0)=1$, 
$K(t)=K(-t)$ for all $t$, 
$K(t)=0$ for $|t|\geq1$, and 
$K$ is continuous on $(-1,1)$. 
Popular kernels include the Bartlett kernel $K_{\Bart}(t) = (1-|t|)^+$
and the flat-top truncated kernel $K_{\Flat}(t)=\mathbb{1}_{(|t|< 1)}$.
In (\ref{eqt:diffvarEst}), one may alternatively use 
$\sum_{i=mh+|k|+1}^n (D_{i}-\bar{D}_n)(D_{i-|k|}-\bar{D}_n)/n$
instead of $\widehat{\gamma}_{k}^D$,
where $\bar{D}_n = \sum_{i=mh+1}^n D_i/n$. 
It does not affect the asymptotic results in this article.

The kernel estimator $\widehat{v}$ can be written in a subsampling form.
For each $i=\ell, \ldots,n$, 
define 
the $i$th subsample of size $\ell$ as $\{D_t : t \in \Lambda_i\}$, 
where $\Lambda_i = \{i-\ell+1, \ldots, i\}$.
If $\mathcal{I}\subseteq \{mh+\ell+1, \ldots, n\}$ 
is the set of subsample indices to be used, 
then the \emph{subsampling estimator} of $v$ is defined as  
\begin{align}\label{eqt:Stype_vHat}
	\widehat{v}'
		= \frac{1}{|\mathcal{I}|} \sum_{i\in\mathcal{I}} \widehat{v}'(i),
	\qquad\text{where}\qquad
	\widehat{v}'(i) = \sum_{t,t'\in\Lambda_i} \frac{K(|t-t'|/\ell)}{\ell - |t-t'|} D_t D_{t'},
\end{align}
and $|\mathcal{I}|$ is the total number of subsamples. 
The overlapping subsamples and non-overlapping subsamples utilize 
$\mathcal{I}_1 = \{mh+1+\ell, \ldots, n\}$ and 
\begin{align*}
	\mathcal{I}_0 &= \{mh+1+\ell, 2(mh+1+\ell)  , \ldots, \lfloor n/(mh+1+\ell) \rfloor (mh+1+\ell)\},
\end{align*}
respectively. 
Similar ideas can be found in, e.g., \cite{carlstein86} and \citet{peter1987}.  
The estimator $\widehat{v}'$ can be regarded as a ``bagged'' estimator of $v$ 
by averaging the rough estimators (or weak ``learners'') 
$\{\widehat{v}'(i)\}_{i\in\mathcal{I}}$.
If computational time is a concern, we may use the non-overlapping subsamples. 
However, its statistical efficiency is reduced; see, e.g., \cite{Alexopoulos2011}. 
On the other hand, if the overlapping subsamples are used, 
the estimators $\widehat{v}'$ and $\widehat{v}$ are asymptotically equivalent
in the following sense.

\begin{proposition}[Asymptotic equivalence of $\widehat{v}$ and $\widehat{v}'$]\label{prop:Ktype_Stype_equiv}
Consider $\widehat{v}'$ with the overlapping subsamples
$\mathcal{I} = \mathcal{I}_1$, and the order of differencing $m=m_n$. 
If Assumptions~\ref{ass:weakDep}--\ref{ass:summability_dj} hold, 
$1/\ell + (\ell+mh)/n \rightarrow 0$, and 
$\mathcal{G}\gtrsim \ell+mh$, 
then, for any $\mu(\cdot)$ and $K(\cdot)$, we have
\begin{align}\label{eqt:Ktype_Stype_equiv}
	\|\widehat{v} - \widehat{v}'\| 
		= O\left( \frac{\ell+mh}{n} \right) (1+\| \widehat{v}\|)
			+ r_{\sub} ,
\end{align}
where 
$r_{\sub} = O\{ \mathcal{C}^2(\ell+mh)^4/n^3  + \mathcal{S}^2/n \}$.
\end{proposition}

The proof of Proposition \ref{prop:Ktype_Stype_equiv} 
can be found in Section \ref{sec:proof_Ktype_Stype_equiv} of 
the supplement.
By Minkowski's inequality, 
(\ref{eqt:Ktype_Stype_equiv}) implies 
$\|\widehat{v}-\widehat{v}'\| \leq O\{(\ell+mh)/n\}(1+ v + \|\widehat{v}-v\|) + r_{\sub}$, 
where the root-MSE $\|\widehat{v}-v\|\rightarrow 0$ if $\widehat{v}$ is $\mathcal{L}^2$ consistent for $v$.
So, (\ref{eqt:Ktype_Stype_equiv}) reduces to $\|\widehat{v}-\widehat{v}'\| = O\{(\ell+mh)/n\} + r_{\sub}$
if $\widehat{v}$ is $\mathcal{L}^2$ consistent. 
The remainder term $r_{\sub}$ is negligible if $\ell+mh$ is not too large. 
For example, if $\ell +mh= O(n^{\theta})$ for some $\theta\in(0,1/2]$, 
then $r_{\sub} = O\{(\mathcal{C}^2 + \mathcal{S}^2)/n\}$. 
We emphasize that 
Proposition~\ref{prop:Ktype_Stype_equiv} is true even for 
a possibly divergent $m=m_n$ and 
a possibly non-constant $\mu(\cdot)$
under the regularity conditions in Proposition~\ref{prop:Ktype_Stype_equiv}.

\subsection{Existing variance estimators}
Many popular variance estimators admit the forms 
of (\ref{eqt:diffvarEst}) or (\ref{eqt:Stype_vHat}).
They are presented and categorized in Examples~\ref{eg:uncorrelated}--\ref{eg:general_mInf} according to the values of 
$h$ and $m$.

\begin{example}[Serially uncorrelated case: $h=1$]\label{eg:uncorrelated}
Assume $\gamma_1= \gamma_2 = \ldots =0$. 
Then $v=\sum_{k\in\mathbb{Z}} \gamma_k$ reduces to $v=\gamma_0$.
\cite{hall1990} proposed to estimate $v$ by the $\widehat{v}$ in (\ref{eqt:diffvarEst}) with $h=1$.
Such $\widehat{v}$ is just an estimator of the marginal variance $\gamma_0$ but not $v=\sum_{k\in\mathbb{Z}}\gamma_k$. 
Recently, \cite{TecuapetlaMunk2017} and \citet{LevineTecuapetla2019} extended it to 
estimation of $\gamma_k$ for $M$-dependent time series, i.e., $\gamma_k=0$ for $|k|>M$. 
Their proposal is a special case of $\widehat{\gamma}_k^{D}$ in
(\ref{eqt:diffvarEst}) when $n\rightarrow\infty$.
The assumption of $M$-dependence can be restrictive in real applications.
Most importantly, they did not consider the estimation of $v=\sum_{k\in\mathbb{Z}} \gamma_k$.\hfill $\blacksquare$
\end{example}

\begin{example}[Constant-mean case: $m=0$]\label{eg:constantMean}
Assume $\mu_1 = \ldots = \mu_n$.
Let $D_i' = X_i - \bar{X}_n$ for each $i$.
Then $v$ is estimated by $\widehat{v}$ and $\widehat{v}'$ with $m=0$:
\begin{align}\label{eqt:vHat_m0}
	\widehat{v} = \sum_{|k|\leq \ell} 
			\frac{K( k/\ell)}{n}\sum_{i=1+|k|}^n D_i'D_{i-|k|}',
	\qquad 
	\widehat{v}' = \frac{1}{|\mathcal{I}|} \sum_{i\in\mathcal{I}} \sum_{t,t'\in\Lambda_i} 
		\frac{K(|t-t'|/\ell)}{\ell - |t-t'|} D_t' D_{t'}'.
\end{align}
The kernel estimator $\widehat{v}$ in (\ref{eqt:vHat_m0})
has a long history in statistics, econometrics and operational research; see, e.g., 
\citet{newwy_west_1987}, \citet{andrews1991}, and \citet{politis1999}.
The subsampling estimator $\widehat{v}'$ in (\ref{eqt:vHat_m0}) is studied in, e.g., 
\cite{song1995} and \citet{chanyau2015_hoc}.
If $K=K_{\Bart}$ and $\mathcal{I}=\mathcal{I}_1$, 
then $\widehat{v}' = \sum_{i=\ell}^n (\sum_{j=i-\ell+1}^i D'_j)^2/\{\ell(n-\ell+1) \}$
is the well-known overlapping batch means (OBM) estimator \citep{OBM1984}. 
Besides, \citet{carlstein86} and \citet*{Alexopoulos2011} studied 
the non-overlapping and partially overlapping subsamples,
however, these schemes are suboptimal in terms of MSE. \hfill $\blacksquare$
\end{example}

\begin{example}[General case: $m=1$, $h\rightarrow\infty$]\label{eg:general_m1}
In the presence of both serial dependence and time-varying means, 
the estimation of $v$ is less well studied. 
Let $\widetilde{v}'=\ell\sum_{i\in \mathcal{I}} (S_i - S_{i-\ell})^2 /\{2 |\mathcal{I}|\}$, 
where 
$S_i = \sum_{j=i-\ell+1}^{i} X_j/\ell$ is the $i$th subsample mean. 
This class of estimators is independently proposed by various authors. 
For example,
\citet{DetteWu2019} used $\mathcal{I} = \mathcal{I}_1$, 
whereas  
\citet{wu2001}, \citet{wu2004CP}, \citet{wu_zhao_2007}, \citet{dette2020}, 
and \citet{ChenWangWu2021} used $\mathcal{I} = \mathcal{I}_0$.
In either case, 
$\widetilde{v}'$ is just a special case of $\widehat{v}'$ with $m=1$, $h=\ell=O(n^{1/3})$ and 
$K=K_{\Bart}$. 
Moreover, none of them provides the optimal value of $\ell$. \hfill $\blacksquare$
\end{example}

\begin{example}[General case: $m\rightarrow\infty$, $h=1$]\label{eg:general_mInf}
\cite{Altissimoa_Corradic_2003} proposed to locally center $X_i$  
by using the $D_i''$ defined in (\ref{eqt:exampleD})
with $H= K_{\Flat}$, $K = K_{\Bart}$, 
and $m\rightarrow\infty$. 
Their estimator is asymptotically equivalent to 
$\widehat{v}$ 
with 
$d_j =  \left\{  \mathbb{1}_{(j=\lfloor m/2\rfloor )} - w_j \right\}/ c$
for $j=0, \ldots, m$, 
where $c,w_0, \ldots, w_m$ are some constants such that $w_0 + \ldots + w_m = 1$ and $\delta_0 = 1$.
A similar proposal can be found in \citet{JuhlXiao2009}.\hfill $\blacksquare$
\end{example}

Some estimators cannot be expressed as 
(\ref{eqt:diffvarEst}) or (\ref{eqt:Stype_vHat}); see Examples~\ref{eg:1CPvEst}--\ref{eg:insuff_diff}.
All of them are suboptimal or require restrictive assumptions.

\begin{example}[Removal of one CP]\label{eg:1CPvEst}
\citet{Crainiceanu2007} proposed to estimate one single potential CP $T_1$ by 
the standard CUSUM-type estimator $\widehat{T}_1$. 
After centering 
$\{X_i\}_{i=1}^{\widehat{T}_1-1}$ and $\{X_i\}_{i=\widehat{T}_1}^n$ by their respective sample means,
one may apply (\ref{eqt:vHat_m0}) to the centered series to estimate $v$.
This method is vulnerable to the one-CP assumption. 
Although it can be extended to handle multiple CPs,
the accumulated errors may ruin the final estimator.  

Recently, 
\cite{DehlingFriedWendler2020} proposed to split $X_{1:n}$ into three disjoint subsamples of 
(approximately) equal length
so that (\ref{eqt:vHat_m0}) can be applied to each of the three subsamples. 
The final estimator is the sample median of the three estimators.
This method incurs a huge loss of efficiency.  
It is remarked that their method is applied to ranks of $X_{1:n}$ instead of $X_{1:n}$,
but their idea is still applicable generally. 
\hfill $\blacksquare$
\end{example}

\begin{example}[Mean and median of absolute deviations]\label{eg:Wu_median}
Apart from the estimator $\widetilde{v}'$ in Example~\ref{eg:general_m1}, 
\citet{wu_zhao_2007} also proposed two other estimators that utilize 
the sample mean and sample median of $\{|S_i-S_{i-\ell}|\}_{i\in\mathcal{I}_1}$, i.e.,  
\[
	\widetilde{v}'' = \frac{\pi}{4( \lfloor n/\ell\rfloor -1)^2}  \sum_{i\in\mathcal{I}_1}|S_i - S_{i-\ell}|
	\qquad \text{and}\qquad
	\widetilde{v}'''
		= \frac{1}{2z_{3/4}^2} \median_{i\in\mathcal{I}_1}|S_i - S_{i-\ell}|,
\] 
where $\median_{k\in\mathcal{K}} x_k$ is the sample median of $\{x_k\}_{k\in\mathcal{K}}$,
and $z_p$ is the $100p\%$ quantile of $\Normal(0,1)$.
They proved that the convergence rates of $\widetilde{v}''$ and $\widetilde{v}'''$ are 
much slower than that of the $\widehat{v}'$ in Example~\ref{eg:general_m1}. 
\hfill $\blacksquare$
\end{example}

\begin{example}[Insufficient differencing]\label{eg:insuff_diff}
\cite{chan2020} proposed an estimator that is asymptotically equivalent to 
$\widecheck{v}=\sum_{|k|\leq \ell} K(k/\ell) \widecheck{\gamma}_k$, 
where $\widecheck{\gamma}_k = \sum_{i=k+\ell+1}^n X_i(X_{i-k}-X_{i-k+\ell})/n$.
It is an incomplete special case of $\widehat{v}$ with $m=1$ and $h=\ell$. 
It is an incomplete version because $\widecheck{\gamma}_k$ 
is constructed by the product of the raw observation $X_i$ and 
the difference statistic $X_{i-k}-X_{i-k+\ell}$, 
whereas our proposed statistic $\widehat{\gamma}_{k}^D$ in (\ref{eqt:diffvarEst})
is constructed by the product of two difference statistics $D_i$ and $D_{i-k}$. 
We prove that $\widehat{v}$ is uniformly better than this 
``insufficient'' difference-based estimator $\widecheck{v}$. 
\hfill $\blacksquare$
\end{example}

We also remark that 
some statistical procedures do not require estimation of the LRV by utilizing self normalization; 
see, e.g., \cite{lobato2001} and \citet{shaoxf2010}. 
However, different specifically designed self-normalizers may be needed for 
handling different types of mean structure; see, e.g., \cite{zhao2011}, \citet{zhang2018}, and \citet{pevsta2020}. 
This alternative approach may also lead to a decrease in power or statistical efficiency.
Nevertheless, they enjoy some added appealing properties.   
We refer interested readers to an excellent review by \citet{shao2015review}.

\subsection{Interpretation and representation}\label{sec:representation}
Recall the definitions of $h$ and $\ell$ in (\ref{eqt:def_m_l_h}).
We parametrize $h = \ell \lambda \in \mathbb{N}$ for some 
$\lambda:=\lambda_n \rightarrow\lambda_{\infty} \in[0, \infty]$. 
The goal of this section is to provide statistical interpretations of $\widehat{v}$
under different values of $\lambda$. 

\begin{proposition}\label{prop:effKernel}
Suppose Assumptions~\ref{ass:weakDep}--\ref{ass:summability_dj} hold, and
$1/\ell + (\ell+mh)/n = o(1)$. 
Let the \emph{differencing kernel} be
\begin{equation}\label{eqt:effKernel}
	K_{\diff}(t) = \sum_{s=\left\lceil -(1+t)/\lambda\right\rceil}^{\left\lfloor (1-t)/\lambda\right\rfloor}
				\delta_{|s|} K\left( t + \lambda s \right), \qquad t\in\mathbb{R}. 
\end{equation}
Define the \emph{differencing kernel estimator} as 
$\widehat{v}_{\diff} = \sum_{|k| \leq \ell+mh} K_{\diff}\left( {k}/{\ell}\right) \widehat{\gamma}_{k}^X$, 
where $\widehat{\gamma}_{k}^X = \sum_{i=|k|+1}^n (X_i - \bar{X}_n)(X_{i-|k|}- \bar{X}_n)/n$. 
\begin{enumerate}
	\item \label{item:repKdiff}(Representation) 
			If $\mu(t)\equiv \mu_0$ for all $t\in[0,1]$, then, as $n \rightarrow\infty$,
			\begin{align}\label{eqt:Error_vHat_vDiffHat}
				\|\widehat{v} - \widehat{v}_{\diff}\| = O\left\{ (\ell+mh)/n \right\}.
			\end{align}
			It remains true if 
			$\widehat{\gamma}^X_k$ is replaced by
			$\widehat{\gamma}^Z_k = \sum_{i=|k|+1}^n Z_iZ_{i-|k|}/n$
			in $\widehat{v}_{\diff}$.
	\item \label{item:KdiffProperty}(Differencing property) 
			If $m=0$, then $K_{\diff} = K$.
			If $m>0$, then $K_{\diff}$ satisfies that 
			$\sum_{|k|\leq \ell+mh} K_{\diff}\left( {k}/{\ell}\right) = 0$.
	\item \label{item:Kdiff0eq1}(Kernel necessity) If $\lambda\geq 1$, then $K_{\diff}(0)=1$. 
	\item \label{item:matching}(Matching property) If $\lambda \geq 2$, then $K_{\diff}(t) = K(t)$
			for all $t\in[-1,1]$.
\end{enumerate}
\end{proposition}

The proof can be found in Section~\ref{sec:proof_prop:effKernel} of the supplement. 
Proposition~\ref{prop:effKernel} (\ref{item:repKdiff}) states that  
$\widehat{v}$ smoothes $\{\widehat{\gamma}^X_k\}$ 
by the distorted kernel $K_{\diff}$ instead of 
the intended kernel $K$. 
So, $\widehat{v}$ may not inherit the properties of $K$, 
e.g., the higher-order property of achieving a faster convergence rate of $\widehat{v}$ 
as in \cite{andrews1991}.
Proposition~\ref{prop:effKernel} (\ref{item:KdiffProperty}) states that $\widehat{v}_{\diff}$
is invariant to $\mu(\cdot)$ because the kernel weights must sum to zero.
This property is not achieved by most commonly used kernels. 
Proposition~\ref{prop:effKernel} (\ref{item:Kdiff0eq1}) states that $K_{\diff}$ satisfies 
the minimal requirement as a kernel if $\lambda \geq 1$.
A striking fact conveyed by Proposition~\ref{prop:effKernel} (\ref{item:matching}) is that 
if $\lambda=h/\ell$ is large enough ($\geq 2$), then 
\[
	\widehat{v} 
		\;\equiv\; \sum_{|k|\leq \ell} K\left(\frac{k}{\ell}\right) \widehat{\gamma}^D_k 
		\;=\; \sum_{|k|\leq \ell} K\left(\frac{k}{\ell}\right) \widehat{\gamma}^X_k  
			+ r_{\diff},
\]
where the remainder term $r_{\diff}$ involves only $\{\widehat{\gamma}_k\}_{|k|>\ell}$, 
which are expected to have a negligible contribution to $\widehat{v}$ owing to weak dependence 
(Assumption~\ref{ass:weakDep}). 
In this case, 
$\widehat{v}$ correctly uses the intended kernel $K(\cdot)$ to smooth 
$\{\widehat{\gamma}_k\}_{|k|\leq \ell}$.
However, the existing estimators in Example~\ref{eg:general_m1}, 
which employ $\lambda=h/\ell=1$, 
incorrectly utilize another kernel.  
Example~\ref{eg:plotKeff} below visualizes this fact.

\begin{example}\label{eg:plotKeff}
Consider $m=1$ and $K=K_{\Bart}$. 
Figure~\ref{fig:dTAVC_effectiveKernel} visualizes how $K_{\diff}$ changes with $\lambda$.
When $\lambda <1$, we have $K_{\diff}(0) \neq 1$.
So, $K_{\diff}$ is not even qualified as a kernel for estimating $v$. 
When $1\leq \lambda<2$, we have $K_{\diff}(0)=1$ but $K_{\diff}(t)\not\equiv K(t)$ even for $|t|\leq 1$. 
It implies that $K_{\diff}$ does not share the same properties as $K$. 
When $\lambda \geq 2$, we have $K_{\diff}(t) \equiv K(t)$ for all $|t|\leq 1$.
In this case, $K_{\diff}$ copies most properties of the kernel $K$. \hfill $\blacksquare$

\begin{figure} 
\begin{center}
\includegraphics[width=\textwidth]{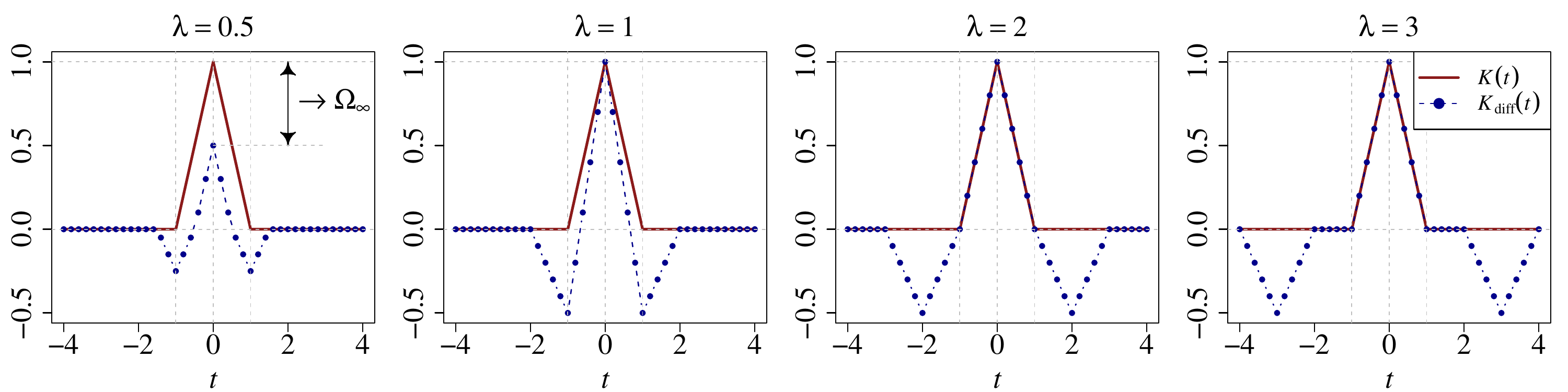} 
\end{center} 
\vspace{-0.6cm}
\caption{Comparisons between $K$ and $K_{\diff}$
when $m=1$, $K=K_{\Bart}$ and $\lambda\in\{1/2, 1,2,3\}$.
The quantity $\Omega_{\infty}$ in the first plot is defined in Assumption~\ref{ass:d_K_orthogonalNot}.}
		\label{fig:dTAVC_effectiveKernel}
\vspace{-0.4cm} 
\end{figure}
\end{example}

\begin{remark}\label{rem:KeffForm}
The differencing kernel $K_{\diff}$ may depend on $n$ when $\lambda=\lambda_n$ 
or $m=m_n$ depends on $n$.
We do not recommend to use $\widehat{v}_{\diff}$ in practice
as it has a larger influence by a non-constant $\mu(\cdot)$ than our proposed $\widehat{v}$ 
and $\widehat{v}'$ in (\ref{eqt:diffvarEst}) and (\ref{eqt:Stype_vHat}).
The representation (\ref{eqt:Error_vHat_vDiffHat}) is true only under the constant $\mu(\cdot)$ assumption. 
If $\mu(\cdot)$ is not a constant, the upper bound for $\|\widehat{v} - \widehat{v}_{\diff}\|$ 
is larger than that in (\ref{eqt:Error_vHat_vDiffHat}). 
It may not be enough to establish asymptotic equivalence. 
Nevertheless, it is informative to use $K_{\diff}$ to 
understand the proposed $\widehat{v}$. 
\end{remark}

\subsection{Dual representations}\label{sec:dual}
There are two types of ambiguous dual representations of $\widehat{v}$. 
First, a high-order but sparse difference sequence can be represented by 
a lower-order sequence with a larger lag $h$, 
e.g., the 4th order sequence $\{-1/\sqrt{2},0,0,0,1/\sqrt{2}\}$ 
with $h=1$ is equivalent to 
the 1st order  sequence $\{ -1/\sqrt{2},1/\sqrt{2}\}$ with $h=4$.
Sparse difference sequences lead to $\delta_s=0$ for all small $s$.
Second, consider $\widehat{v}$ with kernel $K^{\circ}(\cdot)$, bandwidth $\ell^{\circ}$, and lag $h^{\circ}$.
The estimator does not change if we stretch the kernel to $K(\cdot) = K^{\circ}(C\cdot)$
with $\ell = C\ell^{\circ}$ and $h = h^{\circ}$, where $C>1$.
This type of stretched kernel truncates earlier than $\pm1$, i.e., $K(t)=0$ for $|t|\geq 1/C$.
Assumption~\ref{ass:d_K_orthogonalNot} below rules out these ambiguous dual representations.

\begin{assumption}[Unambiguity]\label{ass:d_K_orthogonalNot}
For $\lambda <1$, 
$2\sum_{s=1}^{\lfloor 1/\lambda \rfloor} \delta_s K(\lambda s) \rightarrow \Omega_{\infty} \neq 0$.
\end{assumption}

Most commonly used kernels and difference sequences satisfy Assumption~\ref{ass:d_K_orthogonalNot}.
For example, if $K=K_{\Bart}$,  Assumption~\ref{ass:d_K_orthogonalNot} is satisfied for all
$0\leq m\leq 10$, all $h$, and all $\{d_j\}$ in Example~\ref{eg:dj}.
Indeed, $\Omega_{\infty}$ measures the limiting gap between 1 and $K_{\diff}(0)$, 
i.e., $K_{\diff}(0)-1 \rightarrow \Omega_{\infty}$ as $n\rightarrow\infty$;
see Figure~\ref{fig:dTAVC_effectiveKernel}. 

However, it is possible that Assumption~\ref{ass:d_K_orthogonalNot} is not satisfied when 
$m\rightarrow \infty$.
It may happen when the difference sequence is approximately ``uncorrelated'', 
i.e., $\delta_s \approx 0$ for $s\neq 0$.
We formalize this situation by the following assumption.

\begin{assumption}[Approximately uncorrelated differencing]\label{ass:d_uncorrelatedDiff}
There are $c',c''>0$ such that $-c''/m\leq\delta_s\leq -c'/m$ for all $s=\pm1, \ldots, \pm m$.
\end{assumption}

Note that Assumption~\ref{ass:d_uncorrelatedDiff} is satisfied by the 
local difference sequence and \cite{hall1990}'s difference sequence (see Example~\ref{eg:dj}). 
We will show in Section~\ref{sec:BWsel} that 
the first type of sequence is nearly optimal whereas the second type is asymptotically optimal.

\section{Invariance to time-varying means}\label{sec:invariance}
\subsection{Strength of robustness}
When the true mean function is $\mu(\cdot)$, 
we denote the bias, variance and MSE of $\widehat{v}$ by 
$\Bias_{\mu}(\widehat{v} ; v)$, $\Var_{\mu}(\widehat{v})$ and $\MSE_{\mu}(\widehat{v}; v)$,
respectively. 
If $\mu(t)\equiv 0$, we write them as 
$\Bias_{0}(\widehat{v}; v)$, $\Var_{0}(\widehat{v})$ and $\MSE_{0}(\widehat{v}; v)$
to emphasize that it is the ideal situation as $X_i = Z_i$ for all $i$. 
If the estimand is $v$, we may omit the argument ``$v$'' in the bias and MSE.

\begin{definition}[Robustness against mean functions]\label{def:rob}
Let $\mathcal{M}$ be a family of mean functions. 
An estimator $\widehat{\theta}$ of $\theta$ is said to be 
\emph{strictly robust in $\mathcal{M}$} if 
			$\MSE_{\mu}(\widehat{\theta};\theta) \sim \MSE_{0}(\widehat{\theta};\theta)$ for all $\mu\in\mathcal{M}$; and
\emph{loosely robust in $\mathcal{M}$} if 
			$\MSE_{0}(\widehat{\theta};\theta)\rightarrow 0$ implies 
			$\MSE_{\mu}(\widehat{\theta};\theta)\rightarrow 0$ for all $\mu\in\mathcal{M}$.
\end{definition}

If $\widehat{\theta}$ is strictly robust, its first order $\mathcal{L}^2$ asymptotic properties
are the same  for any $\mu\in\mathcal{M}$. 
It is the most desirable.
If $\widehat{\theta}$ is loosely robust, its $\mathcal{L}^2$ consistency is maintained within $\mathcal{M}$
but the convergence rate and MSE may be different.

\begin{theorem}[Robustness]\label{thm:robustness}
Let $\kappa :=\int_{-1}^1 K(t) \, \dd t \neq 0$.
Suppose Assumptions~\ref{ass:weakDep}--\ref{ass:summability_dj} hold, 
$\ell/n\rightarrow 0$, and
$\mathcal{G} \gtrsim \ell+mh$.
Also let $\ell, h\in\mathbb{N}$ and $m\in\mathbb{N}_0$, which are possibly divergent.
Then
\begin{align*}
	\Bias_{\mu}(\widehat{v}; v)
		&=  \Bias_{0}(\widehat{v}; v)
			+ \left\{ \kappa\ell\mathcal{V} + O\left( \frac{\ell}{n}\right) \right\} \mathbb{1}_{(m=0)}
			+ R_{\bias} ,  \\
	\sqrt{ \Var_{\mu}(\widehat{v}) }
		&= \sqrt{ \Var_{0}(\widehat{v}) }
			+O\left\{ \frac{\ell\left( \mathcal{C} +\mathcal{S} \mathcal{J}  \right)}{\sqrt{n}} \right\} \mathbb{1}_{(m=0)}
			+ R_{\se} ,
\end{align*}
where 
\begin{align*}
	R_{\bias}
		&= O\left[ \frac{\ell}{n} \left\{  (\ell\mathbb{1}_{m=0} +mh ) \mathcal{S}^2\mathcal{J}
			+  \frac{(\ell\mathbb{1}_{m=0} +mh )^2\mathcal{C}^2}{n}\right\} \right], \\
	R_{\se} 
		&= O\left[ \frac{\ell}{\sqrt{n}} \left\{ \frac{mh\mathcal{C}}{n} 
			+ \mathcal{S} \left(\frac{mh\mathcal{J}}{n}\right)^{1/2} \right\}\right]. 
\end{align*}
\end{theorem}

The proof
can be found in Section~\ref{sec:proof_robustness} of the supplement.
Theorem~\ref{thm:robustness} states that the bias and variance of $\widehat{v}$
are governed by
(i) the performance of $\widehat{v}$ when $\mu \equiv 0$, 
(ii) the order $m$, and 
(iii) the mean function $\mu(\cdot)$. 

Factor (i) is the idealistic performance.
The estimator $\widehat{v}$ is good if 
$\Bias_{\mu}(\widehat{v}) \sim \Bias_{0}(\widehat{v})$ and 
$\Var_{\mu}(\widehat{v}) \sim \Var_{0}(\widehat{v})$. 
Factor (ii) has the greatest impact on $\widehat{v}$. 
When $m=0$,
the estimator reduces to the classical estimators in Example~\ref{eg:constantMean}. 
It has a divergent bias if the mean is a fixed non-constant function. 
If $\mathcal{V}=o(1/\ell)$, $\widehat{v}$ is still consistent. 
However, in this case, the variability of $\mu(\cdot)$ is diminishing  as $n\rightarrow\infty$. 
This robustness is insufficient for most real applications.  
Similar results have been documented in, e.g., \cite{GoncalvesWhite2002}. 
Factor (iii) depends on $\mu(\cdot)$ only through $\mathcal{C}$, 
$\mathcal{S}$ and $\mathcal{J}$.
When $m>0$, $\mu(\cdot)$ only affects $R_{\bias}$ and $R_{\se}$, 
which are typically negligible; see Section~\ref{sec:wellMean}.  

Besides, Theorem~\ref{thm:robustness} is also applicable to 
the estimators in Examples~\ref{eg:uncorrelated}--\ref{eg:general_mInf}. 
We compare them in Example~\ref{eg:compare_robustness} below. 

\begin{example}[Comparison]\label{eg:compare_robustness}
In the presence of time-varying mean and autocorrelation, 
our proposed estimator $\widehat{v}$ and 
the existing estimators in Examples~\ref{eg:general_m1}--\ref{eg:insuff_diff} can be used for estimating $v$. 
We compare their robustness as follows. 
\begin{itemize}
	\item Our framework in (\ref{eqt:diffvarEst}) covers the proposed $\widehat{v}$ and 
		the estimators in Examples~\ref{eg:general_m1}--\ref{eg:general_mInf}.
		Although they use different values of $m>0$ and $h$, 
		all of them satisfy 
		\[
			\MSE_{\mu}(\widehat{v}) 
				= \left\{ \Bias_{0}(\widehat{v}) + R_{\bias} \right\}^2
					+ \left\{ \sqrt{ \Var_{0}(\widehat{v}) } + R_{\se} \right\}^2
		\]
		according to Theorem~\ref{thm:robustness} with 
		\[
			R_{\bias} = O\left[ \frac{\ell}{n} \left\{  H \mathcal{S}^2\mathcal{J}
							+  \frac{H^2\mathcal{C}^2}{n}\right\} \right]
			\quad \text{and} \quad
			R_{\se} = O\left\{ \frac{\ell}{\sqrt{n}} \left( \frac{H\mathcal{C}}{n} + \mathcal{S} \left(\frac{H\mathcal{J}}{n}\right)^{1/2} \right)\right\},
		\]
		where $H:=mh$. 
		Clearly, as long as $\ell$ and $H$ remain unchanged, 
		the orders of $R_{\bias}$ and $R_{\se}$ do not change with $h$ and $m$. 
		In other words, all these estimators are equally robust against the mean functions asymptotically. 
		It is worth mentioning that we further enhance the finite-sample robustness 
		of our proposed estimator in Section \ref{sec:rcp}. 
	\item Since the estimators in Examples~\ref{eg:1CPvEst}--\ref{eg:insuff_diff} do not fall into our framework, 
		we compare their robustness via simulation in Section~\ref{sec:sim_MSE}. 
		It indicates that our proposed $\widehat{v}$ is more robust against the mean functions
		than all competitors.  
	\hfill $\blacksquare$
\end{itemize}
\end{example}

\subsection{Class of well-behaved mean functions}\label{sec:wellMean}
We will prove in Sections~\ref{sec:consistency}--\ref{sec:optimal} that, 
under the assumption $u_q := \sum_{k\in\mathbb{Z}}|k|^q|\gamma_k|<\infty$,
the estimator $\widehat{v}$ is consistent and rate-optimal with $\MSE_0(\widehat{v})\asymp n^{-2q/(1+2q)}$ when 
$m\in\mathbb{N}$, $h/\ell=\lambda\in(0,\infty)$, and $\ell\asymp n^{1/(1+2q)}$; see (\ref{eqt:optMSE}).
The same optimal MSE is also achieved by the standard estimators; see, e.g., \citet{andrews1991}.
Using the baseline $\MSE_0(\widehat{v})\asymp n^{-2q/(1+2q)}$, 
Theorem~\ref{thm:robustness}
shows that $\widehat{v}$ is strictly robust in 
\begin{align}
	\mathcal{M}_{q} 
		&:= \{\mu(\cdot) : R_{\bias}^2 + R_{\se}^2 = o(\MSE_0(\widehat{v}))\} \nonumber \\
		&= \left\{ \mu(\cdot) :\mathcal{C}^2 = o\left( n^{\frac{3q-1}{1+2q}}\right), 
				\mathcal{S}^2\mathcal{J} = o\left( n^{\frac{q-1}{1+2q}}\right), 
				\mathcal{G} \gtrsim n^{\frac{1}{1+2q}}
			\right\}, \label{eqt:Mq_setMeanFun}
\end{align}
which covers a large class of mean functions. 
A larger class can similarly be derived 
if we only require $\widehat{v}$ to be loosely robust.
Example~\ref{eg:rob} discuss a special case when $q=2$.

\begin{example}[Strict robustness of $\widehat{v}$ with $q=2$]\label{eg:rob}
Suppose $u_2 <\infty$. 
The optimal MSE satisfies $\MSE_0(\widehat{v})=O(n^{-4/5})$.
Then $\mathcal{M}_{2} = \{ \mu(\cdot): \mathcal{C}^2 = o(n), \mathcal{S}^2\mathcal{J} = o(n^{1/5}), 
\mathcal{G} \gtrsim n^{1/5}\}$, 
which includes 
(i) all Lipschitz continuous functions with $o(n^{1/2})$ Lipschitz constants, 
(ii) all step functions with $o(n^{1/5})$ number of finite-jump discontinuities that are separated by 
at least $O(n^{1/5})$, and 
(iii) a sum of (i) and (ii).
For example, all mean functions in Figure~\ref{fig:dTAVC_I006_TSplot} are members of $\mathcal{M}_{2}$.
\hfill $\blacksquare$
\end{example}

\section{Classes of consistent and rate-optimal estimators}\label{sec:consistency}
It is unclear whether $\widehat{v}$ is consistent for $v=\sum_{k\in\mathbb{Z}}\gamma_k$ even 
under the constant mean assumption because the ACVF of $\{D_i\}$ is not equal to $\gamma_k$:
\begin{align}\label{eqt:ACVF_D}
	\gamma_k^D := \Cov(D_i, D_{i+k}) = \sum_{j,j'=0}^m d_j d_{j'} \gamma_{h(j-j')+k}
		= \sum_{|s|\leq m} \delta_{s} \gamma_{hs+k}
		\neq \gamma_k .
\end{align}

In this section, we study the conditions for consistency and rate optimality for $\widehat{v}$
when the mean function $\mu(\cdot)$ is mildly non-constant in the sense that 
$\mathcal{J},\mathcal{S},\mathcal{C}\asymp 1$.
Our asymptotic theory requires the following regularity conditions on $K$.

\begin{assumption}[Near-origin property]\label{ass:Kernel_q}
The kernel $K$ satisfies 
that there exist $q\in\mathbb{N}$ and $B\in\mathbb{R}\setminus\{0\}$ such that 
$\{K(t)-K(0)\}/|t|^q \rightarrow B$ as $t\downarrow 0$.
\end{assumption}

\begin{assumption}[Near-boundary property]\label{ass:Kernel_q_tail}
The kernel $K$ satisfies that
there exist $q'\in\mathbb{N}$ and $B'\in\mathbb{R}\setminus\{0\}$ such that 
$\{K(1)-K(1-t)\}/|t|^{q'} \rightarrow B'$ as $t\downarrow 0$.
\end{assumption}

Assumption~\ref{ass:Kernel_q} is standard. 
The index $q$ is called 
the characteristic exponent (CE) of $K(\cdot)$; see \cite{Parzen1957}. 
We say that a kernel $K$ is of order $q$ if Assumption~\ref{ass:Kernel_q} is satisfied. 
The larger the value of $q$, the flatter the kernel is around $0$.
It governs the order of bias of $\widehat{v}$ in the stationary case.
In particular, for the kernel estimator 
with a $q$th order kernel in Example~\ref{eg:constantMean}, 
if $u_q=\sum_{k\in\mathbb{Z}}|k|^q|\gamma_k|<\infty$, 
the best possible bias is $O(1/\ell^q)$,
and the resulting optimal MSE is $O(n^{-2q/(1+2q)})$; 
see, e.g., \citet{andrews1991}.
Therefore, we say that $\widehat{v}$ is rate optimal if 
its MSE attains $O(n^{-2q/(1+2q)})$ for all time series 
that satisfy $u_q<\infty$.
Some commonly used kernels are shown in Table~\ref{tab:kernel}.
We suggest to use \cite{Parzen1957}'s kernel 
$K_q(t) = (1-|t|^q)^+$
as a convenient choice as it satisfies Assumption~\ref{ass:Kernel_q} with any specified $q\in\mathbb{N}$.

Assumption~\ref{ass:Kernel_q_tail} is non-standard. 
It states the flatness of $K(t)$ when $t\uparrow 1$;
see Table~\ref{tab:kernel} for a summary of kernels that satisfy Assumption~\ref{ass:Kernel_q_tail}. 
Although Assumption~\ref{ass:Kernel_q_tail} does not affect   
the convergence rate of the classical estimators in Example~\ref{eg:constantMean}, 
it plays an important role for difference-based estimators when $h/\ell \rightarrow 1$.

\begin{table}
\begin{center}
\scriptsize
\begin{tabular}{|l|l|l|l|}
\hline
Kernel & Definition & Assumption~\ref{ass:Kernel_q} & Assumption~\ref{ass:Kernel_q_tail} \\ \hline\hline
Bartlett  \citep{newwy_west_1987} & $K(t)=(1-|t|)^+$ & $q=1$ & $q'=1$ \\
Tukey--Hanning  \citep{andrews1991} & $K(t)=\{1+\cos(\pi t)\}\mathbb{1}(|t|\leq 1)/2$ & $q=2$ & $q'=2$ \\
Parzen  \citep{gallant1987} & 
				$K(t) = \left\{\begin{array}{ll}1-6t^2+6|t|^3, & |t|\leq 1/2;\\
				 2\{(1-|t|)^+\}^3, & |t|>1/2 .\end{array}\right.$ & $q=2$ & $q'=3$ \\
$q$th order polynomial \citep{Parzen1957} & $K(t) = (1-|t|^q)^+$ & $q\in\mathbb{N}$ & $q'=1$ \\ 
Lugsail \citep{VatsFlegal2018} & $K(t) = \{ K_0(t) - cK_0(rt) \}/(1-c)$ 
			& Same as $K_0$ & Depends on $K_0$ \\
Trapezoidal  \citep{politis1995} & 
		$K(t) = \left\{\begin{array}{ll}1, & |t|\leq c';\\
				 (1-|t|)^+/(1-c'), & |t|>c' .\end{array}\right.$
& Not satisfied & $q'=1$ \\
Truncated  \citep{white1984} & $K(t) = \mathbb{1}(|t|< 1)$ & Not satisfied & Not satisfied \\
Modified $q$th order polynomial  & Equation (\ref{eqt:polyKernel_modified}) & $q\in\mathbb{N}$ & $q'=q$ \\ \hline
\end{tabular}
\end{center}
\vspace{-0.2cm}
\caption{Some commonly used kernels.
The last two columns indicate the values of $q$ and $q'$ so that 
Assumptions~\ref{ass:Kernel_q} and \ref{ass:Kernel_q_tail} are satisfied, respectively. 
In lugsail kernel, $r\geq 1$, $c\in[0,1)$ and $K_0$ is any initial kernel. 
In trapezoidal kernel, $c'\in(0,1]$.
Trapezoidal and truncated kernels do not satisfy Assumption~\ref{ass:Kernel_q} because 
$B=0$ for any $q\in\mathbb{N}$.
}\label{tab:kernel}
\vspace{-0.4cm} 
\end{table}

\subsection{Fixed-$m$ difference-based estimators}\label{sec:fixedmConsis}
Theorem~\ref{thm:consistency_finite} below studies 
the consistency and rate-optimality of $\widehat{v}$ when $0<m<\infty$ 
in different regimes according to the limiting value of $h/\ell$; see (\ref{eqt:def_m_l_h})
for the definitions of $m$, $\ell$ and $h$.

\begin{theorem}[Finite-$m$ regime]\label{thm:consistency_finite}
Suppose that $\mu(\cdot)$ satisfies (\ref{eqt:muDef}) with 
$\mathcal{J},\mathcal{S},\mathcal{C}\asymp 1$; and 
$\{Z_i\}_{i\in\mathbb{Z}}$ satisfies Assumption~\ref{ass:weakDep}
and $u_q=\sum_{k\in\mathbb{Z}}|k|^q|\gamma_k|<\infty$ for some $q\in\mathbb{N}$.
Let $\ell$ be an unknown-free bandwidth satisfying $1/\ell + (\ell+mh)/n \rightarrow 0$.
Suppose $0<m<\infty$ is fixed, and Assumption~\ref{ass:Kernel_q} holds. 
Under the least favorable data generating mechanism, we have the following results.
\begin{enumerate}
	\item \label{item:finite_m_1} If $h/\ell \rightarrow 0$, 
			then $\widehat{v}$ is inconsistent in $\mathcal{L}^2$. 
	\item \label{item:finite_m_2} If $h/\ell \rightarrow \lambda_{\infty} \in(0,1)$, 
			then, under Assumption~\ref{ass:d_K_orthogonalNot},
			$\widehat{v}$ is inconsistent in $\mathcal{L}^2$.
	\item \label{item:finite_m_3} If $h/\ell \rightarrow 1$, 
			then, under Assumption~\ref{ass:Kernel_q_tail} and $|h-\ell|=O(1)$, 
			the best possible MSE is
			$\MSE_{\mu}(\widehat{v}) \asymp n^{-2(q\wedge q')/\{1+2(q\wedge q')\}}$,
			which is achieved by $\ell\asymp n^{1/\{1+2(q\wedge q')\}}$.
	\item \label{item:finite_m_4} If $h/\ell \rightarrow \lambda_{\infty} \in(1, \infty)$, then 
			the best possible MSE is
			$\MSE_{\mu}(\widehat{v}) \asymp n^{-2q/(1+2q)}$, 
			which is achieved by $\ell \asymp n^{1/(1+2q)}$. 
	\item \label{item:finite_m_5} Suppose $h/\ell \rightarrow \infty$. 
			\begin{enumerate}
				\item If $q=1$, then $\widehat{v}$ is rate suboptimal in $\mathcal{L}^2$. 
				\item If $q>1$, then 
					the best possible MSE is
					$\MSE_{\mu}(\widehat{v}) \asymp n^{-2q/(1+2q)}$, 
					which is achieved by $\ell\asymp n^{1/(1+2q)}$ and $n^{1/(1+2q)} \ll h \lesssim n^{q/(1+2q)}$.
			\end{enumerate}
\end{enumerate}
\end{theorem}

The proof can be found in Section~\ref{sec:proof_consistency_finite} of the supplement.
From Theorem~\ref{thm:consistency_finite} (\ref{item:finite_m_1})--(\ref{item:finite_m_2}),
$\widehat{v}$ with $h/\ell\rightarrow\lambda_{\infty}\in[0,1)$
should never be used as it is guaranteed to be inconsistent for $v$. 
Theorem~\ref{thm:consistency_finite} (\ref{item:finite_m_3})--(\ref{item:finite_m_5}) state the fastest possible convergence rate of $\widehat{v}$. 
Under Assumption~\ref{ass:Kernel_q}, 
the optimal MSE in the stationary case is
$O\{n^{-2q/(1+2q)}\}$; see \cite{andrews1991}.
In case (\ref{item:finite_m_5}),  
the rate optimality cannot be achieved for handling time series that satisfies $u_q<\infty$ with $q=1$ only.
In case (\ref{item:finite_m_3}), 
the rate optimality cannot be achieved by 
all $q$th order kernels unless they satisfy Assumption~\ref{ass:Kernel_q_tail}
with $q'\geq q$, 
which means that $K(t)$ is flatter or equally flat near the boundary $t\uparrow 1$ than near the origin $t\downarrow 0$.
The condition $|h-\ell| = O(1)$ means that $h/\ell\rightarrow 1$ sufficiently quickly.
The requirement $q'\geq q$ is not satisfied by all kernels; see Table~\ref{tab:kernel}.  
For example, \cite{Parzen1957}'s kernel 
$K_q(t) = (1-|t|^q)^+$ satisfies Assumption~\ref{ass:Kernel_q} for any $q\in\mathbb{N}$, 
but it only satisfies Assumption~\ref{ass:Kernel_q_tail} with $q'=1$.
Hence, the rate-optimality cannot be achieved when $q>1$.
One may design a smooth and differentiable kernel with $q'=q$ as follows:
\begin{eqnarray}\label{eqt:polyKernel_modified}
	\widetilde{K}(t) = \left\{ \begin{array}{ll}
			1-|t|^q+a|t|^{q+1}+b|t|^{q+2}	& \text{if $|t|\leq 1/2$};\\
			(1-|t|)^q - a(1-|t|)^{q+1} - b(1-|t|)^{q+2}	& \text{if $1/2< |t|\leq 1$}; \\
			0	& \text{if $|t|>1$},\\
			\end{array}\right.
\end{eqnarray}
where $a = 4-(q+1)2^q$ and $b = q2^{q+1}-4$.	
In this case, $\MSE(\widehat{v}) = O(n^{-2q/(1+2q)})$ if $\ell\asymp n^{1/(1+2q)}$. 
In case (\ref{item:finite_m_4}),  
it is more well-behaved as
$\widehat{v}$ is rate-optimal 
for all $K(\cdot)$ without any additional assumption. 
We remark that users always know whether $\widehat{v}$ is consistent or not 
as the values of $m$, $\ell$ and $h$ are specified by users.
In practice, we suggest to select $h=\lambda_{\infty}\ell$ whenever it is possible 
so that $h/\ell$ equals to $\lambda_{\infty}$ not only in the limit but also in finite samples.

\subsection{Divergent-$m$ difference-based estimator}
This section investigates the convergence properties of $\widehat{v}$ with
$m=m_n\rightarrow\infty$ as $n\rightarrow\infty$.

\begin{theorem}[Divergent-$m$ regime]\label{thm:consistency_infinite}
Assume all conditions in Theorem~\ref{thm:consistency_finite} 
except that $0<m<\infty$ is replaced by 
$m=m_n\rightarrow\infty$.
In addition, suppose Assumption~\ref{ass:summability_dj} holds. 
Under the least favorable data generating mechanism, we have the following results.
\begin{enumerate}
	\item \label{item:infinite_m_1} Suppose $h/\ell \rightarrow 0$.
			\begin{enumerate}
				\item If Assumption~\ref{ass:d_K_orthogonalNot} holds, then $\widehat{v}$ is inconsistent in $\mathcal{L}^2$.
				\item If Assumption~\ref{ass:d_uncorrelatedDiff} holds
						and $K$ is decreasing on $[0,1]$,  
						then (i) $\widehat{v}$ is inconsistent in $\mathcal{L}^2$ for $\ell/h \gtrsim m$,
						and (ii) $\widehat{v}$ is rate suboptimal in $\mathcal{L}^2$ for $\ell/h\ll m$. 
			\end{enumerate}
	\item \label{item:infinite_m_2} Suppose $h/\ell \rightarrow \lambda_{\infty} \in(0,1)$.
			\begin{enumerate}
				\item If Assumption~\ref{ass:d_K_orthogonalNot} holds, then $\widehat{v}$ is inconsistent in $\mathcal{L}^2$.
				\item If Assumption~\ref{ass:d_uncorrelatedDiff} holds
						and $K$ is decreasing on $[0,1]$,  
						then $\widehat{v}$ is rate suboptimal in $\mathcal{L}^2$. 
			\end{enumerate}
	\item \label{item:infinite_m_3} Suppose $h/\ell \rightarrow 1$.  
			Suppose further that $|h-\ell|=O(1)$, 
			and Assumptions~\ref{ass:d_uncorrelatedDiff} and \ref{ass:Kernel_q_tail} hold. 
			\begin{enumerate}
				\item If $q =1$ or $\ell^q \gg m \ell^{q'}$, 
						then $\widehat{v}$ is rate suboptimal in $\mathcal{L}^2$. 
				\item If $q>1$ and $\ell^q \lesssim m \ell^{q'}$, then 
						the best possible MSE is 
						$\MSE_{\mu}(\widehat{v}) \asymp n^{-2q/(1+2q)}$, 
						which is achieved by $\ell \asymp n^{1/(1+2q)}$ and 
						any $m\rightarrow\infty$ such that
						$n^{(q-q')/(1+2q)} \lesssim m \lesssim n^{(q-1)/(1+2q)}$.
			\end{enumerate}
	\item \label{item:infinite_m_4} Suppose $h/\ell \rightarrow \lambda_{\infty} \in(1, \infty)$
			and Assumption~\ref{ass:d_uncorrelatedDiff} holds.
			\begin{enumerate}
				\item If $q =1$, then $\widehat{v}$ is rate suboptimal in $\mathcal{L}^2$. 
				\item If $q >1$, then the best possible MSE is $\MSE_{\mu}(\widehat{v}) \asymp n^{-2q/(1+2q)}$, 
						which is achieved by $\ell \asymp n^{1/(1+2q)}$ and 
						$1 \ll m \lesssim n^{(q-1)/(1+2q)}$.
			\end{enumerate}
	\item \label{item:infinite_m_5} Suppose $h/\ell \rightarrow \infty$. 
			\begin{enumerate}
				\item If $q=1$, then $\widehat{v}$ is rate suboptimal in $\mathcal{L}^2$. 
				\item If $q>1$, then the best possible MSE is $\MSE_{\mu}(\widehat{v}) \asymp n^{-2q/(1+2q)}$, 
					which is achieved by 
					$h\gg \ell \asymp n^{1/(1+2q)}$ and $n^{1/(1+2q)} \ll mh \lesssim n^{q/(1+2q)}$.
			\end{enumerate}
\end{enumerate}
\end{theorem}

The proof can be found in Section~\ref{sec:proof_consistency_infinite} of the supplement.
Theorem~\ref{thm:consistency_infinite} implies that  
$\widehat{v}$ with a divergent $m$ 
is inconsistent or suboptimal in Cases~\ref{item:infinite_m_1}--\ref{item:infinite_m_2}. 
Although $\widehat{v}$ is consistent in Cases~\ref{item:infinite_m_3}--\ref{item:infinite_m_5},
the rate optimality cannot be achieved for handling time series that satisfies $u_q<\infty$ with $q=1$ only.
It is worth emphasizing that 
the variance estimators in Example~\ref{eg:general_mInf} utilize a
local centering technique with $h=1$ and $m,\ell\rightarrow \infty$. 
Since they fall in the regime $m\rightarrow\infty$ and $h/\ell \rightarrow 0$, 
the resulting estimators are inadmissible. 
Theorems~\ref{thm:consistency_finite} and \ref{thm:consistency_infinite} 
are summarized in  Table~\ref{tab:summary}.

\begin{table}
\begin{center}
\begin{tabular}{|l|l|l|}
\hline
Regimes	& 												$0<m<\infty$ & 		$m\rightarrow\infty$ \\ \hline\hline
$h/\ell\rightarrow 0$ & 								Inconsistent & 		Inconsistent or suboptimal*\\
$h/\ell\rightarrow \lambda_{\infty} \in (0,1)$ & 		Inconsistent* & 	Inconsistent or suboptimal* \\ 
$h/\ell\rightarrow 1$ & 								May be optimal* & 	Suboptimal* \\
$h/\ell\rightarrow \lambda_{\infty}\in(1, \infty)$ & 	Optimal   & 		Suboptimal* \\
$h/\ell\rightarrow \infty$ & 							Suboptimal & 		Suboptimal \\ \hline
\end{tabular}
\end{center}
\vspace{-0.2cm}
\caption{
Convergence properties of $\widehat{v}$
when $\mu$ satisfies (\ref{eqt:muDef}) with 
$\mathcal{J},\mathcal{S},\mathcal{C}\asymp 1$, and 
$u_q=\sum_{k\in\mathbb{Z}}|k|^q|\gamma_k|<\infty$.
``Inconsistent'' means that $\MSE_{\mu}(\widehat{v}) \nrightarrow 0$ for some time series. 
``Suboptimal'' means that $\MSE_{\mu}(\widehat{v})\rightarrow 0$ at a suboptimal rate for some $q\in\mathbb{N}$. 
``Optimal'' means that $\MSE_{\mu}(\widehat{v})\rightarrow 0$ at the optimal rate for all $q\in\mathbb{N}$.
``May be optimal'' means that $\MSE_{\mu}(\widehat{v})\rightarrow 0$ at the optimal rate for all $q\in\mathbb{N}$
iff the kernel satisfies $q'\geq q$; see Assumption~\ref{ass:Kernel_q_tail}. 
An asterisk ``*'' means that additional regularity conditions on $K$, $\{d_j\}$, $h$ or $\ell$ are needed.
}\label{tab:summary}
\vspace{-0.4cm} 
\end{table}

\section{Asymptotic optimality of variance estimators}\label{sec:optimal} 
\subsection{Mean squared error}
From Section~\ref{sec:consistency}, 
$\widehat{v}$ is always rate optimal 
iff $h/\ell \rightarrow \lambda_{\infty} \in[1,\infty)$ and $m<\infty$.
So, we study $\widehat{v}$ in this regime. 
From now on, we use a fixed $\lambda \equiv h/\ell =\lambda_{\infty} \in[1, \infty)$ and a fixed $m$ for simplicity.

\begin{theorem}[Bias]\label{thm:bias}
Let $\{X_i\}_{i\in\mathbb{Z}}$ be a stationary time series with $\{Z_i\}_{i\in\mathbb{Z}}$ satisfying Assumption~\ref{ass:weakDep}
and $u_q<\infty$ for some $q\in\mathbb{N}$.
Suppose
$1/\ell + \ell/n \rightarrow 0$, 
$h/\ell = \lambda \in[1,\infty)$, and 
$m<\infty$.
\begin{enumerate}
	\item Let $r_{\bias}= o(1/\ell^q)+ O(\ell/n)$. Then
			\begin{eqnarray}\label{eqt:bias}
				\Bias_0(\widehat{v})
					= \sum_{|k|\leq \ell} \left\{ K_{\diff}\left(\frac{k}{\ell}\right) -1 \right\} \gamma_k
						+ r_{\bias} .
			\end{eqnarray}
	\item \label{thm:bias-2} If, in addition, $\lambda\in[2,\infty)$, 
			and $K(\cdot)$ satisfies Assumption~\ref{ass:Kernel_q}, then 
			$\Bias(\widehat{v}) = {B v_q}/{\ell^q} + r_{\bias}$,
			where  $B$ is defined in Assumption~\ref{ass:Kernel_q}, 
			and $v_q = \sum_{k\in\mathbb{Z}}|k|^q \gamma_k$.
\end{enumerate}
\end{theorem}

\begin{theorem}[Variance]\label{thm:var}
Let $\{X_i\}_{i\in\mathbb{Z}}$ be a stationary time series with $\{Z_i\}_{i\in\mathbb{Z}}$ satisfying Assumption~\ref{ass:weakDep}.
Suppose
$1/\ell + \ell/n \rightarrow 0$, 
$h/\ell = \lambda \in[2,\infty)$, and 
$m<\infty$.
Let $A = \int_0^1 K^2(t)\, \dd t$, $\Delta_m = \sum_{|s|\leq m} \delta_s^2$, and 
$r_{\se}^2 = o(\ell/n)$.  
Then  
\begin{eqnarray*}
	\Var_0(\widehat{v}) = \frac{4A\ell v^2 \Delta_m}{n} + r_{\se}^2.
\end{eqnarray*}
\end{theorem}

The proofs 
of Theorems~\ref{thm:bias} and \ref{thm:var}
can be found in Sections~\ref{sec:proof_bias} and \ref{sec:proof_var} of the supplement, 
respectively.
Note that  
if $\ell = o\{n^{1/(1+q)}\}$, then $r_{\bias}^2 + r_{\se}^2 = o\{\MSE_0(\widehat{v})\}$.
In this case, Theorem~\ref{thm:bias} (\ref{thm:bias-2}) and Theorem~\ref{thm:var} imply that 
\[
	\MSE_0(\widehat{v}) \sim \frac{B^2 v_q^2}{\ell^{2q}} + \frac{4A\ell v^2 \Delta_m}{n} , 
\]
whose magnitude is controlled by the following factors. 
\begin{itemize}
	\item (Kernel $K$) 
			The constants $B^2$ and $A$ are determined by the user-specified kernel $K(\cdot)$.
			In practice, different users may use different kernels. 
			Our theory is flexible enough to support general kernels, 
			but the existing estimators in Example~\ref{eg:general_m1}
			only support the Bartlett kernel.
	\item (Serial dependence $\{\gamma_k\}$)
			The process-dependent constants $v$ and $v_q$ are functions of 
			$\{\gamma_k \}_{k\in\mathbb{Z}}$.
			They govern the magnitudes of variance and squared-bias, respectively. 
			Although users cannot control them,
			it is possible to select the best $\ell$ to adapt to the observed dependence structure of $X_{1:n}$; 
			see Section~\ref{sec:BWsel}.
	\item (Lag $h$) 
			The value of $h = \lambda \ell$
			has a great impact on the bias. 
			When $\lambda$ is small ($\lambda\in[1,2)$), 
			the bias (\ref{eqt:bias}) admits no simple form.
			When $\lambda$ is large enough ($\lambda\in[2,\infty)$), 
			the bias is asymptotically unaffected by the differencing operation
			because of  
			the matching property in 
			Proposition~\ref{prop:effKernel} (\ref{item:matching}). 
			We recommend $\lambda=2$ due to finite-sample consideration.
	\item (Difference sequence $\{d_j\}$)
			In the regime $\lambda\in[2,\infty)$,  
			the variance neatly depends on $m$ and $\{d_j\}_{j=0}^m$
			only through $\Delta_m$.
			In practice, we should pick the $\{d_j\}$
			to minimize the MSE; see Section~\ref{sec:BWsel}. 
	\item (Bandwidth $\ell$ and sample size $n$)
			The bandwidth $\ell$ affects the squared-bias and variance in an opposite direction. 
			A large $\ell$ leads to a large variance and a small squared-bias.
			So, the well-known bias-variance tradeoff occurs. 
			Section~\ref{sec:BWsel} discusses the selection of $\ell$.
	\item (Reminder term $r_{\bias}$)
			It has two parts: 
			$o(1/\ell^q)$ and $O(\ell/n)$.  
			The term $o(1/\ell^q)$ comes from approximating $\sum_{|k|\leq \ell} |k/\ell|^q \gamma_k$
			by $v_q/\ell^q$.
			The other one comes 
			from approximating $\widehat{v}$ by $\widehat{v}_{\diff}$; 
			see Proposition~\ref{prop:effKernel}. 
	\item (Reminder term $r_{\se}$)
			It appears when we apply the invariance principle 
			(Theorem 3 of \cite{wu2011}) to 
			approximate the moments of $\widehat{v}$ by the 
			moments of a Brownian motion. 
\end{itemize}

\subsection{Optimal parameters selection}\label{sec:BWsel}
From now on, we use $\lambda=2$ as it has the best properties. 
The optimal $\ell$ and $\{d_j\}$ are derived for each $m$ below. 

\begin{corollary}\label{coro:optim}
Suppose the conditions stated in Theorem~\ref{thm:bias} (\ref{thm:bias-2}) and Theorem~\ref{thm:var} hold.
In addition, assume $v_q\neq 0$, $\ell = o\{n^{1/(1+q)}\}$ and $\lambda= 2$.
\begin{enumerate}
	\item \label{coro:optim-1} The MSE-optimal $\ell$ is given by 
			\begin{eqnarray}\label{eqt:optEll}
				\ell^{\star} = \left\{ \frac{q (v_q/v)^2 B^2n}{2A\Delta_m} \right\}^{1/(1+2q)}.
			\end{eqnarray}
	\item \label{coro:optim-2} For each $m$, the optimal $\{d_j\}_{j=0}^m$, 
			denoted by $\{d_j^{\star}\}_{j=0}^m$, satisfies 
			$\delta_1 = \ldots = \delta_m = -1/(2m)$, 
			which implies $\Delta_m = 1+1/(2m)$.
			The solution $\{d_j^{\star}\}_{j=0}^m$ is a sequence of 
			universal constants that depends on $m$ only. 
\end{enumerate}
\end{corollary}

The proof
can be found in Section~\ref{sec:proof_optim} of the supplement.
From Corollary~\ref{coro:optim},
if the optimal $\ell^{\star}$ is used, the optimal MSE satisfies
\begin{align}\label{eqt:optMSE}
	n^{2q/(1+2q)} \MSE_0(\widehat{v})/ v^2
		\rightarrow (1+2q) \left\{ B^2 \left( \frac{2A\Delta_m}{q} \right)^{2q} (v_q/v)^2 \right\}^{1/(1+2q)} =: M_{(m)},
\end{align}
which can be generalized to $n^{2q/(1+2q)} \MSE_{\mu}(\widehat{v})/ v^2 \rightarrow M_{(m)}$ if 
$\mu \in \mathcal{M}_q$; see (\ref{eqt:Mq_setMeanFun}) for the definition of $\mathcal{M}_q$.
So, the $\mathcal{L}^2$ convergence rate of $\widehat{v}$ is $n^{q/(1+2q)}$, which is the best possible.

If, in addition, the optimal $\{d_j^{\star}\}_{j=0}^m$ is used, we have 
$M_{(1)}>M_{(2)}>\ldots$. 
Indeed, for any $\epsilon>0$, there is $m\in\mathbb{N}$ such that $|M_{(m)}-M_{(0)}|<\epsilon$, 
where $M_{(0)}$ is the best possible MSE 
achieved by the classical (non-robust) estimator $\widehat{v}_{(0)}$.  
Hence, the proposed framework covers the optimal estimator asymptotically. 
The optimal $\{d_j^{\star}\}_{j=0}^m$ can be obtained numerically by the innovation algorithm
(Definition 8.3.1 of \citet{brockwellDavis1991}); 
see Table~\ref{tab:optimalSeq} for the solution. 
It is worth mentioning that the $m$th order local differencing in Example \ref{eg:dj} is nearly optimal
in the sense that $\Delta_m = 1+(2m+1)/(3m^2+3m) \approx 1$ when $m$ is large.
It can be a good and convenient choice in practice.

For reference, we compare
our proposed estimator $\widehat{v}_{(3)}$ (i.e., using $m=3$) with 
the best existing robust estimator $\widehat{v}_{(\Chan)}$ in \cite{chan2020} (see Example~\ref{eg:insuff_diff}) 
and the best proposal $\widehat{v}_{(\Wu)}$ in \cite{wu_zhao_2007} (see Example~\ref{eg:general_m1}).
If $K=K_{\Bart}$, then 
$\widehat{v}_{(3)}$ uniformly dominates $\widehat{v}_{(\Chan)}$ and $\widehat{v}_{(\Wu)}$ in 
the sense that  
\begin{align}\label{eqt:compareMSEexisting}
	\frac{\MSE_0\{\widehat{v}_{(\Chan)}\}}{\MSE_0\{\widehat{v}_{(3)}\}} 
		\rightarrow \left(\frac{12}{7}\right)^{2/3} \approx 1.43
	\qquad \text{and} \qquad
	\frac{\MSE_0\{\widehat{v}_{(\Wu)}\}}{\MSE_0\{\widehat{v}_{(3)}\}} 
		\rightarrow \left(\frac{3}{2} \right)^{4/3} \approx 1.71 ;
\end{align}
see Section~\ref{sec:derivationOfCompareMSEexisting} of the supplement 
for a detailed derivation.

Since $\ell^{\star}$ depends on the unknowns $v_q$ and $v_0$,
we need also to estimate $v_q$.
Similar to (\ref{eqt:diffvarEst}),
we propose to estimate $v_p$ ($p\in\mathbb{N}_0$) by
\begin{equation*}
	\widehat{v}_p
		= \sum_{|k| < \ell} |k|^p K\left(\frac{k}{\ell}\right) \widehat{\gamma}_k^D.
\end{equation*}
We may write $\widehat{v}_p$ as $\widehat{v}_{p,(m)}$ to emphasize the order $m$.
The following corollaries show that $\widehat{v}_{p}$ is a consistent and robust estimator of $v_p$.

\begin{corollary}[Consistency]\label{coro:vpHat} 
Let $p\in\mathbb{N}_0$ and $\{X_i\}_{i\in\mathbb{Z}}$ be a stationary time series 
with $\{Z_i\}_{i\in\mathbb{Z}}$ satisfying Assumption~\ref{ass:weakDep}
and $u_{p+q}<\infty$ for some $q\in\mathbb{N}$.
Suppose
$K(\cdot)$ satisfies Assumption~\ref{ass:Kernel_q}, 
$1/\ell + \ell^{1+2p}/n \rightarrow 0$, 
$h/\ell = \lambda \in [2,\infty)$, and
$m<\infty$.
Then 
\begin{align*}
	\MSE_0(\widehat{v}_p; v_p)
		= \left( \frac{B v_{p+q}}{\ell^q} + r_{p,\bias} \right)^2
			+ \left(\frac{4A_p\ell^{1+2p} v^2 \Delta_m}{n} + r_{p,\se}^2\right),
\end{align*}
where $r_{p,\bias}=o(1/\ell^q)+ O(\ell^{1+p}/n)$;
$r_{p,\se}^2 = o(\ell^{1+2p}/n)$;
the constant $B$ is defined in Assumption~\ref{ass:Kernel_q}; and
$A_p = \int_0^1 |t|^{2p}K^2(t)\, \dd t$. 
\end{corollary}

\begin{corollary}[Robustness]\label{coro:robustness_general}
Assume the conditions in Theorem~\ref{thm:robustness}.
Let $p\in\mathbb{N}_0$, $R_{p,\bias}=O(\ell^p R_{\bias})$ and $R_{p,\se} = O(\ell^p R_{\se})$. 
Then
\begin{align*}
	\Bias_{\mu}(\widehat{v}_p; v_p)
		&=  \Bias_{0}(\widehat{v}_p; v_p)
			+ \left\{ \kappa\ell^{1+p}\mathcal{V} + O\left( \frac{\ell^{1+p}}{n}\right) \right\} \mathbb{1}_{(m=0)}
			+ R_{p,\bias} ,  \\
	\sqrt{ \Var_{\mu}(\widehat{v}_p) }
		&= \sqrt{ \Var_{0}(\widehat{v}_p) }
			+O\left\{ \frac{\ell^{1+p}\left( \mathcal{C} +\mathcal{S} \mathcal{J}  \right)}{\sqrt{n}} \right\} 
			\mathbb{1}_{(m=0)}
			+ R_{p,\se} ,
\end{align*}
\end{corollary}

The proofs of Corollaries~\ref{coro:vpHat} and \ref{coro:robustness_general}
can be found in Sections~\ref{sec:proof_vpHat} and \ref{sec:proof_robustness_general} of the supplement, 
respectively.
By Corollary~\ref{coro:vpHat}, we know that, in the constant mean case, 
$\widehat{v}_p$ converges in $\mathcal{L}^2$ optimally, i.e., 
$\MSE_0(\widehat{v}_p; v_p)  = O(n^{-2q/(1+2p+2q)})$, if $\ell \asymp n^{1/(1+2p+2q)}$.
By Corollary~\ref{coro:robustness_general}, 
if $m\in\mathbb{N}$, 
$h/\ell=\lambda\in(0,\infty)$, and 
$\ell\asymp n^{1/(1+2p+2q)}$ are used for $\widehat{v}_p$, then
$\widehat{v}_p$ is strictly robust in 
\begin{align*}
	\mathcal{M}_{p,q} 
		&:= \{\mu(\cdot) : R_{p,\bias}^2 + R_{p,\se}^2 = o(\MSE_0(\widehat{v}_{p};v_p))\}\\
		&= \left\{ \mu(\cdot) :\mathcal{C}^2 = o\left( n^{\frac{3p+3q-1}{1+2p+2q}}\right), 
				\mathcal{S}^2\mathcal{J} = o\left( n^{\frac{p+q-1}{1+2p+2q}}\right), 
				\mathcal{G} \gtrsim n^{\frac{1}{1+2p+2q}}
			\right\}.
\end{align*}
So, under $u_{p+q}<\infty$, we have $\MSE_{\mu}(\widehat{v}_p; v_p) \sim \MSE_{0}(\widehat{v}_p; v_p)$ for $\mu \in \mathcal{M}_{p,q}$.

\section{Implementation issue and generalization}\label{sec:imp_gen}
\subsection{Rough centering procedure}\label{sec:rcp}
If there are obvious jumps and trends,
one may \emph{roughly} remove them. 
It improves finite-sample performance. 
We only \emph{roughly} center $X_{1:n}$ because consistent centering
either distorts the autocovariance structure or deteriorates the convergence rate of $\widehat{v}$; 
see Theorem~\ref{thm:consistency_infinite}.
We propose a two-step \emph{rough centering procedure} (RCP).
The first and second steps remove obvious jumps and trends, respectively.

In step 1, we locate the $N$ most obvious CP times $t_1, \ldots , t_N$,
where $N\leq N'$ for some $N'<\infty$, e.g., $N'=10$.
We initialize $X_i^{(1)}=X_i$ for each $i$, and 
iterate the following steps for $k=1, 2, \ldots$.
Let $\xi^{(k)}_i = \sum_{i'=i}^{i+b-1}X_{i'}^{(k)}/b - \sum_{i'=i-b+1}^i X_{i'}^{(k)}/b$
be the local batch-mean difference at time $i$
for $i=b, \ldots, n-b+1$, where $b = \lfloor n^{1/3} \rfloor$.
An unusually large $\xi^{(k)}_i$ indicates that $i$ is a potential CP.
Denote the distance of $\xi_i^{(k)}$ from Tukey's fences by 
$O_i^{(k)} = \max\{
				0,
				\xi_i^{(k)}-4Q^{(k)}_1+3Q^{(k)}_3, 
				4Q^{(k)}_3-3Q^{(k)}_1-\xi_i^{(k)}\}$, where 
$Q^{(k)}_1$ and $Q^{(k)}_3$ are 
the lower and upper quartiles of $\{ \xi^{(k)}_i \}_{i=b}^{n-b+1}$, respectively. 
By Tukey's rule,
if 
\[
	I^{(k)} := \{i\in\{b,\ldots,n-b+1\}\setminus\{t_1, \ldots, t_{k-1}\}:O_i^{(k)}>0\} \neq \emptyset,
\]
we define the $k$th most obvious CP as $t_k = \argmax_{i\in I^{(k)}} O_i^{(k)}$, and  
let $X^{(k+1)}_i = X^{(k)}_i -  [ X_{t_k}^{(k)}-X_{t_k-1}^{(k)} ]_{-M}^{M}\mathbb{1}(i\geq t_k)$,
where $[\cdot]_{-M}^{M} = \max\{-M,\min(\cdot,M)\}$ 
and $M = M' \{\sum_{i=2}^n (X_i-X_{i-1})^2/(2n)\}^{1/2}$
for some $M'\in\mathbb{R}^+$, e.g., $M'=100$.
The iteration stops if $I^{(k)}=\emptyset$ or $k=N'$.
So, the jump-removed series is 
\begin{align}
	X_i^{\dag} 
		&= X_i - \sum_{\substack{j\in\{1,\ldots, N\} : t_j\leq i}}  
			[ X_{t_j}-X_{t_j-1} ]_{-M}^{M}\;.\label{eqt:removeJump}
\end{align}

In step 2, we run a segmented linear regression on 
$X_{t_j}^{\dag} , \ldots, X_{t_{j+1}-1}^{\dag}$ against  
$t_j, \ldots, t_{j}-1$ for each $j=0, \ldots, N$, where $t_0 = 1$ and $t_{N+1}=n+1$. 
After shifting the segmented regression lines to ensure continuity, we obtain 
\begin{align}
	X_{i}^{\ddag}
		&= X_i^{\dag} - \sum_{j=0}^{N} 
			\left\{ \widehat{\beta}_{j,0} + \widehat{\beta}_{j,1} (i-t_j) \right\}
			\mathbb{1}(t_{j}\leq i \leq t_{j+1}-1), \label{eqt:removeTrend}
\end{align}
where, for each $j=0, \ldots, N$,   
$\widehat{\beta}_{j,1} = \widehat{\alpha}_{j,1}$, 
$\widehat{\beta}_{j,0} = \sum_{j'=0}^{j-1} \widehat{\alpha}_{j',1}(t_{j'+1}-1-t_{j'})$, 
$(\widehat{\alpha}_{j,0},\widehat{\alpha}_{j,1})^{\intercal}
	= (\mathbb{Z}_j^{\intercal}\mathbb{Z}_j)^{-1}\mathbb{Z}_j^{\intercal}\mathbb{X}_j$; and 
$\mathbb{X}_j = (X_{t_j}^{\dag}, \ldots, X_{t_{j+1}-1}^{\dag})^{\intercal}$, 
$\mathbb{Z}_j$ is a $(t_{j+1}-t_j)\times 2$ matrix whose 
first column is a vector of $1$ and the second column is 
$[0, \ldots, t_{j+1}-t_j-1]^{\intercal}$.
Applying $\widehat{v}_p$ on the 
roughly centered time series $\{X_i^{\ddag}\}$, we obtain a finite-sample-adjusted estimator
$\widehat{v}^{\ddag}_p$.
The following corollary ensures that 
the RCP does not inflate the asymptotic MSE.

\begin{corollary}\label{coro:prewhite}
Assume the conditions in Corollary~\ref{coro:vpHat}. 
If $p\in\mathbb{N}_0$, $q\in\mathbb{N}$, $m\in\mathbb{N}$ and $\ell=O(n^{1/(1+2p+2q)})$,
then $\MSE_{\mu}(\widehat{v}^{\ddag}_p ; v_p) \sim \MSE_{\mu}(\widehat{v}_p ; v_p)$ 
for any $\mu\in\mathcal{M}_{p,q}$.
\end{corollary}

The proof
can be found in Section~\ref{sec:proof_prewhite} of the supplement.
For clarity, denote $\widehat{v}_p = \widehat{v}_p(X_{1:n}; \ell, K, m, d)$
if data $X_{1:n}$, bandwidth $\lceil \ell \rceil$, kernel $K(\cdot)$, order $m$, 
difference sequence $\{d_j\}_{j=0}^m$, 
and $\lambda=h/\ell =2$ are used.  
In practice, we always use the (universally) optimal difference sequence $\{d_j^{\star}\}$
in Table~\ref{tab:optimalSeq}.
The suggested estimator of the LRV $v$ is 
\begin{align}\label{eqt:suggestedEst}
	\widehat{v}^{\star} = \widehat{v}_0(X_{1:n}^{\ddag};\widehat{\ell}^{\star}, K, m, d^{\star}), 
	\quad \text{where} \quad
	\widehat{\ell}^{\star} = 
		\left\{ \frac{q (\widehat{v}_{q}^{\sharp}/\widehat{v}^{\sharp})^2 B^2n}{2A\Delta_m} \right\}^{1/(1+2q)},
\end{align}
where 
$\widehat{v}_{q}^{\sharp} = \widehat{v}_q(X_{1:n}^{\ddag}; 2n^{1/(5+2q)}, K_2, m, d^{\star})$ and 
$\widehat{v}^{\sharp} = \widehat{v}_0(X_{1:n}^{\ddag}; 2n^{1/5}, K_2, m, d^{\star})$
are pilot estimators of $v_q$ and $v$, respectively.
We recommend 
$m=3$ and \cite{Parzen1957}'s kernel $K_q(t) = (1-|t|^q)^+$ with $q=2$;
see Section~\ref{sec:experiment} for some simulation evidence.

\subsection{Multivariate time series}
We generalize (\ref{eqt:modelX}) to the multivariate setting  
such that $\{Z_i\}_{i\in\mathbb{Z}}$ is a sequence of 
$S$-dimensional zero-mean stationary noises with $\gamma_k = \E(Z_0 Z_k^{\intercal})$; and 
$\mu_i = \mu(i/n) \in\mathbb{R}^S$ are $S$-dimensional signals. 
We also use the decomposition in (\ref{eqt:muDef}) but $\xi_0, \ldots, \xi_{\mathcal{J}}\in\mathbb{R}^{S}$.
Denote the $s$th element of a vector $e$ by $e^{[s]}$, and 
the $(r,s)$th element of a matrix $E$ by $E^{[r,s]}$.

We assume $Z_i = g(\ldots,\varepsilon_{i-1},\varepsilon_i)$, where 
$\{\varepsilon_i\}_{i\in\mathbb{Z}}$
are i.i.d. $\mathbb{R}^{S'}$-dimensional innovations.
Generalize (\ref{eqt:depMea}) to 
$\theta_{p,i}^{[s]}:=\| Z_i^{[s]} - Z_{i,\{0\}}^{[s]} \|_p$   
and
$\Theta_p^{[s]} := \sum_{i=0}^{\infty} \theta_{p,i}^{[s]}$, 
for $s=1, \ldots, S$.
Assumption~\ref{ass:weakDep} is generalized as follows.

\begin{assumptionCustom}{\ref*{ass:weakDep}*}[Weak dependence]\label{ass:weakDepGen}
The $S$-dimensional zero-mean strictly stationary $\{Z_i\}$ 
satisfies $\E(Z_1^{[s]})^4<\infty$ and $\Theta_4^{[s]} <\infty$ for all $s=1, \ldots, S$. 
\end{assumptionCustom}

The quantity of interest is the long-run variance (-covariance matrix)
$v = \lim_{n\rightarrow\infty} n\E(\bar{Z}_n\bar{Z}_n^{\intercal}) = \sum_{k\in\mathbb{Z}} \gamma_k 
\in\mathbb{R}^{S\times S}$.
Our proposed estimator of $v$ admits the same form as (\ref{eqt:diffvarEst})
with $\widehat{\gamma}^D_k = \sum_{i=mh+|k|+1}^n D_i D_{i-|k|}^{\intercal}/n$.
We define $v_p = \sum_{k\in\mathbb{Z}} |k|^p \gamma_k$ and 
$u_p = \sum_{k\in\mathbb{Z}} |k|^p |\gamma_k|$, where $|\gamma_k|$ is the entry-wise absolute value of 
$\gamma_k$. 
Let $w$ be a $S\times S$ matrix whose $(r,s)$th element is 
\[
	w^{[r,s]} = (v^{[r,r]}v^{[s,s]}+v^{[r,s]}v^{[r,s]})/2, 
	\qquad r,s\in\{1, \ldots, S\}. 
\]

\begin{corollary}\label{coro:multi}
The results in Corollaries~\ref{coro:vpHat} and \ref{coro:robustness_general} remain valid 
if $(\widehat{v}_p , v_p , v_{p+q}, v^2)$ is replaced by 
$(\widehat{v}_p^{[r,s]} , v_p^{[r,s]} , v_{p+q}^{[r,s]}, w^{[r,s]})$ 
for each $r,s\in\{1, \ldots, S\}$, provided that 
Assumption~\ref{ass:weakDep} is replaced by Assumption~\ref{ass:weakDepGen}, and 
$u_{p+q}<\infty$ is replaced by $u_{p+q}^{[r,s]}<\infty$ for all $r,s$. 
\end{corollary}

The proof
can be found in Section~\ref{sec:proof_multi} of the supplement.

\section{Experiments, applications, and real-data examples}\label{sec:experiment}
We consider a non-linear time series model for all simulation studies in this section. 
Let $\{Z_i'\}_{i=1}^n$ be generated from a threshold autoregressive (TAR) model:
\begin{align}\label{eqt:def_timeseriesZp}
	Z_i' = \left\{ 
		\begin{array}{ll}
		\theta_1 Z'_{i-1} + \varepsilon_i & \text{if $Z'_{i-1}\geq0$;}\\
		\theta_2 Z'_{i-1} + \varepsilon_i & \text{if $Z'_{i-1}<0$,}
		\end{array}
		\right.
\end{align}
where $\theta_1, \theta_2$ are the AR parameters in regimes 1 and 2, respectively, 
and $\varepsilon_i$ follow $\Normal(0,1)$ independently. 
We use $\theta_1 \in \{0.1, \ldots, 0.9\}$ and $\theta_2 = 0.5$.
Let $Z_i = Z_i'/\sqrt{v}$, where $v=v_{\theta_1, \theta_2}$ is 
the LRV for the time series (\ref{eqt:def_timeseriesZp}). 
So, $\{Z_i\}$ is stationary and satisfies Assumption~\ref{ass:weakDep}; 
see \cite{wu2011}.

\subsection{Efficiency and robustness}\label{sec:sim_MSE}
Let $\mu(t) = \Xi\{e^{t}+ \mathbb{1}(t>0.3)+2\mathbb{1}(t>0.6)+4\mathbb{1}(t>0.8)\}$, 
where $\Xi\in \{0,1,\ldots,4\}$.  
When $\Xi\neq 0$, it contains three jumps and an exponentially increasing trend. 
We study the following estimators of $v$:
(i) $\widehat{v}_{(\Wu)}$ the best proposal in \cite{wu_zhao_2007} (see Example~\ref{eg:general_m1});
(ii)  $\widehat{v}_{(\DFW)}$ the OBM estimator with \cite{DehlingFriedWendler2020}'s adjustment (see Example~\ref{eg:1CPvEst});
(iii)  $\widehat{v}_{(\Chan)}$ the suggested estimator in \cite{chan2020} (see Example~\ref{eg:insuff_diff});
(iv) $\widehat{v}_{(0)}^{\star}$ a classical estimator (see Example~\ref{eg:constantMean}); and 
(v) $\widehat{v}_{(m)}^{\star}$ the proposed estimator with $m=1,2,3$. 
We use $K(t)= (1-t^2)^+$ in $\widehat{v}_{(0)}^{\star},\ldots, \widehat{v}_{(3)}^{\star}$.
We compare their (a) efficiency and (b) robustness via $10^4$ replications. 

For (a), the values of $\MSE_0(\cdot)$  
are reported for different values of $\theta_1$. 
For (b), the values of $\MSE_\mu(\cdot)$ are reported for different values of $\Xi$ when $\theta_1 = 0.4$. 
The results when $n=200$ are plotted in Figure~\ref{fig:dTAVC_A032_n200}. 
Our worst proposal $\widehat{v}_{(1)}^{\star}$ is already considerably more efficient than the 
existing estimators. 
The improvement of $\widehat{v}_{(2)}^{\star}$ over $\widehat{v}_{(1)}^{\star}$ is substantial. 
Although the improvement of $\widehat{v}_{(3)}^{\star}$ over $\widehat{v}_{(2)}^{\star}$ looks incremental, 
the advantage of $\widehat{v}_{(3)}^{\star}$ becomes obvious when $n$ increases;
see Figure~\ref{fig:dTAVC_A032_n400} of the supplement 
for the results when $n=400$. 
When $\Xi\neq0$, the existing estimators still show a certain degree of robustness 
relative to the non-robust $\widehat{v}_{(0)}^{\star}$. 
But their MSEs are substantially affected. 
All of our proposed estimators are of nearly constant risk for all $\Xi$.
It is remarked that the asymptotic relative efficiency does not improve significantly when $m\geq 4$ because 
\[
	\frac{\MSE(\widehat{v}_{(2)})}{\MSE(\widehat{v}_{(1)})} - 1 \approx -13.6\%, \qquad
	\frac{\MSE(\widehat{v}_{(3)})}{\MSE(\widehat{v}_{(2)})} - 1 \approx -5.4\%, \qquad
	\frac{\MSE(\widehat{v}_{(4)})}{\MSE(\widehat{v}_{(3)})} - 1 \approx -2.9\%.
\]
A similar finding is also documented in \citet{hall1990}.
Moreover, an excessively large $m$ may affect the robustness of $\widehat{v}_{(m)}$ in finite samples. 
Hence, in practice, we suggest using $m=3$.

For reference, we also compute the following estimators:
(vi) the estimators in \citet{wu_zhao_2007} that use the sum of absolute (SA) differences 
and the median of absolute (MA) differences (see Example~\ref{eg:Wu_median});
(vii) \cite{Altissimoa_Corradic_2003}'s estimator (see Example~\ref{eg:general_mInf}); 
(viii) \citet{Crainiceanu2007}'s estimator (see Example~\ref{eg:1CPvEst}); and 
(ix) \citet{JuhlXiao2009}'s estimator (see Example~\ref{eg:general_mInf}).
Their performances are obviously worse than other estimators either in terms of 
efficiency or robustness;
see Figure~\ref{fig:dTAVC_A032_n200_full} of the supplement.  
Additional simulation experiments that further investigate the robustness (in terms of $\mathcal{S}$ and $\mathcal{C}$) 
can be found in Section~\ref{sec:robustness_experiment_additional} of the supplement.

\begin{figure} 
\begin{center}
\includegraphics[width=\textwidth]{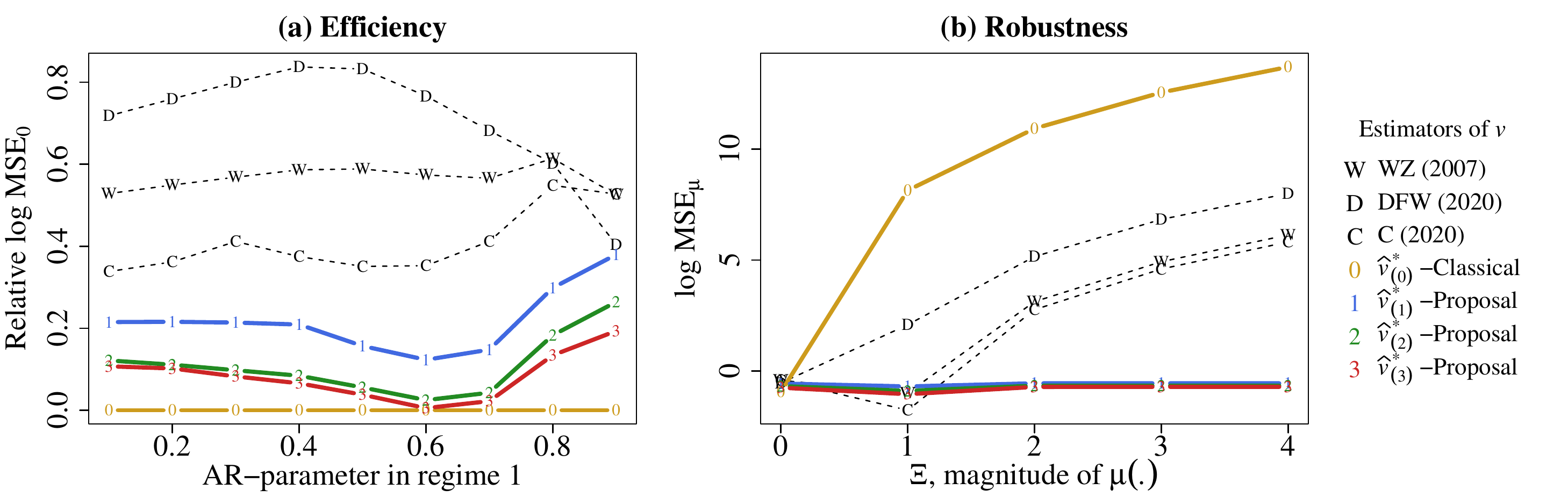} 
\end{center} 
\vspace{-0.6cm}
\caption{ 
		(a) The values of $\log\{\MSE_0(\cdot)/\MSE_0(\widehat{v}_{(0)}^{\star})\}$  
			when $\mu_1=\ldots = \mu_n=0$.
		(b) The values of $\log \MSE_\mu(\cdot)$ against $\Xi$ (i.e., the magnitude of the mean function).}
		\label{fig:dTAVC_A032_n200}
\vspace{-0.4cm} 
\end{figure}

\subsection{Hypothesis tests for structural breaks}
An estimator of $v$ is needed in many hypothesis testing problems.
We present two examples here:
(a) the KS CP test, and 
(b) a structural break test in the presence of trends \citep{wu_zhao_2007}.
The test statistic (a) is defined in (\ref{eqt:KS}).
The test statistic (b) is defined as 
\begin{align}\label{eqt:WZ2007test}
	T_n(v) = \frac{1}{k_n \sqrt{v}} \max_{k_n\leq i\leq n-k_n} \left\vert \sum_{j=i+1}^{k_n+i} X_j - \sum_{j=i-k_n+1}^i X_j \right\vert ,
\end{align}
where $k_n = n^{\beta}$ and $1/2<\beta< 2/3$.
The estimators (i)--(v) with $m=3$ described in Section~\ref{sec:sim_MSE} are used to estimate 
the $v$ in (\ref{eqt:KS}) and (\ref{eqt:WZ2007test}).

In reality, when a structural break occurs, 
the mean may suddenly jump to a high level  
but return to a lower level after that. 
So, we consider 
$\mu(t) = \Xi\{10\mathbb{1}(t>0.3)-9\mathbb{1}(t>0.35)\}$.
The mean jumps from 0 to $10\Xi$ at $t=0.3$ and drops to $\Xi$ at $t=0.35$.
Figure~\ref{fig:dTAVC_B020_case3} shows the powers and size-adjusted powers of the tests (a) and (b) 
when $n=200$ and the nominal size is 5\%.

\begin{figure} 
\begin{center}
\includegraphics[width=\textwidth]{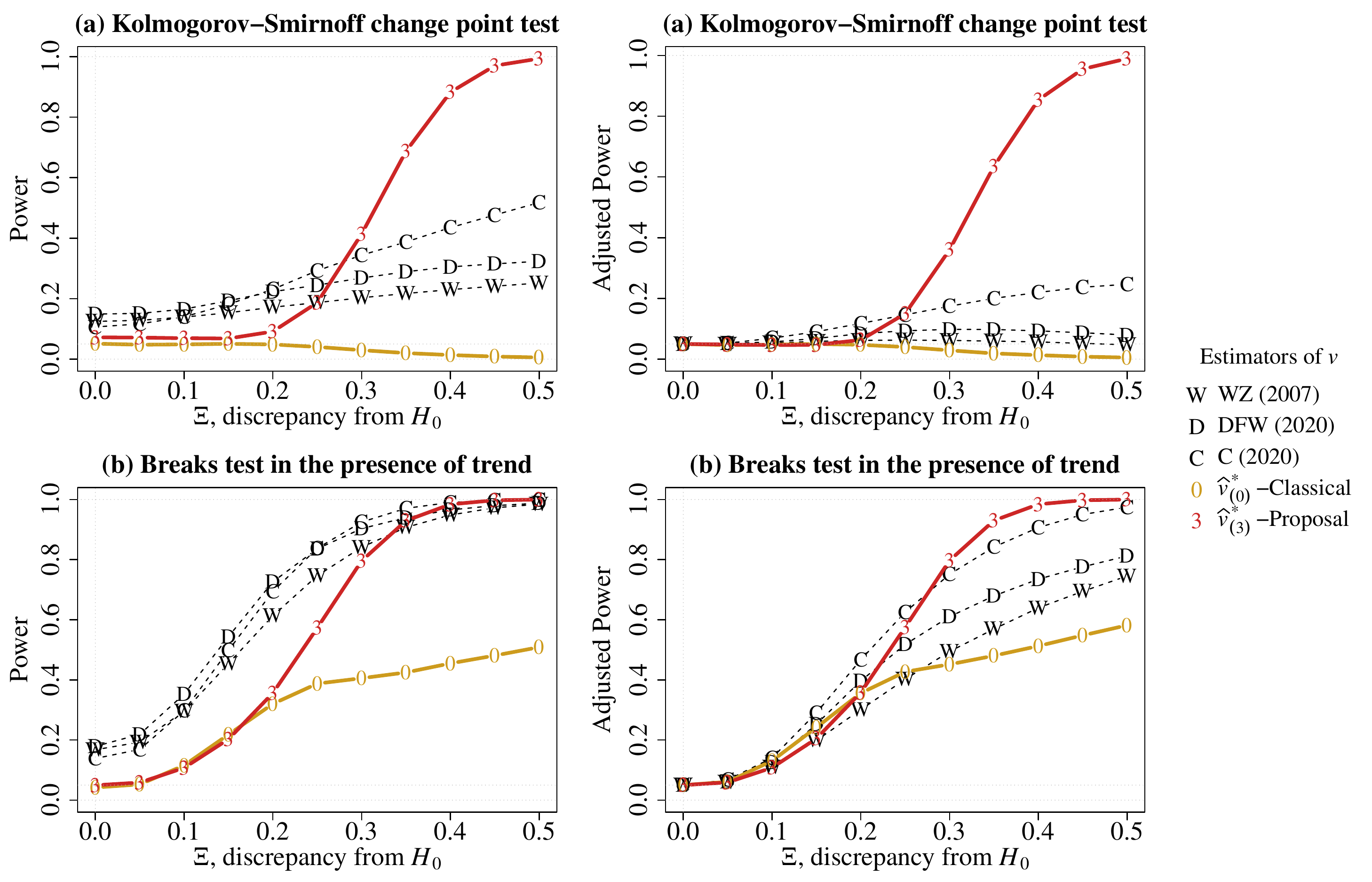} 
\end{center} 
\vspace{-0.6cm}
\caption{The powers and adjusted powers of different CP tests with various estimators of $v$.}
\label{fig:dTAVC_B020_case3}
\vspace{-0.4cm}
\end{figure}

First, the sizes (type-I error rates) of the tests with
$\widehat{v}_{(\Wu)}$, $\widehat{v}_{(\DFW)}$ or $\widehat{v}_{(\Chan)}$ are not well-controlled at the 
nominal value 
because these estimators of $v$ are not accurate under the null hypothesis $H_0:\mu(t)\equiv0$.
Second, the powers of the tests with 
$\widehat{v}_{(\Wu)}$, $\widehat{v}_{(\DFW)}$ and $\widehat{v}_{(\Chan)}$ are low
because these estimators are not robust against $\mu$ under the alternative hypothesis $H_1$. 
For test (a) with $\widehat{v}_{(\Wu)}$ and $\widehat{v}_{(\DFW)}$, 
it even fails to be monotonically powerful with respect to  $\Xi$.
It means that it is harder to reject a more obviously wrong $H_1$. 
Clearly, the tests (a) and (b) with our proposed estimator $\widehat{v}_{(3)}^{\star}$ 
control the size well and 
are monotonically powerful.

\subsection{Simultaneous confidence bands for trends}
An estimator of $v$ is an essential component in many automatic procedures, 
e.g., bandwidth selection for nonparametric regression, and construction of SCBs, etc.
In this section, we present the local linear regression estimator \citep{wu_zhao_2007}:
\begin{align}\label{eqt:localLinearReg}
	\widehat{\mu}_{b}(t) = 2\bar{\mu}_{b}(t) - \bar{\mu}_{b\sqrt{2}}(t),
	\quad \text{where}\quad
	\bar{\mu}_{b}(t) = \sum_{i=1}^n \frac{H\{(t-i/n)/b\}}{nb} X_i,
\end{align}
and $H(\cdot)$ is a kernel, e.g., the Gaussian kernel, and $b$ is a bandwidth.
They suggest that $b$ could be selected as $b_n = 2 (\widehat{v}/\widehat{\gamma}_0)^{1/5} b^*$, 
where $\widehat{\gamma}_0 = \sum_{i=1}^n\{X_i - \widehat{\mu}_{b^*}(i/n)\}^2/n$, 
$b^*$ is the optimal bandwidth under the i.i.d. assumption \citep{RuppertSheatherWand95}, 
and $\widehat{v}$ is an estimator of $v$.
Since $b_n$ is crucial to the performance of $\widehat{\mu}_{b_n}(\cdot)$, 
an efficient estimator of $v$ is important. 

SCBs for $\mu(\cdot)$ directly depend on $\widehat{v}$. 
In particular, 95\% SCBs are given by $\widehat{\mu}_{b_n}(t) \pm \sqrt{\widehat{v}}{q}_{0.95}$, 
where ${q}_{0.95}$ is the 95\% quantile 
of $\sup_{0\leq t\leq 1}|\widehat{\mu}_{b_n^{\circ}}^{\circ}(t)|$
with $\widehat{\mu}_{b_n^{\circ}}^{\circ}(t)$ and $b_n^{\circ}$ computed on i.i.d. data 
$X_1^{\circ}, \ldots, X_n^{\circ}\sim \Normal(0,1)$.
The quantile $q_{0.95}$ can be easily obtained from simulation;
see Table 2 of \cite{wu_zhao_2007}.
A simulation experiment is performed to compare the coverage probability and the expected half-width of the 
95\% SCBs when $n=200$ and the true mean function is $\mu(t) = \cos(2\pi t)$.
We try different bandwidth $0.05\leq b \leq 0.1$.
The results are shown in Table~\ref{table:dTAVC_C}. 
The SCBs with $\widehat{v}_{(0)}^{\star}$ or $\widehat{v}_{(\DFW)}$ are over-covered, 
and their expected half-widths are large.
The SCBs with $\widehat{v}_{(\Wu)}$ or $\widehat{v}_{(\Chan)}$ are under-covered. 
The SCBs with our proposed $\widehat{v}_{(3)}^{\star}$ have a quite accurate coverage rate and
a reasonable expected width.

\begin{table*}[t]
\centering
\begin{tabular}{|c|ccccc|ccccc|}
\hline
$b$	 & WZ (2007) & DFW (2020) & C (2020) & $\widehat{v}_{(0)}^{\star}$ (Classical) & $\widehat{v}_{(3)}^{\star}$  (Proposal)  \\\hline\hline
	$0.05$	&	$ 92.0\%\,(  1.4)$	&	$ 99.4\%\,(  2.1)$	&	$ 91.8\%\,(  1.4)$	&	$100\%\,(  3.1)$	&	$ 94.9\%\,(  1.5)$	\\[0.5ex]
	$0.06$	&	$ 91.7\%\,(  1.2)$	&	$ 99.2\%\,(  1.9)$	&	$ 91.7\%\,(  1.2)$	&	$100\%\,(  2.8)$	&	$ 94.7\%\,(  1.3)$	\\[0.5ex]
	$0.07$	&	$ 91.4\%\,(  1.2)$	&	$ 99.1\%\,(  1.8)$	&	$ 91.3\%\,(  1.1)$	&	$100\%\,(  2.6)$	&	$ 94.4\%\,(  1.2)$	\\[0.5ex]
	$0.08$	&	$ 91.4\%\,(  1.1)$	&	$ 99.1\%\,(  1.6)$	&	$ 91.3\%\,(  1.1)$	&	$100\%\,(  2.4)$	&	$ 94.4\%\,(  1.2)$	\\[0.5ex]
	$0.09$	&	$ 91.3\%\,(  1.0)$	&	$ 99.1\%\,(  1.5)$	&	$ 91.2\%\,(  1.0)$	&	$100\%\,(  2.3)$	&	$ 94.3\%\,(  1.1)$	\\[0.5ex]
	$0.1$	&	$ 90.9\%\,(  1.0)$	&	$ 99.0\%\,(  1.4)$	&	$ 91.1\%\,(  1.0)$	&	$100\%\,(  2.2)$	&	$ 94.2\%\,(  1.0)$	\\[0.5ex]
\hline
\end{tabular}
\caption{The coverage probabilities of  95\% SCBs for $\mu(\cdot)$ under different 
		bandwidth $b$ and different estimators of $v$.
		The numbers inside parentheses are the expected half-widths.}  
\label{table:dTAVC_C}
\vspace{-0.4cm} 
\end{table*}

\subsection{Southern hemispheric land and ocean temperature}
The earth surface temperature has been an actively discussed topic in various fields. 
This section studies the southern hemispheric land and ocean monthly temperature
from 1880 to 2018 ($n=139\times 12$).  
The dataset is freely accessible from the website of NOAA's National Centers for Environmental Information.

Since land temperature changes more rapidly than ocean temperature, 
the land temperature is expected to be more volatile. 
The LRV $v$ is a measure of the stochastic variability of the average. 
Using $\widehat{v}_{(3)}^{\star}$, 
the long-run standard deviations ($\sqrt{v}$) 
for the land and ocean series are about $0.525$ and $0.313$, respectively. 
We can also compute the long-run correlation between the land and ocean average temperature. 
It is about 0.651, which indicates a moderately strong correlation. 
If the non-robust $\widehat{v}_{(0)}^{\star}$ is used, 
it is inflated to 0.966
as $\widehat{v}_{(0)}^{\star}$ mistakenly regards co-movement of trends as correlation.

\begin{figure} 
\begin{center}
\includegraphics[width=\textwidth]{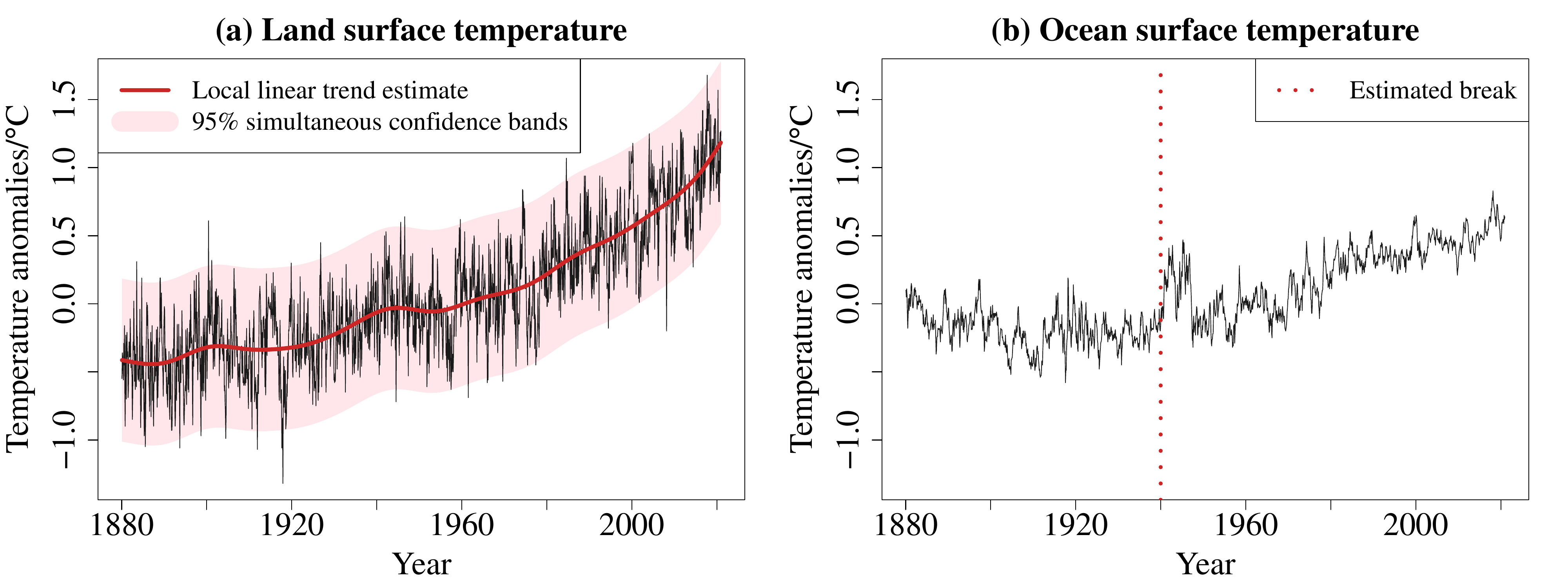} 
\end{center} 
\vspace{-0.6cm}
\caption{(a) Trend estimate and 95\% SCBs for the mean of the land surface temperature.
		(b) Estimated break location for the mean of the ocean surface temperature.}
		\label{fig:dTAVC_F001}
\vspace{-0.45cm} 
\end{figure}

Next, we test whether the global temperature has a structural break
even it has a possibly increasing trend. 
The test in \cite{wu_zhao_2007} is used. 
The $p$-values for the land series and ocean series are 15.4\% and $<10^{-5}$, 
respectively.  
Since the land mean temperature is smooth, we produce a trend estimate together with 95\%
SCBs in Figure~\ref{fig:dTAVC_F001} (a). 
The bandwidth selected for $\widehat{\mu}_{b_n}(\cdot)$ in (\ref{eqt:localLinearReg}) 
is $b_n = 2(0.276/0.062)^{1/5}(0.017) \approx 0.046$.
For the ocean temperature, the structural break is detected at around 1940 using the break location 
estimator proposed in \cite{wu_zhao_2007}; 
see Figure~\ref{fig:dTAVC_F001} (b).  
We suspect that 
it is easier to detect a structural break in the ocean series 
because the ocean temperature has a smaller LRV relative to the variation of $\mu(\cdot)$.

\section{Summary, discussion, and future work}\label{sec:conclusion}
This article presents a general class of difference-based estimators $\widehat{v}$ (\ref{eqt:diffvarEst}) 
for the long-run variance (\ref{eqt:def_v}).
Many existing estimators are special cases.
We derive the regimes in which $\widehat{v}$ is consistent and rate-optimal; see Table~\ref{tab:summary}. 
In particular, the intuitive estimator with locally centered $X_{1:n}$ is inadmissible.
We also derive detailed $\mathcal{L}^2$ properties of $\widehat{v}$.
It is proven to be asymptotically optimal even 
in the presence of trends and a possibly divergent number of change points. 
The suggested estimator is stated in (\ref{eqt:suggestedEst}).
We list some possible future work below.

\begin{itemize}
\item From Theorem~\ref{thm:robustness}, a possible estimator of 
$\mathcal{V} = \int_0^1 \left\{ \mu(t) - \bar{\mu} \right\}^2 \, \dd t$ is 
$\widehat{\mathcal{V}}_{(m)} = \{\widehat{v}_{(0)}-\widehat{v}_{(m)}\}/\{\ell\int_{-1}^1 K(t) \, \dd t\}$
for some $m>0$. 
It can be used to test constancy of $\mu(\cdot)$, i.e., $H_0:\mathcal{V}=0$.
The resulting test is expected to perform well because the estimator $\widehat{v}_{(m)}$ has a small MSE. 
\item To test for a variance change point in the presence of a non-constant mean, 
a possible test statistic can be constructed by using $\widehat{Q}_{(m)} = \sum_{i} D_i^2/n$ with parameters chosen to minimize 
the long-run variance of $\log\widehat{Q}_{(m)}$,
where $D_i$ is our proposed $m$th order difference statistics.
This test could be used as an alternative to the recent work by \cite{GSDR2019}. 
\end{itemize}
All these directions rely on the optimal framework proposed in this article.

%
%
%
%

\begin{acks}[Acknowledgments]
The authors would like to thank the anonymous referees, an Associate
Editor and the Editor for their constructive comments that improved the
scope and presentation of the paper. 
\end{acks}

\begin{funding}
This research was partially supported by grants GRF-2130730 and GRF-2130788 provided by Research Grants Council of HKSAR. 
\end{funding}

\begin{supplement}
\stitle{Appendix \ref{sec:proof_main}: Proofs of Main Results.}
\sdescription{The proofs of 
Propositions~\ref{prop:Ktype_Stype_equiv}, \ref{prop:effKernel}, 
Theorems~\ref{thm:robustness}, \ref{thm:consistency_finite}, \ref{thm:consistency_infinite}, \ref{thm:bias}, \ref{thm:var}, 
Corollaries~\ref{coro:optim}, \ref{coro:vpHat}, and 
Corollaries~\ref{coro:prewhite}, \ref{coro:multi}
are placed in Sections \ref{sec:proof_Ktype_Stype_equiv}--\ref{sec:proof_multi}, 
respectively. 
}
\end{supplement}
\begin{supplement}
\stitle{Appendix \ref{sec:auxillary}: Auxiliary results.}
\sdescription{
Technical results of independent interest are stated in 
Sections \ref{sec:proof_rough_var}--\ref{sec:proof_blocking_vHatp}. 
The derivation of (\ref{eqt:compareMSEexisting}) is stated in Section \ref{sec:derivationOfCompareMSEexisting}. 
}
\end{supplement}
\begin{supplement}
\stitle{Appendix \ref{sec:additional_plots}: Additional plots.}
\sdescription{
It contains additional simulation results for Section~\ref{sec:sim_MSE}. 
}
\end{supplement}
\begin{supplement}
\stitle{Appendix \ref{sec:robustness_experiment_additional}: Additional simulation experiments.}
\sdescription{It contains additional simulation experiments about 
the robustness against mean functions. 
}
\end{supplement}



\bibliographystyle{imsart-nameyear} 
\bibliography{myRef.bib}

\begin{thebibliography}{58}

\bibitem[\protect\citeauthoryear{Alexopoulos, Goldsman and
  Wilson}{2011}]{Alexopoulos2011}
\begin{barticle}[author]
\bauthor{\bsnm{Alexopoulos},~\bfnm{C.}\binits{C.}},
  \bauthor{\bsnm{Goldsman},~\bfnm{D.}\binits{D.}} \AND
  \bauthor{\bsnm{Wilson},~\bfnm{J.~R.}\binits{J.~R.}}
(\byear{2011}).
\btitle{Overlapping Batch Means: Something more for Nothing?}
\bjournal{Proceedings of the Winter Simulation Conference}
\bpages{401--411}.
\end{barticle}
\endbibitem

\bibitem[\protect\citeauthoryear{Altissimoa and
  Corradic}{2003}]{Altissimoa_Corradic_2003}
\begin{barticle}[author]
\bauthor{\bsnm{Altissimoa},~\bfnm{F.}\binits{F.}} \AND
  \bauthor{\bsnm{Corradic},~\bfnm{V.}\binits{V.}}
(\byear{2003}).
\btitle{Strong rules for detecting the number of breaks in a time series}.
\bjournal{J. Econometrics}
\bvolume{117}
\bpages{207--244}.
\end{barticle}
\endbibitem

\bibitem[\protect\citeauthoryear{Anderson}{1971}]{anderson1971}
\begin{bbook}[author]
\bauthor{\bsnm{Anderson},~\bfnm{T.~W.}\binits{T.~W.}}
(\byear{1971}).
\btitle{The Statistical Analysis of Time Series}.
\bpublisher{Wiley, New York}.
\end{bbook}
\endbibitem

\bibitem[\protect\citeauthoryear{Andrews}{1991}]{andrews1991}
\begin{barticle}[author]
\bauthor{\bsnm{Andrews},~\bfnm{D.~W.~K.}\binits{D.~W.~K.}}
(\byear{1991}).
\btitle{Heteroskedasticity and Autocorrelation Consistent Covariance Matrix
  Estimation}.
\bjournal{Econometrica}
\bvolume{59}
\bpages{817--858}.
\end{barticle}
\endbibitem

\bibitem[\protect\citeauthoryear{Bradley}{2005}]{bradley2005}
\begin{barticle}[author]
\bauthor{\bsnm{Bradley},~\bfnm{R.~C.}\binits{R.~C.}}
(\byear{2005}).
\btitle{Basic properties of strong mixing conditions. A survey and some open
  questions}.
\bjournal{Probability surveys}
\bvolume{2}
\bpages{107--144}.
\end{barticle}
\endbibitem

\bibitem[\protect\citeauthoryear{Brockwell and
  Davis}{1991}]{brockwellDavis1991}
\begin{bbook}[author]
\bauthor{\bsnm{Brockwell},~\bfnm{P.~J.}\binits{P.~J.}} \AND
  \bauthor{\bsnm{Davis},~\bfnm{R.~A.}\binits{R.~A.}}
(\byear{1991}).
\btitle{Time Series: Theory and Methods}.
\bpublisher{Springer, New York}.
\end{bbook}
\endbibitem

\bibitem[\protect\citeauthoryear{Carlstein}{1986}]{carlstein86}
\begin{barticle}[author]
\bauthor{\bsnm{Carlstein},~\bfnm{E}\binits{E.}}
(\byear{1986}).
\btitle{The Use of Subseries Values for Estimating the Variance of a General
  Statistic from a Stationary Sequence}.
\bjournal{Ann. Statist.}
\bvolume{14}
\bpages{1171--1179}.
\end{barticle}
\endbibitem

\bibitem[\protect\citeauthoryear{Chan}{2020}]{chan2020}
\begin{barticle}[author]
\bauthor{\bsnm{Chan},~\bfnm{K.~W.}\binits{K.~W.}}
(\byear{2020}).
\btitle{Mean-Structure and Autocorrelation Consistent Covariance Matrix
  Estimation}.
\bjournal{To appear in Journal of Business \& Economic Statistics}.
\end{barticle}
\endbibitem

\bibitem[\protect\citeauthoryear{Chan and Yau}{2016}]{chanyau2013}
\begin{barticle}[author]
\bauthor{\bsnm{Chan},~\bfnm{K.~W.}\binits{K.~W.}} \AND
  \bauthor{\bsnm{Yau},~\bfnm{C.~Y.}\binits{C.~Y.}}
(\byear{2016}).
\btitle{New Recursive Estimators of the Time-Average Variance Constant}.
\bjournal{Stat. Comput.}
\bvolume{26}
\bpages{609--627}.
\end{barticle}
\endbibitem

\bibitem[\protect\citeauthoryear{Chan and Yau}{2017a}]{chanyau2014_rTACM}
\begin{barticle}[author]
\bauthor{\bsnm{Chan},~\bfnm{K.~W.}\binits{K.~W.}} \AND
  \bauthor{\bsnm{Yau},~\bfnm{C.~Y.}\binits{C.~Y.}}
(\byear{2017}a).
\btitle{Automatic Optimal Batch Size Selection for Recursive Estimators of
  Time-average Covariance Matrix}.
\bjournal{J. Amer. Statist. Assoc.}
\bvolume{112}
\bpages{1076-1089}.
\end{barticle}
\endbibitem

\bibitem[\protect\citeauthoryear{Chan and Yau}{2017b}]{chanyau2015_hoc}
\begin{barticle}[author]
\bauthor{\bsnm{Chan},~\bfnm{K.~W.}\binits{K.~W.}} \AND
  \bauthor{\bsnm{Yau},~\bfnm{C.~Y.}\binits{C.~Y.}}
(\byear{2017}b).
\btitle{High Order Corrected Estimator of Asymptotic Variance with Optimal
  Bandwidth}.
\bjournal{Scand. J. Statist.}
\bvolume{44}
\bpages{866--898}.
\end{barticle}
\endbibitem

\bibitem[\protect\citeauthoryear{Chen, Wang and Wu}{2021}]{ChenWangWu2021}
\begin{barticle}[author]
\bauthor{\bsnm{Chen},~\bfnm{L.}\binits{L.}},
  \bauthor{\bsnm{Wang},~\bfnm{W.}\binits{W.}} \AND
  \bauthor{\bsnm{Wu},~\bfnm{W.~B.}\binits{W.~B.}}
(\byear{2021}).
\btitle{Inference of breakpoints in high-dimensional time series}.
\bjournal{To appear in J. Amer. Statist. Assoc.}
\end{barticle}
\endbibitem

\bibitem[\protect\citeauthoryear{Crainiceanu and
  Vogelsang}{2007}]{Crainiceanu2007}
\begin{barticle}[author]
\bauthor{\bsnm{Crainiceanu},~\bfnm{C.~M.}\binits{C.~M.}} \AND
  \bauthor{\bsnm{Vogelsang},~\bfnm{T.~J.}\binits{T.~J.}}
(\byear{2007}).
\btitle{Nonmonotonic Power for Tests of a Mean Shift in a Time Series}.
\bjournal{J. Stat. Comput. Simul.}
\bvolume{77}
\bpages{457--476}.
\end{barticle}
\endbibitem

\bibitem[\protect\citeauthoryear{Cs\"{o}rg\"{o} and
  Horv\'{a}th}{1997}]{csorgo1997}
\begin{bbook}[author]
\bauthor{\bsnm{Cs\"{o}rg\"{o}},~\bfnm{M.}\binits{M.}} \AND
  \bauthor{\bsnm{Horv\'{a}th},~\bfnm{L.}\binits{L.}}
(\byear{1997}).
\btitle{Limit Theorems in Change-Point Analysis}.
\bpublisher{Wiley, New York}.
\end{bbook}
\endbibitem

\bibitem[\protect\citeauthoryear{Dalla, Giraitis and
  Phillips}{2015}]{DallaGiraitisPhillips2015}
\begin{bbook}[author]
\bauthor{\bsnm{Dalla},~\bfnm{V.}\binits{V.}},
  \bauthor{\bsnm{Giraitis},~\bfnm{L.}\binits{L.}} \AND
  \bauthor{\bsnm{Phillips},~\bfnm{P.~C.~B.}\binits{P.~C.~B.}}
(\byear{2015}).
\btitle{Testing Mean Stability of Heteroskedastic Time Series}.
\bpublisher{Manuscript}.
\end{bbook}
\endbibitem

\bibitem[\protect\citeauthoryear{Dehling, Fried and
  Wendler}{2020}]{DehlingFriedWendler2020}
\begin{barticle}[author]
\bauthor{\bsnm{Dehling},~\bfnm{H.}\binits{H.}},
  \bauthor{\bsnm{Fried},~\bfnm{R.}\binits{R.}} \AND
  \bauthor{\bsnm{Wendler},~\bfnm{M.}\binits{M.}}
(\byear{2020}).
\btitle{A robust method for shift detection in time series}.
\bjournal{Biometrika}
\bvolume{107}
\bpages{647--660}.
\end{barticle}
\endbibitem

\bibitem[\protect\citeauthoryear{Dette, Eckle and Vetter}{2020}]{dette2020}
\begin{barticle}[author]
\bauthor{\bsnm{Dette},~\bfnm{H.}\binits{H.}},
  \bauthor{\bsnm{Eckle},~\bfnm{T.}\binits{T.}} \AND
  \bauthor{\bsnm{Vetter},~\bfnm{M.}\binits{M.}}
(\byear{2020}).
\btitle{Multiscale change point detection for dependent data}.
\bjournal{Scand. J. Statist.}
\bvolume{47}
\bpages{1243--1274}.
\end{barticle}
\endbibitem

\bibitem[\protect\citeauthoryear{Dette and Wu}{2019}]{DetteWu2019}
\begin{barticle}[author]
\bauthor{\bsnm{Dette},~\bfnm{H.}\binits{H.}} \AND
  \bauthor{\bsnm{Wu},~\bfnm{W.}\binits{W.}}
(\byear{2019}).
\btitle{Change point analysis in non-stationary processes -- a mass excess
  approach}.
\bjournal{To appear in Ann. Statist.}
\bvolume{47}
\bpages{3578--3608}.
\end{barticle}
\endbibitem

\bibitem[\protect\citeauthoryear{Gallant}{1987}]{gallant1987}
\begin{bbook}[author]
\bauthor{\bsnm{Gallant},~\bfnm{A.~R.}\binits{A.~R.}}
(\byear{1987}).
\btitle{Nonlinear Statistical Models}.
\bpublisher{John Wiley \& Sons}.
\end{bbook}
\endbibitem

\bibitem[\protect\citeauthoryear{Gao et~al.}{2019}]{GSDR2019}
\begin{barticle}[author]
\bauthor{\bsnm{Gao},~\bfnm{Z.}\binits{Z.}},
  \bauthor{\bsnm{Shang},~\bfnm{Z.}\binits{Z.}},
  \bauthor{\bsnm{Du},~\bfnm{P.}\binits{P.}} \AND
  \bauthor{\bsnm{Robertson},~\bfnm{J.~L.}\binits{J.~L.}}
(\byear{2019}).
\btitle{Variance Change Point Detection Under a Smoothly-Changing Mean Trend
  with Application to Liver Procurement}.
\bjournal{J. Amer. Statist. Assoc.}
\bvolume{114}
\bpages{773--781}.
\end{barticle}
\endbibitem

\bibitem[\protect\citeauthoryear{Gon\c{c}alves and
  White}{2002}]{GoncalvesWhite2002}
\begin{barticle}[author]
\bauthor{\bsnm{Gon\c{c}alves},~\bfnm{S.}\binits{S.}} \AND
  \bauthor{\bsnm{White},~\bfnm{H.}\binits{H.}}
(\byear{2002}).
\btitle{The Bootstrap of the Mean for Dependent Heterogeneous Arrays}.
\bjournal{Econometric Theory}
\bvolume{18}
\bpages{1367--1384}.
\end{barticle}
\endbibitem

\bibitem[\protect\citeauthoryear{G{\'o}reckia, Horv{\'a}thb and
  Kokoszka}{2018}]{GHK2018}
\begin{barticle}[author]
\bauthor{\bsnm{G{\'o}reckia},~\bfnm{T}\binits{T.}},
  \bauthor{\bsnm{Horv{\'a}thb},~\bfnm{L.}\binits{L.}} \AND
  \bauthor{\bsnm{Kokoszka},~\bfnm{P.}\binits{P.}}
(\byear{2018}).
\btitle{Change point detection in heteroscedastic time series}.
\bjournal{Econometrics and Statistics}
\bvolume{7}
\bpages{63--88}.
\end{barticle}
\endbibitem

\bibitem[\protect\citeauthoryear{Hall, Kay and Titterinton}{1990}]{hall1990}
\begin{barticle}[author]
\bauthor{\bsnm{Hall},~\bfnm{P.}\binits{P.}},
  \bauthor{\bsnm{Kay},~\bfnm{J.~W.}\binits{J.~W.}} \AND
  \bauthor{\bsnm{Titterinton},~\bfnm{D.~M.}\binits{D.~M.}}
(\byear{1990}).
\btitle{Asymptotically optimal difference-based estimation of variance in
  nonparametric regression}.
\bjournal{Biometrika}
\bvolume{77}
\bpages{521--528}.
\end{barticle}
\endbibitem

\bibitem[\protect\citeauthoryear{Horv\'{a}th, Kokoszka and
  Steinebach}{1999}]{horvath1999}
\begin{barticle}[author]
\bauthor{\bsnm{Horv\'{a}th},~\bfnm{L.}\binits{L.}},
  \bauthor{\bsnm{Kokoszka},~\bfnm{P.}\binits{P.}} \AND
  \bauthor{\bsnm{Steinebach},~\bfnm{J.}\binits{J.}}
(\byear{1999}).
\btitle{Testing for Changes in Multivariate Dependent Observations with an
  Application to Temperature Changes}.
\bjournal{J. Multivariate Anal.}
\bvolume{68}
\bpages{96--119}.
\end{barticle}
\endbibitem

\bibitem[\protect\citeauthoryear{Ibragimov}{1962}]{ibragimov1962}
\begin{barticle}[author]
\bauthor{\bsnm{Ibragimov},~\bfnm{I.~A.}\binits{I.~A.}}
(\byear{1962}).
\btitle{Some limit theorems for stationary processes}.
\bjournal{Theory of Probability \& Its Applications}
\bvolume{7}
\bpages{349--382}.
\end{barticle}
\endbibitem

\bibitem[\protect\citeauthoryear{Juhl and Xiao}{2009}]{JuhlXiao2009}
\begin{barticle}[author]
\bauthor{\bsnm{Juhl},~\bfnm{T.}\binits{T.}} \AND
  \bauthor{\bsnm{Xiao},~\bfnm{Z.}\binits{Z.}}
(\byear{2009}).
\btitle{Tests for changing mean with monotonic power}.
\bjournal{J. Econometrics}
\bvolume{148}
\bpages{12-24}.
\end{barticle}
\endbibitem

\bibitem[\protect\citeauthoryear{K\"{u}nsch}{1989}]{kunsch89}
\begin{barticle}[author]
\bauthor{\bsnm{K\"{u}nsch},~\bfnm{H.~R.}\binits{H.~R.}}
(\byear{1989}).
\btitle{The Jackknife and the Bootstrap for General Stationary Observations}.
\bjournal{Ann. Statist.}
\bvolume{17}
\bpages{1217--1241}.
\end{barticle}
\endbibitem

\bibitem[\protect\citeauthoryear{Levine and
  Tecuapetla-G{\'o}mez}{2019}]{LevineTecuapetla2019}
\begin{barticle}[author]
\bauthor{\bsnm{Levine},~\bfnm{M}\binits{M.}} \AND
  \bauthor{\bsnm{Tecuapetla-G{\'o}mez},~\bfnm{I.}\binits{I.}}
(\byear{2019}).
\btitle{{ACF} estimation via difference schemes for a semiparametric model with
  $m$-dependent errors}.
\bjournal{arXiv preprint arXiv:1905.04578}.
\end{barticle}
\endbibitem

\bibitem[\protect\citeauthoryear{Liu and Wu}{2010}]{wu2010}
\begin{barticle}[author]
\bauthor{\bsnm{Liu},~\bfnm{W.}\binits{W.}} \AND
  \bauthor{\bsnm{Wu},~\bfnm{W.~B.}\binits{W.~B.}}
(\byear{2010}).
\btitle{Asymptomatic of Spectral Density Estimates}.
\bjournal{Econometric Theory}
\bvolume{26}
\bpages{1218--1245}.
\end{barticle}
\endbibitem

\bibitem[\protect\citeauthoryear{Lobato}{2001}]{lobato2001}
\begin{barticle}[author]
\bauthor{\bsnm{Lobato},~\bfnm{I.~N.}\binits{I.~N.}}
(\byear{2001}).
\btitle{Testing That a Dependent Process Is Uncorrelated}.
\bjournal{J. Amer. Statist. Assoc.}
\bvolume{96}
\bpages{1066--1076}.
\end{barticle}
\endbibitem

\bibitem[\protect\citeauthoryear{Meketon and Schmeiser}{1984}]{OBM1984}
\begin{barticle}[author]
\bauthor{\bsnm{Meketon},~\bfnm{M.~S.}\binits{M.~S.}} \AND
  \bauthor{\bsnm{Schmeiser},~\bfnm{B.}\binits{B.}}
(\byear{1984}).
\btitle{Overlapping batch means: something for nothing?}
\bjournal{Proceedings of the 16th Conference on Winter Simulation}
\bpages{226--230}.
\end{barticle}
\endbibitem

\bibitem[\protect\citeauthoryear{Newey and West}{1987}]{newwy_west_1987}
\begin{barticle}[author]
\bauthor{\bsnm{Newey},~\bfnm{W.~K.}\binits{W.~K.}} \AND
  \bauthor{\bsnm{West},~\bfnm{K.~D.}\binits{K.~D.}}
(\byear{1987}).
\btitle{A simple, positive semi-definite, heteroskedasticity and
  autocorrelation consistent covariance matrix}.
\bjournal{Econometrica}
\bvolume{55}
\bpages{703--708}.
\end{barticle}
\endbibitem

\bibitem[\protect\citeauthoryear{Parzen}{1957}]{Parzen1957}
\begin{barticle}[author]
\bauthor{\bsnm{Parzen},~\bfnm{E}\binits{E.}}
(\byear{1957}).
\btitle{On Consistent Estimates of the Spectrum of a Stationary Time Series}.
\bjournal{The Annals of Mathematical Statistics}
\bvolume{28}
\bpages{329--348}.
\end{barticle}
\endbibitem

\bibitem[\protect\citeauthoryear{Pe{\v{s}}ta and Wendler}{2020}]{pevsta2020}
\begin{barticle}[author]
\bauthor{\bsnm{Pe{\v{s}}ta},~\bfnm{M.}\binits{M.}} \AND
  \bauthor{\bsnm{Wendler},~\bfnm{M.}\binits{M.}}
(\byear{2020}).
\btitle{Nuisance-parameter-free changepoint detection in non-stationary
  series}.
\bjournal{Test}
\bvolume{29}
\bpages{379--408}.
\end{barticle}
\endbibitem

\bibitem[\protect\citeauthoryear{Politis}{2011}]{politis2011}
\begin{barticle}[author]
\bauthor{\bsnm{Politis},~\bfnm{D.~N.}\binits{D.~N.}}
(\byear{2011}).
\btitle{Higher-order accurate, positive semidefinite estimation of large-sample
  covariance and spectral density matrices}.
\bjournal{Econometric Theory}
\bvolume{27}
\bpages{703--744}.
\end{barticle}
\endbibitem

\bibitem[\protect\citeauthoryear{Politis and Romano}{1995}]{politis1995}
\begin{barticle}[author]
\bauthor{\bsnm{Politis},~\bfnm{D.~N.}\binits{D.~N.}} \AND
  \bauthor{\bsnm{Romano},~\bfnm{J.~P.}\binits{J.~P.}}
(\byear{1995}).
\btitle{Bias Corrected Nonparametric Spectral Estimation}.
\bjournal{J. Time Series Anal.}
\end{barticle}
\endbibitem

\bibitem[\protect\citeauthoryear{Politis, Romano and Wolf}{1999}]{politis1999}
\begin{bbook}[author]
\bauthor{\bsnm{Politis},~\bfnm{D.~N.}\binits{D.~N.}},
  \bauthor{\bsnm{Romano},~\bfnm{J.~P.}\binits{J.~P.}} \AND
  \bauthor{\bsnm{Wolf},~\bfnm{M.}\binits{M.}}
(\byear{1999}).
\btitle{Subsampling.}
\bpublisher{Springer, New York}.
\end{bbook}
\endbibitem

\bibitem[\protect\citeauthoryear{Rice}{1984}]{rice1984}
\begin{barticle}[author]
\bauthor{\bsnm{Rice},~\bfnm{J.}\binits{J.}}
(\byear{1984}).
\btitle{Bandwidth Choice for Nonparametric Regression}.
\bjournal{Ann. Statist.}
\bvolume{12}
\bpages{1215--1230}.
\end{barticle}
\endbibitem

\bibitem[\protect\citeauthoryear{Rosenblatt}{1956}]{Rosenblatt56}
\begin{barticle}[author]
\bauthor{\bsnm{Rosenblatt},~\bfnm{M}\binits{M.}}
(\byear{1956}).
\btitle{A central limit theorem and a strong mixing condition}.
\bjournal{Proc. Natl. Acad. Sci. USA}
\bvolume{42}
\bpages{43--47}.
\end{barticle}
\endbibitem

\bibitem[\protect\citeauthoryear{Ruppert, Sheather and
  Wand}{1995}]{RuppertSheatherWand95}
\begin{barticle}[author]
\bauthor{\bsnm{Ruppert},~\bfnm{D.}\binits{D.}},
  \bauthor{\bsnm{Sheather},~\bfnm{S.~J.}\binits{S.~J.}} \AND
  \bauthor{\bsnm{Wand},~\bfnm{M.~P.}\binits{M.~P.}}
(\byear{1995}).
\btitle{An Effective Bandwidth Selector for Local Least Squares Regression}.
\bjournal{J. Amer. Statist. Assoc.}
\bvolume{90}
\bpages{1257--1270}.
\end{barticle}
\endbibitem

\bibitem[\protect\citeauthoryear{Shao}{2010}]{shaoxf2010}
\begin{barticle}[author]
\bauthor{\bsnm{Shao},~\bfnm{Xiaofeng}\binits{X.}}
(\byear{2010}).
\btitle{A self-normalized approach to confidence interval construction in time
  series}.
\bjournal{J. R. Statist. Soc. B}
\bvolume{72}
\bpages{343--366}.
\end{barticle}
\endbibitem

\bibitem[\protect\citeauthoryear{Shao}{2015}]{shao2015review}
\begin{barticle}[author]
\bauthor{\bsnm{Shao},~\bfnm{X.}\binits{X.}}
(\byear{2015}).
\btitle{Self-Normalization for Time Series: A Review of Recent Developments}.
\bjournal{J. Amer. Statist. Assoc.}
\bvolume{110}
\bpages{1797--1817}.
\end{barticle}
\endbibitem

\bibitem[\protect\citeauthoryear{Song and Schmeiser}{1995}]{song1995}
\begin{barticle}[author]
\bauthor{\bsnm{Song},~\bfnm{W.~T.}\binits{W.~T.}} \AND
  \bauthor{\bsnm{Schmeiser},~\bfnm{B.~W.}\binits{B.~W.}}
(\byear{1995}).
\btitle{Optimal Mean-Squared-Error Batch Sizes}.
\bjournal{Manage. Sci.}
\bvolume{41}
\bpages{110--123}.
\end{barticle}
\endbibitem

\bibitem[\protect\citeauthoryear{Taqqu and Eberlein}{1986}]{TaqquEberlein86}
\begin{bbook}[author]
\bauthor{\bsnm{Taqqu},~\bfnm{M.}\binits{M.}} \AND \bauthor{\bsnm{Eberlein}}
(\byear{1986}).
\btitle{Dependence in Probability and Statistics}.
\bpublisher{Birkh\"{a}user Basel}.
\end{bbook}
\endbibitem

\bibitem[\protect\citeauthoryear{Tecuapetla‐G{\'o}mez and
  Munk}{2017}]{TecuapetlaMunk2017}
\begin{barticle}[author]
\bauthor{\bsnm{Tecuapetla‐G{\'o}mez},~\bfnm{I}\binits{I.}} \AND
  \bauthor{\bsnm{Munk},~\bfnm{A.}\binits{A.}}
(\byear{2017}).
\btitle{Autocovariance Estimation in Regression with a Discontinuous Signal and
  $m$-Dependent Errors: A Difference-Based Approach}.
\bjournal{Scand. J. Statist.}
\bvolume{44}
\bpages{346--368}.
\end{barticle}
\endbibitem

\bibitem[\protect\citeauthoryear{Vats and Flegal}{2021}]{VatsFlegal2018}
\begin{barticle}[author]
\bauthor{\bsnm{Vats},~\bfnm{D.}\binits{D.}} \AND
  \bauthor{\bsnm{Flegal},~\bfnm{J.~M.}\binits{J.~M.}}
(\byear{2021}).
\btitle{Lugsail lag windows for estimating time-average covariance matrices}.
\bjournal{To appear in Biometrika}.
\end{barticle}
\endbibitem

\bibitem[\protect\citeauthoryear{Welch}{1987}]{peter1987}
\begin{barticle}[author]
\bauthor{\bsnm{Welch},~\bfnm{P.~D.}\binits{P.~D.}}
(\byear{1987}).
\btitle{On the relationship between batch means, overlapping batch means and
  spectral estimation}.
\bjournal{Proceedings of the Winter Simulation Conference}
\bpages{320--323}.
\end{barticle}
\endbibitem

\bibitem[\protect\citeauthoryear{White}{1984}]{white1984}
\begin{bbook}[author]
\bauthor{\bsnm{White},~\bfnm{H.}\binits{H.}}
(\byear{1984}).
\btitle{Asymptotic theory for econometricians}.
\bpublisher{Academic Press, New York}.
\end{bbook}
\endbibitem

\bibitem[\protect\citeauthoryear{Wolkonski and Rozanov}{1959}]{wolkonski1959}
\begin{barticle}[author]
\bauthor{\bsnm{Wolkonski},~\bfnm{V.~A.}\binits{V.~A.}} \AND
  \bauthor{\bsnm{Rozanov},~\bfnm{Y.~A.}\binits{Y.~A.}}
(\byear{1959}).
\btitle{Some limit theorems for random functions, Part I}.
\bjournal{Theory Probab. Appl}
\bvolume{4}
\bpages{178--197}.
\end{barticle}
\endbibitem

\bibitem[\protect\citeauthoryear{Wu}{2004}]{wu2004CP}
\begin{bincollection}[author]
\bauthor{\bsnm{Wu},~\bfnm{W.~B.}\binits{W.~B.}}
(\byear{2004}).
\btitle{A test for detecting changes in mean}.
In \bbooktitle{Time Series Analysis and Applications to Geophysical Systems},
(\beditor{\bfnm{D.~R.}\binits{D.~R.}~\bsnm{Brillinger}},
  \beditor{\bfnm{E.~A.}\binits{E.~A.}~\bsnm{Robinson}} \AND
  \beditor{\bfnm{F.}\binits{F.}~\bsnm{Schoenberg.}}, eds.)
\bvolume{139}
\bpages{105--122}.
\bpublisher{Springer-Verlag New York}.
\end{bincollection}
\endbibitem

\bibitem[\protect\citeauthoryear{Wu}{2005}]{wu2005}
\begin{barticle}[author]
\bauthor{\bsnm{Wu},~\bfnm{W.~B.}\binits{W.~B.}}
(\byear{2005}).
\btitle{Nonlinear system theory: Another look at dependence}.
\bjournal{Proc. Natl. Acad. Sci. USA}
\bvolume{102}
\bpages{14150--14154}.
\end{barticle}
\endbibitem

\bibitem[\protect\citeauthoryear{Wu}{2007}]{wu2007}
\begin{barticle}[author]
\bauthor{\bsnm{Wu},~\bfnm{W.~B.}\binits{W.~B.}}
(\byear{2007}).
\btitle{Strong invariance principles for dependent random variables}.
\bjournal{Ann. Probab.}
\bvolume{35}
\bpages{2294--2320}.
\end{barticle}
\endbibitem

\bibitem[\protect\citeauthoryear{Wu}{2011}]{wu2011}
\begin{barticle}[author]
\bauthor{\bsnm{Wu},~\bfnm{W.~B.}\binits{W.~B.}}
(\byear{2011}).
\btitle{Asymptotic theory for stationary processes}.
\bjournal{Stat. Interface}
\bvolume{4}
\bpages{207--226}.
\end{barticle}
\endbibitem

\bibitem[\protect\citeauthoryear{Wu, Woodroofe and Mentz}{2001}]{wu2001}
\begin{barticle}[author]
\bauthor{\bsnm{Wu},~\bfnm{W.~B.}\binits{W.~B.}},
  \bauthor{\bsnm{Woodroofe},~\bfnm{M.}\binits{M.}} \AND
  \bauthor{\bsnm{Mentz},~\bfnm{G.}\binits{G.}}
(\byear{2001}).
\btitle{Isotonic regression: another look at the change point problem}.
\bjournal{Biometrika}
\bvolume{88}
\bpages{793--804}.
\end{barticle}
\endbibitem

\bibitem[\protect\citeauthoryear{Wu and Zhao}{2007}]{wu_zhao_2007}
\begin{barticle}[author]
\bauthor{\bsnm{Wu},~\bfnm{W.~B.}\binits{W.~B.}} \AND
  \bauthor{\bsnm{Zhao},~\bfnm{Z.}\binits{Z.}}
(\byear{2007}).
\btitle{Inference of trends in time series}.
\bjournal{J. R. Statist. Soc. B}
\bvolume{69}
\bpages{391--410}.
\end{barticle}
\endbibitem

\bibitem[\protect\citeauthoryear{Xiao and Wu}{2012}]{xiaowu2012}
\begin{barticle}[author]
\bauthor{\bsnm{Xiao},~\bfnm{H.}\binits{H.}} \AND
  \bauthor{\bsnm{Wu},~\bfnm{W.~B.}\binits{W.~B.}}
(\byear{2012}).
\btitle{Covariance matrix estimation for stationary time series}.
\bjournal{Ann. Statist.}
\bvolume{40}
\bpages{466--493}.
\end{barticle}
\endbibitem

\bibitem[\protect\citeauthoryear{Zhang and Lavitas}{2018}]{zhang2018}
\begin{barticle}[author]
\bauthor{\bsnm{Zhang},~\bfnm{T.}\binits{T.}} \AND
  \bauthor{\bsnm{Lavitas},~\bfnm{L.}\binits{L.}}
(\byear{2018}).
\btitle{Unsupervised self-normalized change-point testing for time series}.
\bjournal{J. Amer. Statist. Assoc.}
\bvolume{113}
\bpages{637--648}.
\end{barticle}
\endbibitem

\bibitem[\protect\citeauthoryear{Zhao}{2011}]{zhao2011}
\begin{barticle}[author]
\bauthor{\bsnm{Zhao},~\bfnm{Z.}\binits{Z.}}
(\byear{2011}).
\btitle{A Self-Normalized Confidence Interval for the Mean of a Class of
  Nonstationary Processes}.
\bjournal{Biometrika}
\bvolume{98}
\bpages{81--90}.
\end{barticle}
\endbibitem

\end{thebibliography}

\newpage
\begin{center}
\bfseries\MakeUppercase{Supplementary Note to ``Optimal Difference-based Variance Estimators in Time Series: A General Framework''}
\end{center}
\appendix

\appendix
\section{Proofs of Main Results}\label{sec:proof_main}
Throughout the proof, we may drop the subscript ``$\mu$'' in $\MSE_{\mu}$, $\Var_{\mu}$ and $\Bias_{\mu}$
in the general-mean case when there is no confusion.
Similarly, we may drop the subscript ``$0$'' in $\MSE_{0}$, $\Var_{0}$ and $\Bias_{0}$ in the zero-mean case
when there is no confusion.

\subsection{Proof of Proposition~\ref{prop:Ktype_Stype_equiv}}\label{sec:proof_Ktype_Stype_equiv}
Rewritting $\widehat{v}$ and $\widehat{v}'$ in quadratic forms, we have
\begin{align}
	\widehat{v}
		&= \sum_{i=1}^n\sum_{i'=1}^n
			\left\{ \frac{1}{n} K\left( \frac{i-i'}{\ell} \right) \mathbb{1}_{(\min(i,i')\geq mh+1)} D_i D_{i'} \right\} , \label{eqt:vHat_quad}\\
	\widehat{v}' 
		&= \sum_{i=1}^n\sum_{i'=1}^n
			\left\{ \frac{1}{n} K\left( \frac{i-i'}{\ell} \right) \mathbb{1}_{(\min(i,i')\geq mh+1)} D_i D_{i'}  \right\}
			 A_{i,i'} B_n ,\label{eqt:vHatp_quad} \nonumber
\end{align}
where 
\begin{align*}
	A_{i,i'}&=  
		 \min\left\{ \frac{\min(i,i')-mh-1}{\ell-|i-i'|}, \frac{n-\max(i,i')+1}{\ell-|i-i'|} , 1\right\}, \\
	B_n &= \frac{n}{n-(mh+\ell+1)+1}.
\end{align*}
Hence, by Minkowski inequality, we have 
\begin{eqnarray}\label{eqt:diff_K_S}
	\|\widehat{v}-\widehat{v}'\|
	&\leq& \|B_n\widehat{v}-\widehat{v}'\| + (1-B_n)\|\widehat{v}\| \nonumber \\
	&\leq& \|B_n\widehat{v}-\widehat{v}' - \E(B_n\widehat{v}-\widehat{v}')\| 
				+ |\E(B_n\widehat{v}-\widehat{v}')|
				+ \frac{\ell+mh}{n} \|\widehat{v}\|. 
\end{eqnarray}
Note that $A_{i,i'}=1$ when $\max(i,i') \geq \ell + mh+1$ and $\min(i,i') \leq n-\ell+1$; 
and $|A_{i,i'}|\leq 1$ for all $i,i'$. 
Hence, among $n^2$ summands in $B_n\widehat{v}$ and $\widehat{v}'$, 
at most $O(\ell+mh)^2$ of them are not identical, 
and the coefficients of $D_i D_{i'}$ are $O(1/n)$.
Since $\Theta_2\leq \Theta_4<\infty$ and $\{|d_j|\}$ is uniformly summable, we have, 
by Minkowski's inequality, that  
\begin{eqnarray}
	\sum_{t=0}^{\infty} \|D_t - D_{t,\{0\}}\|_2
		&\leq& \sum_{t=0}^{\infty} \sum_{j=0}^m |d_j| \|Z_{t-hj} - Z_{t-hj,\{0\}}\|_2 \nonumber \\
		&=& \sum_{j=0}^m |d_j| \Theta_2 
		<\infty, \label{eqt:stableD}
\end{eqnarray}
where $D_{t,\{0\}}$ is a coupled version of $D_t$
with $\varepsilon_0$ being replaced by an i.i.d. copy. 
So, by Lemma 8 of \citet{xiaowu2012}, 
we have 
\begin{eqnarray}\label{eqt:diffvHat_inter1}
	\left\|B_n\widehat{v}-\widehat{v}' - \E(B_n\widehat{v}-\widehat{v}')\right\| 
	\leq O\left( \frac{\ell+mh}{n} \right).
\end{eqnarray}

Recall that $\mu_i = c_i + s_i$, where  
$c_i := c(i/n)$ and 
$s_i := s(i/n) = \sum_{j=0}^{\mathcal{J}} \xi_{j} \mathbb{1}(T_j \leq i < T_{j+1})$ for each $i$. 
Let $\varrho_i = \mathbb{1}(\{i-mh,\ldots, i\}\cap\{T_1, \ldots, T_{\mathcal{J}}\} \neq \emptyset)$
for $i=mh=1, \ldots,n$. 
The quantity $\varrho_i$ indicates whether changes point(s) exist at times $i-mh, \ldots, i$. 
Then, note that 
\begin{eqnarray}
	\E(D_i) = \sum_{j=0}^m d_j \mu_{i-hj}
		= \left( \sum_{j=0}^m d_j c_{i-hj} \right) + \left( \sum_{j=0}^m d_j s_{i-hj} \right) .\label{eqt:ED_1}
\end{eqnarray}
By the definition of $\mathcal{C}$, we have 
\begin{eqnarray}
	\left\vert\sum_{j=0}^m d_j c_{i-hj} \right\vert 
		= \left\vert\sum_{j=0}^m d_j (c_{i-hj} - c_i)\right\vert  
		\leq \sum_{j=0}^m |d_j| O(\mathcal{C}mh/n)
		= O(\mathcal{C}mh/n). \label{eqt:ED_2}
\end{eqnarray} 
Since $\mathcal{G} \gtrsim \ell+mh$, 
we know that, for each $i$, 
there are at most $O(1)$ jumps in the time period $\{i-mh, \ldots, i\}$ for large enough $n$. 
Consequently, $|s_{i-hj}- s_i| = O(\mathcal{S})$ for all $j=0, \ldots,m$. So, 
\begin{eqnarray}
	\left\vert\sum_{j=0}^m d_j s_{i-hj}\right\vert  
		= \left\vert\sum_{j=0}^m d_j (s_{i-hj}- s_i)\right\vert 
		\leq  \varrho_i\sum_{j=0}^m |d_j| O(\mathcal{S})  
		= O(\mathcal{S}\varrho_i).  \label{eqt:ED_3}
\end{eqnarray}
Putting (\ref{eqt:ED_2}) and (\ref{eqt:ED_3}) into (\ref{eqt:ED_1}), we have 
$|\E(D_i)| = O(\mathcal{C}mh/n) + O(\mathcal{S}\varrho_i)$. 
Since at most $O(1)$ of $\{ \varrho_i : i\in\{1+mh, \ldots, \ell+mh,n-\ell, \ldots, n\} \}$ are non-zero 
owing to the assumption $\mathcal{G} \gtrsim \ell+mh$, we obtain 
\begin{eqnarray}\label{eqt:sumsumEDED}
	&&\left\vert \sum_{i=1}^n\sum_{i'=1}^n \frac{1}{n} 
		K\left( \frac{i-i'}{\ell} \right) \mathbb{1}_{(\min(i,i')\geq mh+1)}
		\E(D_i)\E(D_{i'})(1-A_{i,i'}) \right\vert \nonumber \\
	&&\qquad = O\left( \frac{\mathcal{C}^2m^2h^2(\ell+mh)^2}{n^3} + \frac{\mathcal{S}^2}{n}\right) \equiv r_{\sub}.
\end{eqnarray}
Using (\ref{eqt:ACVF_D}) and (\ref{eqt:sumsumEDED}), we have
\begin{eqnarray}
	&&|\E(B_n\widehat{v}-\widehat{v}')| \nonumber \\
	&&\quad\leq B_n\Bigg\vert \sum_{i=1}^n\sum_{i'=1}^n
			\frac{1}{n} K\left( \frac{i-i'}{\ell} \right) \mathbb{1}_{(\min(i,i')\geq mh+1)}
			\nonumber \\
	&&  \qquad\qquad
			\times\left\{ \sum_{|s|\leq m} \delta_{s} \gamma_{hs+i-i'} + \E(D_i)\E(D_{i'}) \right\}
			(1-A_{i,i'})
			\Bigg\vert \nonumber\\
	&&\quad\leq r_{\sub}
			+ O\left(\frac{1}{n}\right)\sum_{|s|\leq m}|\delta_s|   
			\left\{
			\left(\mathop{\sum\sum}_{i, i' \leq mh} |\gamma_{hs+i-i'}| \right) + 
			\left(\mathop{\sum\sum}_{i,i'\geq n-\ell} |\gamma_{hs+i-i'}| \right)
			\right\} \nonumber \\
	&&\quad\leq r_{\sub}
			+ O\left(\frac{1}{n}\right)\sum_{|s|\leq m}|\delta_s|   
			\left[
			\left\{\sum_{|k|\leq mh} ( mh- |k|)|\gamma_{hs+k}|\right\} 
			+\left\{\sum_{|k|\leq \ell} ( \ell- |k|)|\gamma_{hs+k}| \right\}
			\right] \nonumber\\
	&&\quad\leq
		r_{\sub} +
		O\left(\frac{1}{n}\right) \sum_{|s|\leq m}|\delta_s|   (mh +\ell) u_0 \nonumber \label{eqt:kLemma_diffvHat}\\
	&&\quad= r_{\sub}
		+O\left(\frac{\ell+mh}{n}\right) , \nonumber \\
		\label{eqt:diffvHat_inter2}
\end{eqnarray}		
where the last line follows from 
$\sup_{n\in\mathbb{N}}\sum_{|s|\leq m}|\delta_s|<\infty$ and 
$u_0 <\infty$ implied by Assumption~\ref{ass:weakDep}.
Putting (\ref{eqt:diffvHat_inter1}) and (\ref{eqt:diffvHat_inter2}) into (\ref{eqt:diff_K_S}), 
we conclude that 
\[
	\|\widehat{v}-\widehat{v}'\| \leq O\left(\frac{\ell+mh}{n}\right) (1+\|\widehat{v}\|) 
			+  r_{\sub} . 
\]

\subsection{Proof of Proposition~\ref{prop:effKernel}}\label{sec:proof_prop:effKernel}
Let $\mu_0 = \E(X_1)= \ldots = \E(X_n)$ be the common mean of the data. 
For notational simplicity, denote 
$\widehat{Z}_i^{(\theta)} = X_i - \theta$ for each $i$, 
and 
\[
	\widehat{\gamma}_{k}^{(\theta)}
		= \frac{1}{n}\sum_{i=|k|+1}^n \widehat{Z}_i^{(\theta)} \widehat{Z}_{i-|k|}^{(\theta)}   
\] 
for each $k$, 
where $\theta \in\{ \mu_0 , \bar{X} \}$.
So, $\widehat{\gamma}_{k}^{(\mu_0)} = \widehat{\gamma}_{k}^{Z}$ and 
$\widehat{\gamma}_{k}^{(\bar{X})} = \widehat{\gamma}_{k}^{X}$.

\begin{enumerate}
\item Recall the definition of $\widehat{v}'$ defined in (\ref{eqt:Stype_vHat}), 
which does not change if we replace $\bar{X}$ by any fixed $\theta$.
We could re-write $\widehat{v}'$ as follow:
\begin{align*}
	\widehat{v}' 
		&= \frac{1}{n-mh-\ell}\sum_{i=mh+\ell+1}^n \sum_{t=i-\ell+1}^i \sum_{t'=i-\ell+1}^i
				\frac{K\left( \frac{t-t'}{\ell} \right)}{\ell-\left\vert t-t' \right\vert}
				\sum_{j=0}^m \sum_{j'=0}^m d_jd_{j'} 
				\widehat{Z}_{t-hj}^{(\theta)}\widehat{Z}_{t'-hj'}^{(\theta)}\nonumber \\
		&= \frac{1}{n-mh-\ell}\sum_{i=h_m+\ell+1}^n \sum_{t,t'=1}^{\ell} \sum_{j,j'=0}^m 
				\frac{K\left( \frac{t-t'}{\ell} \right)}{\ell-\left\vert t-t' \right\vert}
				d_jd_{j'} \widehat{Z}_{i-\ell+t-jh}^{(\theta)}\widehat{Z}_{i-\ell+t'-j'h}^{(\theta)}  \\
		&= \frac{1}{n-mh-\ell}\sum_{t,t'=1}^{\ell} \sum_{j,j'=0}^m 
				\frac{K\left( \frac{t-t'}{\ell} \right)}{\ell-\left\vert t-t' \right\vert}
				d_jd_{j'} 
				\left\{ \sum_{i=hm-hj+1+t}^{n-\ell+t-jh} 
				\widehat{Z}_{i}^{(\theta)}\widehat{Z}_{i -\{ t-t' - h(j- j')\}}^{(\theta)} \right\}.
\end{align*}
Next, note that $K((k+hs)/\ell)=0$ when $s<\lceil-(\ell+k)/h\rceil$ or $s>\lfloor (\ell-k)/h \rfloor$. 
Then, by changing the running index variable and exchanging the summations, we have 
\begin{align*}
	\widehat{v}_{\diff} 
		&= \sum_{|k|\leq \ell+mh} \left\{ \sum_{s=\lceil-(\ell+k)/h\rceil}^{\lfloor (\ell-k)/h \rfloor} \delta_{|s|}  K\left( \frac{k+hs}{\ell} \right) \right\}\widehat{\gamma}^{(\theta)}_{k} \\
		&= \sum_{s=-m}^m \delta_{|s|} \sum_{k=-\ell+1-hs}^{\ell-1-hs} K\left( \frac{k+hs}{\ell} \right) \widehat{\gamma}^{(\theta)}_{k} \\	
		&= \sum_{s=-m}^m \delta_{|s|} \sum_{k=-\ell+1}^{\ell-1} K\left( \frac{k}{\ell} \right) \widehat{\gamma}^{(\theta)}_{k-hs} \\	
		&= \sum_{j,j'=0}^m d_jd_{j'}
			\sum_{k=-\ell+1}^{\ell-1}
				K\left( \frac{k}{\ell} \right) \widehat{\gamma}^{(\theta)}_{k-h(j-j')}  \\
		&= \frac{1}{n}\sum_{t,t'=1}^{\ell} \sum_{j,j'=0}^m 
				\frac{K\left( \frac{t-t'}{\ell} \right)}{\ell-|t-t'|}
				d_jd_{j'} 
				\left\{ \sum_{i=1+\max(0,t-t' - h(j- j'))}^{n+\min(0,t-t' - h(j- j'))} 
				\widehat{Z}_{i}^{(\theta)}\widehat{Z}_{i - \{t-t' - h(j- j')\}}^{(\theta)}  \right\}.
\end{align*}
Let 
\begin{align*}
	\mathcal{B}_n &= 
				\bigg\{ 1, \ldots, \ell+mh \bigg\} \bigcup 
				\bigg\{n-(\ell+mh)+1, \ldots, n\bigg\}, \\
	\mathcal{B}_{j,j',t,t'} &=
				\bigg\{ \max\{0,t-t'-h(j-j')\}+1, \ldots, hm-hj+t \bigg\} \\
				&\quad\bigcup \bigg\{n-\ell+t-jh+1, \ldots, n+\min\{0,t-t'-h(j-j')\}\bigg\}. 
\end{align*}
Clearly, for any $0\leq j,j'\leq m$, any $1\leq t,t' \leq \ell$ and any $n\in\mathbb{N}$, 
we have $\mathcal{B}_{j,j',t,t'} \subseteq \mathcal{B}_n$, and $|\mathcal{B}_n| = O(\ell+mh)$. 
Then 
\begin{align*}
	\widehat{v}_{\diff} - \widehat{v}' 
		&= \widehat{v}_{\diff} - \left(\frac{n-mh-\ell}{n}\right)\widehat{v}' 
			- \left( \frac{mh+\ell}{n} \right) \widehat{v}' \\
		&= \left\{ \frac{1}{n}\sum_{t,t'=1}^{\ell} \sum_{j,j'=0}^m 
				\frac{K\left( \frac{t-t'}{\ell} \right)}{\ell-\left\vert t-t' \right\vert}
				d_jd_{j'} 
				\sum_{i\in\mathcal{B}_{j,j',t,t'}} 
					\widehat{Z}_{i}^{(\theta)}\widehat{Z}_{i - \{t-t' - h(j- j')\}}^{(\theta)} 
			\right\}\\
		& \qquad	- \left( \frac{mh+\ell}{n} \right) \widehat{v}' \\
		&=\left( \frac{1}{n}\sum_{i,i'\in\mathcal{B}_n} a_{i,i'} \widehat{Z}_i^{(\theta)}\widehat{Z}_{i'}^{(\theta)} 
			\right)
			- \left( \frac{mh+\ell}{n} \right) \widehat{v}' ,
\end{align*} 
where $a_{i,i'} = O(1)$ are some non-random weights for $i,i'\in\mathcal{B}_n$. 
By Lemma 8 of \citet{xiaowu2012}, 
we have 
\begin{eqnarray*}
	\left\| \frac{1}{n}\sum_{i,i'\in\mathcal{B}_n} a_{i,i'} \widehat{Z}_i^{(\theta)}\widehat{Z}_{i'}^{(\theta)}  \right\|
		\leq O\left( \frac{\ell+mh}{n} \right) 
\end{eqnarray*}
for $\theta\in\{\mu_0, \bar{X}\}$.
By Minkowski inequality, we have  
\begin{align*}
	\|\widehat{v}_{\diff} - \widehat{v}'\| 
	&\leq \left\| \frac{1}{n}\sum_{i,i'\in\mathcal{B}_n} a_{i,i'} \widehat{Z}_i^{(\theta)}\widehat{Z}_{i'}^{(\theta)} \right\|
			+ \left( \frac{\ell+mh}{n} \right) \|\widehat{v}'\| \\
	&\leq O\left( \frac{\ell+mh}{n} \right) + O\left( \frac{\ell+mh}{n} \right) \|\widehat{v}'\|.
\end{align*}
In view of Proposition~\ref{prop:Ktype_Stype_equiv}, 
we have $\|\widehat{v}_{\diff} - \widehat{v}\| =  O\{(\ell+mh)/n\}( 1+ \|\widehat{v}\|)$, 
which also implies  $\|\widehat{v}_{\diff} - \widehat{v}\| =  O\{(\ell+mh)/n\}( 1+ \|\widehat{v}_{\diff}\|)$.
It is easy to see that $\E(\widehat{\gamma}_k)=\gamma_k + O(1/n)$ for $|k|\leq \ell+mh$.
Finally, using Lemma \ref{lem:rough_var}, we know that 
{\small
\begin{align}
	\|\widehat{v}_{\diff}\| 
	&\leq v + \|\widehat{v}_{\diff} - v\|
	= v + \left\{ \left\vert \E(\widehat{v}_{\diff})-v \right\vert^2 + \Var(\widehat{v}_{\diff}) \right\}^{1/2} \nonumber\\
	&\leq v + \left[ 
			\left\{ 
			\sup_{t\in\mathbb{R}}|K(t)|
			\sum_{|k|\leq \ell+mh} \left\{ |\gamma_k^Z| + O(1/n)\right\} \sum_{|s|\leq m} |\delta_s|  \right\}^2 +
			O\left( \frac{\ell+mh}{n} \right) \right]^{1/2} \nonumber\\
	&= O(1).  \label{eqt:norm_vHatDiff}
\end{align}
}
So, $\|\widehat{v}_{\diff} - \widehat{v}\| =  O\{(\ell+mh)/n\}$.

\item Note that 
\begin{eqnarray}
	&&\sum_{|k|\leq \ell+mh} K_{\diff}\left( \frac{k}{\ell} \right) \nonumber \\
 		&&\quad= \sum_{|k|\leq \ell+mh} 
			\sum_{|s|\leq m}
				\mathbb{1}_{\left\{ \left\lceil -\frac{\ell+k}{\lambda\ell} \right\rceil
							\leq s \leq \left\lfloor \frac{\ell-k}{\lambda\ell} \right\rfloor\right\}}
					\delta_{|s|} K\left( \frac{k}{\ell} + \lambda s \right) . \label{eqt:sumKdiff_proof}
\end{eqnarray}
Since $K(t) = 0$ for $|t|\geq 0$, we may remove the floor and ceiling in (\ref{eqt:sumKdiff_proof}). 
Swapping the order of summations and using the change of variable $k = k' - \lambda \ell s = k' - hs$, we have 
\begin{align*}
	\sum_{|k|\leq \ell+mh} K_{\diff}\left( \frac{k}{\ell} \right)
		&= \sum_{|k|\leq \ell+mh} 
			\sum_{|s|\leq m}
				\mathbb{1}_{\left\{  -\frac{\ell+k}{\lambda\ell} 
							\leq s \leq \frac{\ell-k}{\lambda\ell} \right\}}
				\delta_{|s|} K\left( \frac{k}{\ell} + \lambda s \right) \nonumber\\
		&= \sum_{|s|\leq m} \delta_{|s|} \sum_{k' =-\ell-h(m-s)}^{\ell+h(m+s)} 
				\mathbb{1}_{\left\{  |k'|\leq \ell \right\}}
				 K\left( \frac{k'}{\ell} \right) \nonumber\\
		&= \left( \sum_{|s|\leq m} \delta_{|s|}\right)
				\left\{ \sum_{|k'| \leq \ell'} 
				 K\left( \frac{k'}{\ell} \right)   \right\} 
		= 0, 
\end{align*}
where the last equality follows from 
the identity $\sum_{|s|\leq m} \delta_{|s|} \equiv 0$;
see Problem 7.3 of \cite{brockwellDavis1991}. 
\item Denote 
	\[
		\underline{s}(t) = \left\lceil -\frac{1+t}{\lambda}\right\rceil
		\qquad \text{and} \qquad
		\overline{s}(t) = \left\lfloor \frac{1-t}{\lambda} \right\rfloor.
	\]
	So, $\underline{s}(0) = \lceil -1/\lambda \rceil$ and $\overline{s}(0) = \lfloor 1/\lambda \rfloor$.
		We consider two cases. 
		\begin{itemize}
			\item Case 1: $\lambda > 1$. We have $\underline{s}(0) = \overline{s}(0) = 0$.
					So, $K_{\diff}(0) = \delta_0 K(0) = 1$ because $\delta_0 = K(0) = 1$. 
			\item Case 2: $\lambda = 1$. We have $\underline{s}(0) = -1$ and $\overline{s}(0) = 1$.
					So, $K_{\diff}(0) = \delta_0 K(0) + \delta_1 K(1) + \delta_1 K(-1) = 1$
					because $K(\pm1)=0$. 
		\end{itemize}
\item  We consider the following two cases. 
\begin{itemize}
\item Case 1: $|t|< 1$. Since $\lambda \geq 2$, we have $\overline{s}(0)=\underline{s}(0)=0$. So, $K_{\diff}(t) = \delta_0 K(t) = K(t)$. 
\item Case 2: $t\in\{-1,1\}$. Since $K(\pm1)=0$, we  have $K_{\diff}(t) = K(t)$.
\end{itemize}
\end{enumerate}

\subsection{Proof of Theorem~\ref{thm:robustness}}\label{sec:proof_robustness} 
Recall that $X_i = \mu_i + Z_i$, where $\E(Z_i) = 0$ for each $i=1,\ldots, n$, 
and $D_t = \sum_{j=0}^m d_j X_{t-jh}$ for each $t=mh+1, \ldots,n$. 
Now expand $D_t$ into two parts: $D_t = D_t^{(\mu)} + D_t^{(Z)}$, where
\begin{align}\label{eqt:DDmuZ}
	D_t^{(\mu)} = \sum_{j=0}^m d_j \mu_{t-jh}
	\qquad \text{and}\qquad
	D_t^{(Z)} = \sum_{j=0}^m d_j Z_{t-jh}.
\end{align}
Since $D_t D_{t'} = D_t^{(\mu)} D_{t'}^{(\mu)}  + D_t^{(\mu)} D_{t'}^{(Z)} + D_t^{(Z)} D_{t'}^{(\mu)} + D_t^{(Z)} D_{t'}^{(Z)}$, we can decompose $\widehat{v}$  
into four parts:
\begin{align}
	\widehat{v} 
		&= \frac{1}{n}\sum_{|k| < \ell} K\left(\frac{k}{\ell}\right) 
			\sum_{i=mh+|k|+1}^n 
			\nonumber\\
		&\qquad\qquad\left( D_i^{(\mu)} D_{i-|k|}^{(\mu)}  + D_i^{(\mu)} D_{i-|k|}^{(Z)} + D_i^{(Z)} D_{i-|k|}^{(\mu)} + D_i^{(Z)} D_{i-|k|}^{(Z)}\right) \nonumber \\
		&= \widehat{v}^{(\mu\mu)} + \widehat{v}^{(\mu Z)} + \widehat{v}^{(Z\mu)} + \widehat{v}^{(ZZ)},\label{eqt:muZdecomp_vp}
\end{align}
where $\widehat{v}^{(\mu\mu)}$, $\widehat{v}^{(\mu Z)}$, $\widehat{v}^{(Z\mu)}$ and $\widehat{v}^{(ZZ)}$
are defined in the obvious way, i.e., 
\begin{align}\label{eqt:vhatp_ab}
	\widehat{v}^{(\alpha \beta)} = \frac{1}{n}\sum_{|k| < \ell} K\left(\frac{k}{\ell}\right) 
				\sum_{i=mh+|k|+1}^n
					D_i^{(\alpha)} D_{i-|k|}^{(\beta)} ,
\end{align}
where $\alpha, \beta \in\{\mu, Z\}$.
Next we calculate the expectation of each of the four terms in (\ref{eqt:muZdecomp_vp}). 
Since $\E(Z_i)=0$ for each $i$, we have $\E(\widehat{v}^{(\mu Z)}) = \E(\widehat{v}^{(Z\mu)}) = 0$. 
Since $\widehat{v}^{(\mu\mu)}$ is non-stochastic, we have $\E(\widehat{v}^{(\mu\mu)}) = \widehat{v}^{(\mu\mu)}$.
Note that $\Bias_{\mu}(\widehat{v}) = \E(\widehat{v}) - v $
and $\Bias_{0}(\widehat{v}) = \E( \widehat{v}^{(ZZ)} ) - v$. 
By Lemma~\ref{lem:size_vHatpmumu}, we have
\begin{align*}
	\Bias_{\mu}(\widehat{v})
		&=  \left\{\E(\widehat{v}^{(ZZ)}) - v\right\} + \E\widehat{v}^{(\mu Z)} + \E\widehat{v}^{(Z\mu)} + \E\widehat{v}^{(\mu\mu)}\\
		&=  \Bias_{0}(\widehat{v})
			+ \left\{ \kappa\ell\mathcal{V} + O\left( \frac{\ell}{n}\right) \right\} \mathbb{1}_{(m=0)} \\
		& \qquad 
		+ O\left[ \frac{\ell}{n} \left\{  (\ell\mathbb{1}_{m=0} +mh ) \mathcal{S}^2\mathcal{J}
					+  \frac{(\ell\mathbb{1}_{m=0} +mh )^2\mathcal{C}^2}{n}\right\} \right].  
\end{align*}

Next we derive the variance. 
Note that $\Var_{0}(\widehat{v}) = \Var( \widehat{v}^{(ZZ)} )$.
Since $\widehat{v}^{(\mu\mu)}$ is non-stochastic, 
$\|\widehat{v}^{(\mu\mu)}-\E\widehat{v}^{(\mu\mu)}\|=0$. 
Using Lemma~\ref{lem:Var_vHatp_muZ}, we have
\begin{align*}
	&\sqrt{ \Var_{\mu}(\widehat{v}) } \\
		&\quad= \left\| \widehat{v} - \E\widehat{v} \right\|  \nonumber\\
		&\quad\leq \left\| \widehat{v}^{(Z Z)}-\E\widehat{v}^{(Z Z)} \right\| 
				+\left\| \widehat{v}^{(\mu Z)} -\E\widehat{v}^{(\mu Z)}\right\|
				+ \left\| \widehat{v}^{(Z\mu)} - \E\widehat{v}^{(Z \mu)}\right\|
			\nonumber \\
		&\quad= \sqrt{ \Var_{0}(\widehat{v}) }
			+O\left\{ \frac{\ell}{\sqrt{n}} \left( \mathcal{C} +\mathcal{S} \mathcal{J}  \right)\right\} \mathbb{1}_{(m=0)}
			+ O\left\{ \frac{\ell}{\sqrt{n}} \left( \frac{mh\mathcal{C}}{n} + \mathcal{S} \left(\frac{mh\mathcal{J}}{n}\right)^{1/2} \right)\right\} .
			\label{eqt:var_mu_v}
\end{align*}
Thus, we obtain the desired result.

\subsection{Proof of Theorem~\ref{thm:consistency_finite}}\label{sec:proof_consistency_finite}
We begin with stating two preliminary findings for any $m>0$ and $1/\ell + (\ell+mh)/n = o(1)$. 
\begin{itemize}
	\item Denote the estimator $\widehat{v}$ 
			by $\widehat{v}^{(ZZ)}$ when the data $\{X_i\}$ are replaced by the noises $\{Z_i\}$.
			Then, $\Bias_0(\widehat{v}) = \Bias(\widehat{v}^{(ZZ)})$,
			$\Var_0(\widehat{v}) = \Var(\widehat{v}^{(ZZ)})$, and 
			$\MSE_0(\widehat{v}) = \MSE(\widehat{v}^{(ZZ)})$, 
			whose values do not depend on the mean function $\mu(\cdot)$, thus, the subscript ``$\mu$'' is dropped. 
			By Theorem \ref{thm:robustness} with $m>0$, $p=0$ and $\mathcal{J},\mathcal{S},\mathcal{C}\asymp 1$, 
			we have 
			\begin{equation}\label{eqt:maxR_bias_se_mu_ignore1}
				\Bias_{\mu}(\widehat{v}) = \Bias(\widehat{v}^{(ZZ)}) + R_{\bias}
				\quad \text{and} \quad
				\Var_{\mu}(\widehat{v})^{1/2} = \Var(\widehat{v}^{(ZZ)})^{1/2} + R_{\se},
			\end{equation}
			where 
			$R_{\bias} = O\left( \ell mh/n \right)$  
			and 
			$R_{\se} =  O\left( \ell mh/n^{3/2} \right) \ll R_{0,\bias}$. 
			Hence, 
			\begin{align}\label{eqt:maxR_bias_se_mu_ignore2}
				\max\left( |R_{\bias}| , |R_{\se}| \right) = O(\ell mh/n).
			\end{align}
			Consequently, the estimator is consistent only if $\ell mh/n = o(1)$.
			Recall that the optimal MSE of $\widehat{v}$ is of order $O(n^{-2q/(1+2q)})$.
			Hence, 
			the estimator is 
			rate optimal only if $\ell mh/n = O(n^{-q/(1+2q)})$.
	\item Let $\widehat{v}_{\diff}^{(ZZ)} = \sum_{|k| \leq \ell+mh} K_{\diff}( k/\ell ) \widehat{\gamma}_{k}^Z$.
			In view of Proposition~\ref{prop:effKernel} (\ref{item:repKdiff}), 
			$\|\widehat{v}^{(ZZ)} - \widehat{v}_{\diff}^{(ZZ)}\| = O\{(\ell+mh)/n\} = O(\ell mh/n)$.
			Hence,  
			it suffices to prove the results for $\widehat{v}_{\diff}^{(ZZ)}$ instead of $\widehat{v}^{(ZZ)}$.
			Write $\widehat{v}_{\diff}^{(ZZ)} = \sum_{|k|\leq \ell+mh} W_k \widehat{\gamma}_k^Z$, where
			$\widehat{\gamma}^Z_k = \sum_{i=|k|+1}^n Z_iZ_{i-|k|}/n$ and 
			\begin{align}\label{eqt:def_Wk}
				W_k = \sum_{s = \underline{s}_k}^{\overline{s}_k} \delta_{s} K\left( \frac{k+hs}{\ell} \right)  ,\quad
				\underline{s}_k = \left\lceil \frac{-\ell-k}{h} \right\rceil, \quad
				\overline{s}_k = \left\lfloor \frac{\ell-k}{h} \right\rfloor .
			\end{align}
			So, $W_0 \rightarrow 1$ is a necessary condition for $\MSE(\widehat{v}^{(ZZ)}) \rightarrow 0$; 
			see Remark~\ref{rem:necessaryCond} for a detailed explanation.  
\end{itemize}
Now, we are ready to prove the theorem in each case. 
Recall that, in this regime, we consider a fixed $m\in\mathbb{N}$.
\begin{enumerate}
	\item Assume $h/\ell\rightarrow 0$. 
			Since $\ell/h \rightarrow \infty$ and $m<\infty$, 
			we know $\ell/h > m$ when $\ell$ is large enough. 
			So, for all large enough $\ell$, Assumption~\ref{ass:Kernel_q} implies that
			\begin{align*}
				W_0 &= \sum_{|s|\leq m} \delta_s K\left( \frac{hs}{\ell}\right) 
					= \sum_{|s|\leq m} \delta_s \left\{ 1+B\left\vert \frac{hs}{\ell} \right\vert^q
							+ o\left( \left\vert\frac{hs}{\ell}\right\vert^q\right)  \right\} \\
					&=B\left( \frac{hm}{\ell} \right)^q \{1+o(1)\}\sum_{|s|\leq m} 
					\left\vert \frac{s}{m} \right\vert^q \delta_s,
			\end{align*}
			where the last line follows from the identity 
			$\sum_{|s|\leq m} \delta_s \equiv 0$. 
			Since $h/\ell\rightarrow 0$, $m<\infty$, and $\sum_{|s|\leq m} |\delta_s| <\infty$, 
			we have 
			\[
				| W_0 | 
				\leq B\left( \frac{hm}{\ell} \right)^q \{1+o(1)\}\sum_{|s|\leq m} |\delta_s|
				\rightarrow 0.
			\]
			Hence, $\widehat{v}$ cannot be consistent for $v$. 
	\item Assume $h/\ell \rightarrow \lambda_{\infty} \in(0,1)$. 
			Note that $\lambda = h/\ell$. So, 
			\begin{align}\label{eqt:W0}
				W_0 \sim 1+2\sum_{s=1}^{\lfloor \ell/h\rfloor} \delta_s K(sh/\ell) 
					\sim 1 + \Omega_{\infty}.  
			\end{align}
			Since $\Omega_{\infty}\neq 0$ 
			(Assumption~\ref{ass:d_K_orthogonalNot}), 
			we know that $W_0 \not\rightarrow 1$, 
			which implies that $\widehat{v}$ is inconsistent. 
	\item Assume $h/\ell \rightarrow 1$. 
			By assumption, $|h-\ell| = O(1)$, so 
			there is a large enough constant $C>0$ such that $|h - \ell| \leq C$ for all $n$. 
			After some simple algebras, we know that  
			$\overline{s}_k = 0$ and $\underline{s}_k = -1$
			if $C < k < \sqrt{\ell}$ and $\ell$ is large enough.
			In this case, 
			\[
				W_k 
					= K(k/\ell) + \delta_1 K\left( \frac{k+\ell-h}{\ell} -1 \right) .
			\]
			Using Assumptions~\ref{ass:Kernel_q} and \ref{ass:Kernel_q_tail}, 
			we have 
			\begin{align}\label{eqt:finite_m_Wk}
				W_k-1
					&\sim B \left( \frac{k}{\ell} \right)^{q} 
									- \delta_1 B' \left( \frac{k+\ell-h}{\ell} \right)^{q'}
					\sim \frac{B''_k}{\ell^{q\wedge q'}} , 
			\end{align}
			where 
			\[
				B''_k = B k^q \mathbb{1}_{(q\leq q')}- \delta_1 B' \left( k+\ell-h \right)^{q'}\mathbb{1}_{(q'\leq q)}.
			\]
			Consider the following two cases.
			\begin{itemize}
				\item Case (i): $\sup_{C<k<\sqrt{\ell}} |B''_k| \nrightarrow 0$.
						The weight $W_k$ is asymptotically equivalent to a kernel of 
							characteristic exponent $q\wedge q'$.
						Hence, $\Bias(\widehat{v}^{(ZZ)}) = O(1/\ell^{q\wedge q'})$; 
						see, e.g.,  Theorem 2 of \cite{chanyau2015_hoc}.
						By Lemma~\ref{lem:rough_var}, $\Var(\widehat{v}^{(ZZ)}) = O(\ell/n)$.
						Consequently, 
						\[
							\MSE(\widehat{v}^{(ZZ)}) = O(1/\ell^{2(q\wedge q')}) + O(\ell/n),
						\] 
						whose order is minimized to 
						$O(n^{-2(q\wedge q')/\{1+2(q\wedge q')\}})$
						when $\ell \asymp n^{1/\{1+2(q\wedge q')\}}$.
						In this case, $\ell m h /n= O(n^{(1-2q)/(1+2q)}) = O(n^{-q/(1+2q)})$.
						So, from (\ref{eqt:maxR_bias_se_mu_ignore1})--(\ref{eqt:maxR_bias_se_mu_ignore2}), 
						we obtain $\MSE_{\mu}(\widehat{v}) = O(n^{-2(q\wedge q')/\{1+2(q\wedge q')\}})$. 
				\item Case (ii): $\sup_{C<k<\sqrt{\ell}} |B''_k| \rightarrow 0$.  
						The weight $W_k$ is equivalent to a kernel of an infinite characteristic exponent.
						Then $\Bias(\widehat{v}^{(ZZ)}) = o(1/\ell^{q\wedge q'})$
						and $\Var(\widehat{v}^{(ZZ)}) = O(\ell/n)$ (Lemma~\ref{lem:rough_var}). 
						Note that we can represent $\Bias(\widehat{v}^{(ZZ)}) = \epsilon/\ell^{q\wedge q'}$, 
						where $\epsilon = \epsilon_n \rightarrow 0$ is an unknown sequence
						depending on $\{\gamma_k\}$.
						In this case, 
						\[
							\MSE(\widehat{v}^{(ZZ)}) = \epsilon^2/\ell^{2(q\wedge q')} + O(\ell/n).
						\]
						It is minimized when $\ell \asymp (\epsilon^2 n)^{1/\{1+2(q\wedge q')\}}$, 
						which depends on the unknown $\epsilon$.
						Unless we make additional assumption on $\{\gamma_k\}$,
						it is not possible to utilize the optimal $\ell$.
						In this case, the best possible unknown-free bandwidth is 
						$\ell \asymp n^{1/\{1+2(q\wedge q')\}}$. 
						Using similar arguments as in case (i), 
						we know that the resulting best possible MSE is  
						$\MSE_{\mu}(\widehat{v})=O(n^{-2(q\wedge q')/\{1+2(q\wedge q')\}})$. 
			\end{itemize}
			
	\item Assume $h/\ell \rightarrow \lambda_{\infty} \in(1, \infty)$. 
			Suppose that $\ell$ is large enough. 
			We consider all $k$ satisfying 
			$0\leq k\leq (\lambda-1)h/(2\lambda)\sim (\lambda_{\infty}-1)\lambda_{\infty}\ell/2$. 
			Then, for large enough $\ell$, 
			we have $\overline{s}_k=\underline{s}_k=0$, which implies $W_k = K(k/\ell)$. 
			Hence, Assumption~\ref{ass:Kernel_q} implies that $\Bias(\widehat{v}^{(ZZ)}) = O(1/\ell^q)$. 
			By Lemma~\ref{lem:rough_var} again, we have $\Var(\widehat{v}^{(ZZ)}) = O(\ell/n)$. 
			Hence, $\MSE(\widehat{v}^{(ZZ)}) = O(1/\ell^{2q}+\ell/n)$, 
			whose order is minimized to $O(n^{-2q/(1+2q)})$ when $\ell \asymp n^{1/(1+2q)}$.  
			In this case, $\ell m h /n= O(n^{(1-2q)/(1+2q)}) = O(n^{-q/(1+2q)})$.
			So, from (\ref{eqt:maxR_bias_se_mu_ignore1})--(\ref{eqt:maxR_bias_se_mu_ignore2}), 
			we obtain $\MSE_{\mu}(\widehat{v}) = O(n^{-2q/(1+2q)})$. 
	\item Assume $h/\ell\rightarrow \infty$. 
			For $0\leq k \leq \ell$ and a large enough $\ell$, we have $W_k = K(k/\ell)$. 
			Assumption~\ref{ass:Kernel_q} implies that $\Bias(\widehat{v}^{(ZZ)}) =  O(1/\ell^q)$. 
			By Lemma~\ref{lem:rough_var}, we have $\Var(\widehat{v}^{(ZZ)}) = O(\ell/n)$. 
			Hence, $\MSE(\widehat{v}^{(ZZ)}) = O(1/\ell^{2q} + \ell/n)$, 
			whose order is minimized to $O(n^{-2q/(1+2q)})$ when 
			$\ell\asymp n^{1/(1+2q)}$. 
			From (\ref{eqt:maxR_bias_se_mu_ignore1})--(\ref{eqt:maxR_bias_se_mu_ignore2}), 
			$\MSE_{\mu}(\widehat{v}) = O(n^{-2q/(1+2q)})$ 
			only if $\ell m h/n = O(n^{-q/(1+2q)})$, or equivalently $h \lesssim n^{q/(1+2q)}$.
			However, it is possible if $q>1$ because $h\gg \ell$ is assumed in this case.
			If $q=1$, $\widehat{v}$ is rate suboptimal. 
\end{enumerate}
It is easy to see that the upper bounds of MSE in cases 3, 4 and 5(b) are achievable by
setting 
$\mu_i = i/n + \mathbb{1}(i/n\leq 1/2)$ and 
$X_i = \theta \varepsilon_{i-1} + \varepsilon_{i}$ for $i=1, \ldots, n$, 
where $\varepsilon_i$ follow $\Normal(0,1)$ independently and $\theta>0$. 
Consequently, we obtain the desired results.

\begin{remark}\label{rem:necessaryCond}
This remark explains why $W_0 \rightarrow 1$ is a necessary condition for 
$\MSE(\widehat{v}^{(ZZ)}) \rightarrow 0$. 
First, $\widehat{v}^{(ZZ)}$ and $\widehat{v}_{\diff}^{(ZZ)}$ are asymptotically equivalent in $\mathcal{L}^2$
according to Proposition~\ref{prop:effKernel}.
Thus, we may replace $\widehat{v}^{(ZZ)}$ by $\widehat{v}^{(ZZ)}_{\diff}$.
Then, we prove the statement by contradiction. 
Suppose that $W_0 \not\rightarrow 1$. 
We are going to prove that $\widehat{v}^{(ZZ)}_{\diff}$ is not necessarily $\mathcal{L}^2$ consistent for $v$.
Consider the following counterexample. 

Let $Z_1, \ldots, Z_n$ be independent $\Normal(0,1)$ random variables. 
In this case, 
$\gamma_0^Z = \E(Z_i^2) = 1$ and $\gamma_k^Z = \E(Z_iZ_{i-|k|}) = 0$ for $k\neq 0$. 
So, $v = \sum_{k\in\mathbb{Z}}\gamma_k^Z = 1$ and 
\[
	\E(\widehat{v}^{(ZZ)}_{\diff}) 
	= \sum_{|k|\leq \ell+mh} W_k \E\left( \frac{1}{n} \sum_{i=|k|+1}^n Z_i Z_{i-|k|} \right)
	= \sum_{|k|\leq \ell+mh} W_k \left(\frac{1}{n} \sum_{i=|k|+1}^n \gamma_k^Z \right)
	= W_0.
\]
Hence, 
the bias is $\E(\widehat{v}^{(ZZ)}_{\diff}) - v = W_0 - 1 \not\rightarrow 0$, 
which implies $\mathcal{L}^2$ inconsistency. 
Then we can conclude that $W_0 \rightarrow 1$ is necessary for $\mathcal{L}^2$ consistency. 
\end{remark}

\subsection{Proof of Theorem~\ref{thm:consistency_infinite}}\label{sec:proof_consistency_infinite}
Assume $m \rightarrow \infty$. 
Note that, in this case, (\ref{eqt:maxR_bias_se_mu_ignore1}), 
(\ref{eqt:maxR_bias_se_mu_ignore2}) and (\ref{eqt:def_Wk})
are still valid. 
\begin{enumerate}
	\item Assume $h/\ell \rightarrow 0$. 
			(a) Similar to (\ref{eqt:W0}), we have $W_0 \sim 1 + \Omega_{\infty} \not\rightarrow 1$
			by Assumption~\ref{ass:d_K_orthogonalNot}.
			So, $\widehat{v}$ is inconsistent. 
			(b) By Assumption~\ref{ass:Kernel_q} and the assumptions made in part 1(b), 
			we have 
			\begin{align}
				|W_1 -1 |
					&= \left\vert K(1/\ell)-1  + \sum_{s=1}^{\overline{s}_k} \delta_s K\left( \frac{1+hs}{\ell} \right)
						+ \sum_{s=-\underline{s}_k}^{-1} \delta_s K\left( \frac{1+hs}{\ell} \right) \right\vert \nonumber \\
					&\geq O\left( \frac{1}{\ell^q} + \frac{\ell}{mh} \right), \label{eqt:diffW1and1}
			\end{align}
			which implies that $|\Bias(\widehat{v}^{(ZZ)})| \geq O\{1/{\ell^q} + {\ell}/({mh})\}$ 
			at least for some time series.
			Clearly, if $\ell/h \gtrsim m$, then $\Bias(\widehat{v}^{(ZZ)})\not\rightarrow 0$, 
			which means that $\widehat{v}^{(ZZ)}$ is not always $\mathcal{L}^2$ consistent. 
			Now, consider $\ell/h \ll m$. 
			By (\ref{eqt:maxR_bias_se_mu_ignore1}) and (\ref{eqt:diffW1and1}), 
			we know that 
			\begin{align}\label{eqt:Bias_infinite_m_regime1}
				\Bias_{\mu}(\widehat{v}) \gtrsim \frac{1}{\ell^q} + \frac{\ell}{mh} + \frac{\ell mh}{n}.
			\end{align}
			Minimizing the right-hand side of (\ref{eqt:Bias_infinite_m_regime1}) with respect to $\ell$ and $mh$, 
			we obtain $\Bias_{\mu}(\widehat{v}) \gtrsim n^{-q/(2+2q)} \gg n^{-q/(1+2q)}$.
			Hence, $\widehat{v}$ is rate suboptimal.
	\item
			Assume $h/\ell\rightarrow \lambda\in(0,1)$. 
			(a) The proof of Theorem~\ref{thm:consistency_finite} (2) also applies here. 		
			(b) Similar to part (1), we know that $|W_1 - 1|\geq O(1/\ell^q+1/m)$. 
			By (\ref{eqt:maxR_bias_se_mu_ignore1}),
			we have 
			\begin{align}\label{eqt:Bias_infinite_m_regime2}
				\Bias_{\mu}(\widehat{v}) \gtrsim \frac{1}{\ell^q} + \frac{1}{m} + \frac{\ell^2m }{n}.
			\end{align}
			Minimizing the right-hand side of (\ref{eqt:Bias_infinite_m_regime2}) with respect to $\ell$ and $m$, 
			we obtain $\Bias_{\mu}(\widehat{v}) \gtrsim n^{-q/(2+2q)} \gg n^{-q/(1+2q)}$.
			Hence, $\widehat{v}$ is rate suboptimal.
	\item Assume $h/\ell \rightarrow 1$. 	
			Note that, (\ref{eqt:finite_m_Wk}) is also true. 
			However, in this case, $\delta_1\asymp 1/m$ as $m\rightarrow \infty$ according to 
			Assumptions~\ref{ass:d_uncorrelatedDiff}.
			Using Assumptions~\ref{ass:Kernel_q} and \ref{ass:Kernel_q_tail}, 
			we have 
			\begin{align*} 
				W_k-1
					\sim \frac{B'''_k}{\min(\ell^q, m \ell^{q'})} , 
			\end{align*}
			where 
			\[
				B'''_k = B k^q \mathbb{1}_{(\ell^q\leq m\ell^{q'})}
					- m\delta_1 B' \left( k+\ell-h \right)^{q'}\mathbb{1}_{(\ell^q\geq m\ell^{q'})}.
			\]

			Similar to the proof of Theorem~\ref{thm:consistency_finite} (3), we 
			consider the following two cases.
			\begin{itemize}
				\item Case (i): $\sup_{C<k<\sqrt{\ell}} |B'''_k| \nrightarrow 0$.  
						By (\ref{eqt:maxR_bias_se_mu_ignore1}), (\ref{eqt:maxR_bias_se_mu_ignore2}) 
						and Lemma~\ref{lem:rough_var}, we have  
						\begin{align}\label{eqt:MSE_infinite_m_regime3}
							\MSE_{\mu}(\widehat{v}) 
								= O\left( \frac{1}{\min(\ell^{2q}, m^2 \ell^{2q'}) } + \frac{\ell}{n} \right) + 
										O\left( \frac{\ell^4m^2}{n^2} \right).
						\end{align}
						\begin{itemize}
							\item If $\ell^q \gg m \ell^{q'}$, then, using (\ref{eqt:MSE_infinite_m_regime3}), 
									we have 
									$\MSE_{\mu}(\widehat{v}) 
											\gg {1}/{ \ell^{2q} } + {\ell}/{n}$.
									Minimizing the lower bound with respect to $\ell$, we obtain 
									$\MSE_{\mu}(\widehat{v}) \gg n^{-2q/(1+2q)}$. 
							\item If $\ell^q \lesssim m \ell^{q'}$, then 
								the first term of (\ref{eqt:MSE_infinite_m_regime3}) is minimized to 
								$O(n^{-2q/(1+2q)})$ when $\ell \asymp n^{1/(1+2q)}$.
								Consider the second term in (\ref{eqt:MSE_infinite_m_regime3}) with 
								$\ell \asymp n^{1/(1+2q)}$, 
								we know that $\ell^4m^2/n^2 \lesssim n^{-2q/(1+2q)}$
								if only if $q>1$ and $1\ll m \lesssim n^{(q-1)/(1+2q)}$.
								Together with the constraint $\ell^q \lesssim m \ell^{q'}$,
								we also know that $m \gtrsim n^{(q-q')/(1+2q)}$.
								Otherwise, $\widehat{v}$ is rate suboptimal. 
						\end{itemize}
				\item Case (ii): $B=\delta_1B'$.  
						Similar to the proof of Theorem~\ref{thm:consistency_finite}(3), 
						unless we make additional assumption on $\{\gamma_k\}$, 
						the best possible MSE is $O(n^{-2q/(1+2q)})$ 
						when $\ell \asymp n^{1/(1+2q)}$, 
						$1\ll m \lesssim n^{(1-q)/(1+2q)}$ and $q>1$.
						Otherwise, $\widehat{v}$ is rate suboptimal. 
			\end{itemize}			
	\item Assume $h/\ell \rightarrow \lambda_{\infty} \in(1, \infty)$.
			Similar to the proofs of Theorem~\ref{thm:consistency_finite} (4) and 
			Theorem~\ref{thm:consistency_infinite} (3), We have 
			\[
				\MSE_{\mu}(\widehat{v}) 
					= O\left( \frac{1}{\ell^{2q}} + \frac{\ell}{n} \right) + 
							O\left( \frac{\ell^4m^2}{n^2} \right),
			\]
			whose order is minimized to $O(n^{-2q/(1+2q)})$ when $\ell \asymp n^{1/(1+2q)}$, 
			$1\ll m \lesssim n^{(q-1)/(1+2q)}$ and $q>1$.
			Otherwise, $\widehat{v}$ is rate suboptimal. 
	\item Assume $h/\ell \rightarrow \infty$.
			By (\ref{eqt:maxR_bias_se_mu_ignore1}), (\ref{eqt:maxR_bias_se_mu_ignore2}) 
			and Lemma~\ref{lem:rough_var}, we have  
			\begin{align}\label{eqt:MSE_infinite_m_regime5}
				\MSE_{\mu}(\widehat{v}) 
					= O\left( \frac{1}{\ell^{2q}} + \frac{\ell}{n} \right) + 
							O\left( \frac{\ell^2m^2h^2}{n^2} \right),
			\end{align}
			where the first term is minimized to 
			$O(n^{-2q/(1+2q)})$ when $\ell \asymp n^{1/(1+2q)}$.
			Consider the second term in (\ref{eqt:MSE_infinite_m_regime5}) with 
			$\ell \asymp n^{1/(1+2q)}$, 
			we know that $\ell^2 m^2 h^2/n^2 \lesssim n^{-2q/(1+2q)}$
			if $q>1$ and $n^{1/(1+2q)} \ll mh \lesssim n^{q/(1+2q)}$.
			Otherwise, $\widehat{v}$ is rate suboptimal. 
\end{enumerate}
It is easy to see that the upper bounds of $\MSE_{\mu}(\widehat{v})$ in cases 3(b), 4(b) and 5(b) are achievable by
setting 
$\mu_i = i/n + \mathbb{1}(i/n\leq 1/2)$ and 
$X_i = \theta \varepsilon_{i-1} + \varepsilon_{i}$ for $i=1, \ldots, n$, 
where $\varepsilon_i$ follow $\Normal(0,1)$ independently and $\theta>0$.

\subsection{Proof of Theorem~\ref{thm:bias}}\label{sec:proof_bias} 
Let $\widehat{v}_{\diff}^{(ZZ)} = \sum_{|k| \leq \ell+mh} K_{\diff}( k/\ell ) \widehat{\gamma}_{k}^Z$.
In view of Proposition~\ref{prop:effKernel}, we know 
\begin{align} 
	\left\vert \E(\widehat{v}-\widehat{v}_{\diff}^{(ZZ)}) \right\vert
		&\leq \| \widehat{v}-\widehat{v}_{\diff}^{(ZZ)} \| 
		\leq O\left( \frac{\ell}{n} \right) . \label{eqt:biasProof0}
\end{align}
Now, we consider the bias of $\widehat{v}^{(ZZ)}_{\diff}$. 
\begin{align}
	\E( \widehat{v}^{(ZZ)}_{\diff} ) - v
		&= \left\{ \sum_{|k| \leq \ell+mh} K_{\diff}\left(\frac{k}{\ell} \right) \frac{1}{n} \sum_{i=|k|+1}^n \E (Z_i Z_{i-|k|}) \right\}  - v \nonumber\\
		&= \left\{ \sum_{|k| \leq \ell+mh} K_{\diff}\left(\frac{k}{\ell} \right) \left( 1 - \frac{|k|}{n}  \right)\gamma_k  \right\} - \sum_{k\in\mathbb{Z}} \gamma_k \nonumber\\
		&= \left[\sum_{|k|\leq \ell} \left\{ K_{\diff}\left(\frac{k}{\ell} \right) - 1\right\} \gamma_k\right] - I_1 + I_2 ,\label{eqt:biasProof1}
\end{align}
where
\begin{align*}
	I_1 &= \sum_{|k| \leq \ell+mh} K_{\diff}\left(\frac{k}{\ell} \right)  \frac{|k|}{n} \gamma_k, \\
	I_2 &= \sum_{\ell<|k|\leq \ell+mh} K_{\diff}\left(\frac{k}{\ell} \right)  \gamma_k
			- \sum_{|k|>\ell} \gamma_k.
\end{align*}
Note that $u_q<\infty$ implies $u_1<\infty$; and 
$K_{\diff}(k/\ell) = O(1)$ for all $k$. 
So, we have $I_1 = O(1/n)$.
Using $u_q<\infty$, we also have
\begin{equation*}\label{eqt:biasProof2}
	\left\vert I_2 \right\vert
		\leq O(1)\sum_{|k|>\ell}\frac{|k|^q}{\ell^q} |\gamma_k| 
		= o(1/\ell^q) .
\end{equation*}			
Using (\ref{eqt:biasProof0}) and (\ref{eqt:biasProof1}), we have 
\begin{align}
	\Bias( \widehat{v}) 
		&= \E( \widehat{v}^{(ZZ)}_{\diff} ) - v + \E(\widehat{v}-\widehat{v}_{\diff}^{(ZZ)}) \nonumber \\
		&=  \sum_{|k|\leq \ell} \left\{ K_{\diff}\left(\frac{k}{\ell} \right) - 1\right\} \gamma_k
			+ o(1/\ell^q)
			+ O(\ell/n).
			\label{eqt:biasProof_result}
\end{align}

Next, we note two properties of $K(\cdot)$ due to Assumption~\ref{ass:Kernel_q}.
First, 
$K(k/\ell)-1 = B|k/\ell|^q + |k/\ell|^q e_k$ 
and $e_k \rightarrow 0$ as $\ell \rightarrow\infty$ for each $k=1, \ldots, \sqrt{\ell}$; and 
So, $\sup_{1\leq k\leq \sqrt{\ell}} |e_k| = o(1)$.
Second, there is a large enough $\overline{B}\in(0, \infty)$ such that 
$|K(t)-1| \leq \overline{B}|t|^q$ for each $t \in [-1,1]$.

Then, by Proposition~\ref{prop:effKernel} 
(\ref{item:KdiffProperty}) and (\ref{item:matching}) with $\lambda\geq 2$ and $0\leq m< \infty$, we have 
\begin{align} 
	&\sum_{|k|\leq \ell} \left\{ K_{\diff}\left(\frac{k}{\ell}\right) -1 \right\} \gamma_k \nonumber\\
	&\qquad= \sum_{|k|\leq \sqrt{\ell}} \left\{ K\left(\frac{k}{\ell}\right) -1 \right\} \gamma_k
			+ \sum_{\sqrt{\ell} < |k|\leq \ell} \left\{ K\left(\frac{k}{\ell}\right) -1 \right\} \gamma_k \nonumber\\
	&\qquad= 2 \sum_{1\leq k\leq \sqrt{\ell}} \left\{ B\left( \frac{k}{\ell}\right)^q + \left( \frac{k}{\ell}\right)^q e_k \right\} \gamma_k 
			+ 2\sum_{\sqrt{\ell} < k\leq \ell} \left\{ K\left(\frac{k}{\ell}\right) -1 \right\} \gamma_k \nonumber\\
	&\qquad= \frac{Bv_q}{\ell^q} + o(1/\ell^q)
		+ I_3 + I_4, \nonumber \\\label{eqt:Bias_proof_2}
\end{align}
where 
\[
	I_3 =  \frac{2}{\ell^q} \sum_{1\leq k\leq \sqrt{\ell}} e_k k^q\gamma_k 
	\qquad \text{and} \qquad
	I_4 = 2\sum_{\sqrt{\ell} < k\leq \ell} \left\{ K\left(\frac{k}{\ell}\right) -1 \right\} \gamma_k .
\]
Using $u_q<\infty$, we have 
\begin{align*} 
	|I_3|
	&\leq \frac{2}{\ell^q} \left(\sup_{1\leq k\leq \sqrt{\ell}} |e_k| \right) \sum_{1\leq k\leq \sqrt{\ell}}  k^q|\gamma_k| 
	= \frac{1}{\ell^q} o(1) u_q
	= o(1/\ell^q) , \\ 
	|I_4|
	&\leq 2 \sum_{\sqrt{\ell} < k\leq \ell} \left\vert K\left(\frac{k}{\ell}\right) -1 \right\vert |\gamma_k| 
	\leq 2 \sum_{\sqrt{\ell} < k\leq \ell} \overline{B}\left\vert \frac{k}{\ell} \right\vert^q |\gamma_k| 
	\leq o(1/\ell^q). 
\end{align*}
Thus, $\Bias(\widehat{v}) = {B v_q}/{\ell^q} + o(1/\ell^q) + O(\ell/n)$ 
in view of (\ref{eqt:biasProof_result}).

\subsection{Proof of Theorem~\ref{thm:var}}\label{sec:proof_var} 
By Lemma~\ref{lem:blocking_vHatp}, we have 
\begin{align}\label{eqt:VarvHat_proof}
	\| \widehat{v}' -\E(\widehat{v}') \|^2
		= \Var(\widehat{v}') 
		= \frac{4\ell v^2}{n}\left\{ \int_{0}^{1} K^2(t)\, \dd t  \right\} 
			\Delta_m + o(\ell/n).
\end{align}
Similar to (\ref{eqt:biasProof0}), we have, by Proposition~\ref{prop:Ktype_Stype_equiv}, that 
\begin{align*}
	\| \widehat{v} -\E(\widehat{v})\| 
		&\leq \| \widehat{v}' -\E(\widehat{v}')\| 
				+ \| \widehat{v} - \widehat{v}'\| 
				+ |\E\widehat{v} - \E\widehat{v}'| \\
		&\leq \| \widehat{v}' -\E(\widehat{v}') \| + 2\| \widehat{v} - \widehat{v}'\| \\
		&\leq \| \widehat{v}' -\E(\widehat{v}') \| + O\left( \frac{\ell}{n} \right).  
\end{align*}
Hence, using (\ref{eqt:VarvHat_proof}) and $\ell/n=o(1)$, we have
\begin{align*}
	\Var(\widehat{v})
		&= \| \widehat{v} -\E(\widehat{v})\|^2 \\
		&= \frac{4\ell v^2}{n}\left\{ \int_{0}^{1} K^2(t)\, \dd t  \right\} 
			\Delta_m
			+ o(\ell/n).
\end{align*}

\subsection{Proof of Corollary~\ref{coro:optim}}\label{sec:proof_optim} 
For (\ref{coro:optim-1}), by Theorem~\ref{thm:bias} (\ref{thm:bias-2}) and Theorem~\ref{thm:var}, we have 
\begin{align*}
	\MSE(\widehat{v}) 
		&\sim  \frac{B^2 v_q^2}{\ell^{2q}} + \frac{4A\ell v^2 \Delta_m}{n}.
\end{align*}
The MSE is optimal when the above two terms are of the same order, 
i.e., $O(1/\ell^{2q}) = O(\ell/n)$.
It gives $\ell \asymp n^{1/(1+2q)}$.
Let $\ell = \phi n^{1/(1+2q)}$ for some $\phi\in\mathbb{R}^+$.
Then the MSE satisfies  
\begin{align}\label{eqt:limitingMSE_for_optimize}
	n^{2q/(1+2q)}\MSE(\widehat{v}) /v^2
		\rightarrow \left( B^2 \eta_q^2 \right) \phi^{-2q} + \left( 4A \Delta_m \right) \phi,
\end{align}
where $\eta_q = |v_q/v|$.
Minimizing the limit in (\ref{eqt:limitingMSE_for_optimize}) with respect to $\phi$, 
we obtain $\phi = \left\{ (q \eta_q^2 B^2)/(2A\Delta_m) \right\}^{1/(1+2q)}.$
Thus, the result follows. 

For (2), we put the optimal $\ell^{\star}$ into the estimator. The resulting MSE is 
\[
	n^{2q/(1+2q)} \MSE(\widehat{v}) 
		\rightarrow (1+2q) \left\{ \left( \frac{2A\Delta_m}{q} \right)^q |B| \eta_q^2 \right\}^{2/(1+2q)},
\]
which is monotonically increasing with $\Delta_m$. 
Note that 
(i) $\Delta_m = 1+ 2\sum_{s=1}^{m} \delta_s^2$ 
because $\delta_0 = 1$ and $\delta_s = \delta_{-s}$ for all $s$;  and 
(ii) $\sum_{s=1}^m \delta_s = -1/2$ because 
$\sum_{|s|\leq m} \delta_s = 0$ (see Problem 7.3 of \citet{brockwellDavis1991}).
So, it suffices to minimize 
\[
	\sum_{s=1}^{m} \delta_s^2 = \sum_{s=1}^{m-1}\delta_{s}^2 + \left( -\frac{1}{2} - \sum_{s=1}^{m-1}\delta_s\right)^2.
\] 
Using calculus, we know that the minimum occurs when $\delta_1 = \ldots = \delta_{m} = -1/(2m)$. 
The solution for $d_0, \ldots, d_m$ exist because 
we could solve for $d_0, \ldots, d_m$ by using the innovation algorithm
(Definition 8.3.1 of \citet{brockwellDavis1991}).
It completes the proof.

\subsection{Proof of Corollary~\ref{coro:vpHat}}\label{sec:proof_vpHat}
Let $H:[-1,1]\rightarrow\mathbb{R}$ be any function. 
Denote 
\[
	\widehat{v}(H) = \sum_{|k| < \ell} H({k}/{\ell}) \widehat{\gamma}_k^D.
\] 
For any fixed kernel $K$, define $K_p(t) = |t|^p K(t)$.
Then 
\[
	\widehat{v} = \widehat{v}(K)
	\qquad \text{and} \qquad
	\widehat{v}_p/\ell^p = \widehat{v}(K_p).
\] 
Hence, the properties of $\widehat{v}_p/\ell^p$ and $\widehat{v}$ can be derived similarly 
with some trivial modifications. 
We detail the steps below.

\begin{itemize}
	\item (Preliminary results) 
			Using similar techniques in (\ref{eqt:norm_vHatDiff}), 
			we know that $\|\widehat{v}_p\| = O(1)$.
			Similar to Propositions~\ref{prop:Ktype_Stype_equiv} and \ref{prop:effKernel}, we have 
			\begin{align}
				\|\widehat{v}_p - \widehat{v}'_p\| 
				&= O\left( \frac{\ell^{1+p}}{n} \right),  
				\label{eqt:vHatp_approx1}\\
				\|\widehat{v}_p - \widehat{v}_{p,\diff}\|_2 
				&= O\left( \frac{\ell^{1+p}}{n} \right).  
				\label{eqt:vHatp_approx2}
			\end{align}
			where $\widehat{v}'_p$ and $\widehat{v}_{p,\diff}$
			are the subsampling form and differencing-kernel form of $\widehat{v}_p$, respectively, i.e., 
			\begin{align*} 
				\widehat{v}'_p
				&= \frac{1}{|\mathcal{I}|} \sum_{i\in\mathcal{I}} 
					\sum_{t,t'\in\Lambda_i} \frac{K(|t-t'|/\ell)|t-t|^p}{\ell - |t-t'|}
								D_t D_{t'},\\
				\widehat{v}_{p,\diff} 
				&= \sum_{|k| \leq \ell+mh} K_{\diff}\left( \frac{k}{\ell}\right) |k|^p \widehat{\gamma}_{k}^X.
			\end{align*}
	\item (Bias) We follow the proof of Theorem~\ref{thm:bias}. Similar to (\ref{eqt:biasProof_result}), 
			we have 
			\begin{align}
				\Bias( \widehat{v}_p) 
					&=  \sum_{|k|\leq \ell} \left\{ K_{\diff}\left(\frac{k}{\ell} \right) - 1\right\} |k|^p\gamma_k 	
					+ o(1/\ell^q)
						+ O(\ell^{1+p}/n). 
						\label{eqr:biasProof_result-general}
			\end{align}
			Similar to (\ref{eqt:Bias_proof_2}), we have 
			\begin{align*}
				\sum_{|k|\leq \ell} \left\{ K_{\diff}\left(\frac{k}{\ell} \right) - 1\right\} |k|^p\gamma_k 
				= \frac{B  v_{p+q}}{\ell^q} + o(1/\ell^q) .
			\end{align*}			 
			Thus, $\Bias( \widehat{v}_p) = B v_{p+q}/\ell^q + o(1/\ell^q) + O(\ell^{1+p}/n)$.  
	\item (Variance) By (\ref{eqt:VarvHat_proof}) with kernel $K_p$, we have  
			\begin{align*}
				\Var(\widehat{v}_p/\ell^p)
					&= \frac{4\ell v^2}{n}\left\{ \int_{0}^{1} |t|^{2p}K^2(t)\, \dd u  \right\} 
						\Delta_m
						+ o(\ell/n).
			\end{align*}
			Multiplying both side by $\ell^{2p}$, we obtain the $\Var(\widehat{v}_p)$.
			Adding the squared bias and the variance, we have the desired expression of the MSE. 
\end{itemize}

\subsection{Proof of Corollary \ref{coro:robustness_general}}\label{sec:proof_robustness_general}
The proof is similar to Section \ref{sec:proof_robustness}.
For a general $p$,
we can derive the the robustness properties of $\widehat{v}(K_p)$. 
Next,
we recall the representation $\widehat{v}_p = \ell^p \times \widehat{v}(K_p)$, 
where $K_p(t) = |t|^p K(t)$. 
Hence, the desired result follows straightforwardly.

\subsection{Proof of Corollary~\ref{coro:prewhite}}\label{sec:proof_prewhite}
According to (\ref{eqt:removeJump})--(\ref{eqt:removeTrend}), 
we denote $X_i^{\ddag} = X_i - \mu^{\ddag}_i$, 
where $\mu^{\ddag}_i = s^{\ddag}_i + c^{\ddag}_i$,
\begin{eqnarray}
	s^{\ddag}_i 
		&=& \sum_{\substack{j\in\{0,\ldots, N\} : t_j\leq i}} \left[X_{t_j}-X_{t_j-1}\right]_{-M}^M, \label{eqt:def_s_ddag}\\
	c^{\ddag}_i
		&=&\sum_{j=0}^{N} 
			\left\{ \widehat{\beta}_{j,0} + \widehat{\beta}_{j,1} (i-t_j) \right\}
			\mathbb{1}(t_{j}\leq i \leq t_{j+1}-1). \nonumber  
\end{eqnarray}
Also recall that $X_i = \mu_i + Z_i$. So, $X_i = Z_i + \eta_i$, 
where $\eta_i = \mu_i - c^{\ddag}_i - s^{\ddag}_i$.
Define $D_t^{\eta},D_t^{c^{\ddag}},D_t^{s^{\ddag}}$ as in (\ref{eqt:DDmuZ})
with $\{\mu_i\}$ being replaced by $\{\eta_i\}$, $\{c^{\ddag}_i\}$, $\{s^{\ddag}_i\}$, 
respectively. 
Similar to (\ref{eqt:vhatp_ab}), denote 
\[
	\widehat{v}^{(\alpha \beta)}_p = \frac{1}{n}\sum_{|k| < \ell} |k|^p K\left(\frac{k}{\ell}\right) 
				\sum_{i=mh+|k|+1}^n
					D_i^{(\alpha)} D_{i-|k|}^{(\beta)} ,
\]
for $\alpha, \beta \in\{ \mu, Z, \eta, c^{\ddag}, s^{\ddag}\}$.
By Minkowski inequality, 
\begin{align*}
	\|\widehat{v}_p^{\ddag} - v_p\| 
		&\leq \|\widehat{v}_p - v_p\| + 
			2\| \widehat{v}^{(Z\eta)}_p\| + 
			\|\widehat{v}^{(\eta\eta)}_p\| \\
		&\leq \|\widehat{v}_p - v_p\| + 
			2\left( \| \widehat{v}^{(Z\mu)}_p\| + \| \widehat{v}^{(Zc^{\ddag})}_p\| + \| \widehat{v}^{(Zs^{\ddag})}_p\| \right)
			+ \|\widehat{v}^{(\eta\eta)}_p\| .
\end{align*}
We find each of the above terms one by one. 
\begin{itemize}
\item (Bound for $ \| \widehat{v}^{(Z\mu)}_p  \|$) A trivial modification of Lemma~\ref{lem:Var_vHatp_muZ} implies that 
\[
	\left\| \widehat{v}^{(Z\mu)}_p \right\| 
		= O\left\{ \frac{\ell^{1+p}}{\sqrt{n}} \left( \frac{\mathcal{C}\ell}{n} + \mathcal{S} \left(\frac{\ell\mathcal{J}}{n}\right)^{1/2} \right)\right\}
		= o\{n^{-q/(1+2p+2q)}\}.
\]
\item (Bound for $\widehat{v}^{(Zs^{\ddag})}_p$) 
By Minkowski's inequality and finiteness of $M'$, we have 
\begin{align*}
	\|M\|_4^2 
		&= M' \left\| \frac{1}{2n} \sum_{i=2}^n (X_i-X_{i-1})^2 \right\|_2 \\
		&\leq O\left( \frac{1}{n} \right) \sum_{i=2}^n \left\|  (X_i-X_{i-1})^2 \right\|_2 
		= O\left( \frac{1}{n} \right) \sum_{i=2}^n \left\|  X_i-X_{i-1} \right\|_4^2 \\
		&\leq O\left( \frac{1}{n} \right) \sum_{i=2}^n 
				\left( 
					\left\|  X_i-\mu_i \right\|_4 
					+ \left\| X_{i-1}-\mu_{i-1} \right\|_4 
					+ \left\vert \mu_i-\mu_{i-1} \right\vert_4 
				\right)^2.
\end{align*}
By assumption, $\| X_i-\mu_i \|_4 = \| X_{i-1} - \mu_{i-1} \|_4 = \{\E(Z_1^4)\}^{1/4} <\infty$ and 
$|\mu_i - \mu_{i-1}| \leq O(\mathcal{S}+\mathcal{C}/n)$. 
Thus, 
\begin{align}\label{eqt:Mbound4}
	\|M\|_4=O(1+\mathcal{S}+\mathcal{C}/n).
\end{align}
Using (\ref{eqt:def_s_ddag}), Minkowski's inequality and $0\leq N\leq N'<\infty$, we have
\begin{align*}
	\left\|\max_{i=1, \ldots, n} s_i^{\ddag} \right\|_4
		&= \left\|
				\sum_{\substack{j\in\{0,\ldots, N\} : t_j\leq i}} 
				\left\{ \left[Z_{t_j}-Z_{t_j-1}\right]_{-M}^M+
				\left(\mu_{t_j}-\mu_{t_j-1}\right) \right\}
			\right\|_4 \\
		&\leq 2(N+1) \left\|M \right\|_4 + (N+1) O(\mathcal{S}+\mathcal{C}/n) \\
		&\leq O(1+\mathcal{S}+\mathcal{C}/n).
\end{align*}
Thus, $\sup_{i} \| D_i^{(s^{\ddag})}\|_4 = O(1+\mathcal{S}+\mathcal{C}/n)$ as $m<\infty$.
By assumption, we clearly have $\sup_i\|D_i^{(Z)}\|_2<\infty$. 
Since at most $O(\ell)$ elements in $\{D_i^{(s^{\ddag})}\}_{i=mh+1}^n$ are non-zero, 
we have, by H\"older's inequality, that 
\begin{align}\label{eqt:vHatZsDagg}
	\|\widehat{v}^{(Zs^{\ddag})}_p\| 
		&\leq O\left\{ \frac{\ell^{2+p}}{n} \right\}
			\left\{ \sup_{i=1, \ldots,n}\|D_i^{(Z)}\|_4 \right\}
			\left\{ \sup_{i=1, \ldots,n} \| D_i^{(s^{\ddag})}\|_4 \right\} \nonumber\\
		&= O\left\{ \frac{\ell^{2+p} ( 1+ \mathcal{S}+\mathcal{C}/n)}{n} \right\}
			\nonumber\\
		&= o\{n^{-q/(1+2p+2q)}\},
\end{align}
where the last line is followed from $\mu\in\mathcal{M}_{p,q}$ with $q>1$.
\item (Bound for $\|\widehat{v}^{(Zc^{\ddag})}_p\|$) For each $j=0, \ldots, N$, 
let $n_j = t_{j+1}-t_{j}$, 
$\overline{t}_j = \sum_{i=t_j}^{t_{j+1}-1}i/n_j$, and 
$\overline{X}_j^{\dag} = \sum_{i=t_j}^{t_{j+1}-1}X_i^{\dag}/n_j$.
By Cauchy–Schwarz inequality, we have, for each $j$, that 
\begin{align*} 
	|\widehat{\beta}_{j,1}|
		&= \left\vert \frac{\sum_{i=t_j}^{t_{j+1}-1}(i-\overline{t}_j)(X_{i}^{\dag}-\overline{X}^{\dag}_j)}{\sum_{i=t_j}^{t_{j+1}-1}(i-\overline{t}_j)^2} \right\vert 
		\leq \left\{  \frac{\sum_{i=t_j}^{t_{j+1}-1}(X_{i}^{\dag}-\overline{X}^{\dag}_j)^2}{\sum_{i=t_j}^{t_{j+1}-1}(i-\overline{t}_j)^2} \right\}^{1/2} \\
		&\leq \left[  \frac{2\sum_{i=t_j}^{t_{j+1}-1}\left\{(Z_{i}-\overline{Z}_j)^2+(\mu_{i}-\overline{\mu}_j)^2+(s^{\ddag}_{i}-\overline{s}^{\ddag}_j)^2\right\}}{\sum_{i=t_j}^{t_{j+1}-1}(i-\overline{t}_j)^2} \right]^{1/2} \\
		&\leq \left[\frac{n_j \left\{ O_{\mathcal{L}^4}(1)+ O(\mathcal{C} + \mathcal{J}\mathcal{S} + M)^2\right\}}{O(n_j^3)} \right]^{1/2} \\
		&= \frac{O_{\mathcal{L}^4}(1) + O(\mathcal{C} + \mathcal{J}\mathcal{S} + M)}{n_j},
\end{align*}
where $R=O_{\mathcal{L}^4}(1)$ means that $\|R\|_4=O(1)$. 
So, together with (\ref{eqt:Mbound4}) and $0\leq N\leq N'<\infty$, we have 
\begin{align*}
	\|\widehat{v}^{(Zc^{\ddag})}_p\| 
		&\leq O\left( \frac{\ell^{1+p}}{n}\right) \left\| \sum_{i=1}^n D_i^{(Z)}D_i^{(c^{\ddag})} \right\| \\
		&\leq O\left( \frac{\ell^{1+p}}{n}\right) \left\| \sum_{j=0}^N \sum_{i=t_j}^{t_{j+1}-1} D_i^{(Z)} O(\ell)\widehat{\beta}_{j,1} \right\| \\
		&= O\left[ \frac{\ell^{2+p} \{1+\mathcal{C}+(\mathcal{J}+1)\mathcal{S}\}}{n} \right].
\end{align*}
Then, using $\mu\in\mathcal{M}_{p,q}$ with $q>1$, we obtain 
\begin{align}\label{eqt:vHatZcDagg}
	\|\widehat{v}^{(Zc^{\ddag})}_p\| 
		= o\{n^{-q/(1+2p+2q)}\}.
\end{align}

\item (Bound for $\|\widehat{v}^{(\eta\eta)}_p\|$) Similar to (\ref{eqt:vHatmumu12}), 
\begin{align*}
	\|\widehat{v}^{(\eta\eta)}_p\|
		&\leq \frac{6}{n}\sum_{|k| < \ell} |k|^p \left\vert K\left(\frac{k}{\ell}\right)\right\vert 
				\\&\qquad 
				\times\sum_{i=mh+1}^n
					 \left(\left\vert D_i^{(c)} \right\vert^2
					+ \left\vert D_i^{(s)}\right\vert^2
					+ \left\| D_i^{(c^{\ddag})} \right\|_4^2
					+ \left\| D_i^{(s^{\ddag})} \right\|_4^2\right).
\end{align*}
Using similar arguments as in (\ref{eqt:vHatmumu1}) and (\ref{eqt:vHatmumu2}),
we have 
\begin{align*}
	\frac{6}{n}\sum_{|k| < \ell} |k|^p \left\vert K\left(\frac{k}{\ell}\right)\right\vert 
		\sum_{i=mh+1}^n \left\vert D_i^{(c)} \right\vert^2
		&\leq O\left( \frac{\ell^{3+p}\mathcal{C}^2}{n^2}  \right)	, \\ 
	\frac{6}{n}\sum_{|k| < \ell} |k|^p \left\vert K\left(\frac{k}{\ell}\right)\right\vert 
		\sum_{i=mh+1}^n \left\vert D_i^{(s)} \right\vert^2
		&\leq O\left( \frac{\ell^{2+p}\mathcal{S}^2\mathcal{J}}{n}  \right)	.  
\end{align*}
Similarly,  
we also have
\begin{align*}
	\frac{6}{n}\sum_{|k| < \ell} |k|^p \left\vert K\left(\frac{k}{\ell}\right)\right\vert 
		\sum_{i=mh+1}^n \left\| D_i^{(c^{\ddag})}\right\|_4^2 
		&=O\left[ \frac{\ell^{3+p}\{1+\mathcal{C}^2+(\mathcal{J}+1)^2\mathcal{S}^2 \}}{n^2}  \right]	, \\
	\frac{6}{n}\sum_{|k| < \ell} |k|^p \left\vert K\left(\frac{k}{\ell}\right)\right\vert 
		\sum_{i=mh+1}^n \left\| D_i^{(s^{\ddag})} \right\|_4^2 
		&= O\left( \frac{\ell^{2+p}(1+\mathcal{S}^2 + \mathcal{C}^2/n)}{n}  \right)	.
\end{align*}
Using $\mu\in\mathcal{M}_{p,q}$ with $q>1$, we also have
$
	\|\widehat{v}^{(\eta\eta)}_p\|
		= o\{n^{-q/(1+2p+2q)}\}.
$
	\item (Bound for $\|\widehat{v}_p - v_p\|$) Since $\|\widehat{v}_p - v_p\| = O\{n^{-q/(1+2p+2q)}\}$ from Corollary~\ref{coro:vpHat}, 
the desired result follows. 

\end{itemize}

\subsection{Proof of Corollary~\ref{coro:multi}}\label{sec:proof_multi}  
The proof is similar to the proofs of Corollary~\ref{coro:vpHat}--\ref{coro:robustness_general} 
after trivial modifications. We detail the steps below. 
\begin{itemize}
	\item (Preliminary results) 
			The asymptotic relations (\ref{eqt:vHatp_approx1}) and (\ref{eqt:vHatp_approx2})
			are true in the multivariate setting. For example, (\ref{eqt:vHatp_approx1}) becomes 
			$\|\widehat{v}_p^{[r,s]} - \widehat{v}'^{[r,s]}_p\| 
				= O( {\ell^{1+p}}/{n} )( 1+{\|\widehat{v}_p^{[r,s]}\|}/{\ell^p} )$, 
			for each $r,s$, 
			where
			\begin{align*} 
				\widehat{v}'^{[r,s]}_p
				= \frac{1}{|\mathcal{I}|} \sum_{i\in\mathcal{I}} 
					\sum_{t,t'\in\Lambda_i} \frac{K(|t-t'|/\ell)|t-t|^p}{\ell - |t-t'|}
								D_t^{[r]} D_{t'}'^{[s]}.
			\end{align*}
			The proof follows by changing $D_t D_{t'}'$ and $\gamma_k$
			to $D_t^{[r]} D_{t'}'^{[s]}$ and $\gamma_k^{[r,s]}$, respectively. 

	\item (Bias) The bound (\ref{eqr:biasProof_result-general}) remains true after changing 
			$\widehat{v}_p$, $\gamma_k$ and $v$ to 
			$\widehat{v}_p^{[r,s]}$, $\gamma_k^{[r,s]}$ and $v^{[r,s]}$, respectively. 
			The proof is the same as the proof of Corollary~\ref{coro:vpHat}
			with $\gamma_k$ and $\widehat{\gamma}_k$
			being replaced by $\gamma_k^{[r,s]}$ and $\widehat{\gamma}_k^{[r,s]}$, respectively.
			We obtain 
			\[
				\Bias_0(\widehat{v}_p^{[r,s]}; v_p^{[r,s]})
					= \frac{B v_{p+q}^{[r,s]}}{\ell^q} + r_{p,\bias}^{[r,s]} ,
			\]
			where $r_{p,\bias}^{[r,s]} = o(1/\ell^q)+ o(\|\widehat{v}_p^{[r,s]}-v_p^{[r,s]}\|)+ O(\ell^{1+p}/n)$.
	\item (Variance)
			The variance term is a bit more complicated. 
			All results are extended trivially to the multivariate case except for Lemma~\ref{lem:momentsA}.
			In the multivariate case, (\ref{eqt:def_I_Ihat_proofLemma1}) and (\ref{eqt:def_I_Ihat_proofLemma2}) in Lemma~\ref{lem:momentsA}
			are, respectively,  updated to 
			\begin{eqnarray*}
				\widehat{I}^{[r,s]} &=& \frac{1}{n} \sum_{i=A}^{B} \sum_{i'=i-C}^{i+D} f(i/n, i'/n) X_i^{[r]} X_{i'}^{[s]}, \\
				I^{[r,s]} &=& \sqrt{v^{[r,r]} v^{[s,s]}}\int_{a}^b \int_{t-c}^{t+d} f(t,t') \,\dd \mathbb{B}_{t'}^{[r]} \,\dd \mathbb{B}_t^{[s]}, 
			\end{eqnarray*}
			where $\{\mathbb{B}_t = (\mathbb{B}_t^{[1]}, \ldots, \mathbb{B}_t^{[S]})^{\intercal} : t\in[a-c, b+d]\}$ is a $S$-dimensional Brownian motion such that $\Var(\mathbb{B}_t) = \rho t$, and 
			$\rho$ is a $S\times S$ matrix whose $(r,s)$th element is 
			$\rho^{[r,s]}=v^{[r,s]}/\sqrt{v^{[r,r]}v^{[s,s]}}$
			for $r,s\in\{1, \ldots, S\}$.
			For each $r,s$, 
			the results in that lemma are generalized as follows.
			\begin{enumerate}
				\item The weak convergence (\ref{eqt:Lemma_I_IHat_weakConv}) becomes 
						$\widehat{I}^{[r,s]}\inD I^{[r,s]}$.
						It follows from the Cram\'{e}r–Wold theorem. 
				\item We change $I_0$, $I_+$, $I_-$, $v$, and $\dd \mathbb{B}_t \,\dd \mathbb{B}_{t'}$ 
						to
						$I_0^{[r,s]}$, $I_+^{[r,s]}$, $I_-^{[r,s]}$, $\sqrt{v^{[r,r]}v^{[s,s]}}$ and 
						$\dd \mathbb{B}_t^{[r]}\,\dd \mathbb{B}_{t'}^{[s]}$, 
						respectively. 
						So, the generalized decomposition is 
					\begin{eqnarray*}
						I_0^{[r,s]} &=& \sqrt{v^{[r,r]}v^{[s,s]}}\rho^{[r,s]}\int_{a}^b f(t,t) \, \dd t, \\
						I_-^{[r,s]} &=& \sqrt{v^{[r,r]}v^{[s,s]}}\int_a^b \int_{t-c}^{t} f(t,t') \,\dd \mathbb{B}_{t'}^{[r]} \,\dd \mathbb{B}_t^{[s]}, \\
						I_+^{[r,s]} &=& \sqrt{v^{[r,r]}v^{[s,s]}}\int_a^{b+d} \int_{\max(a,t'-d)}^{\min(b,t')} f(t,t') \,\dd \mathbb{B}_{t}^{[r]} \,\dd \mathbb{B}_{t'}^{[s]} .
					\end{eqnarray*}
						In $I_0^{[r,s]}$, the extra factor $\rho^{[r,s]}$ appears when we apply It\^{o} isometry:
						$\dd \mathbb{B}_{t}^{[r]} \,\dd \mathbb{B}_{t}^{[s]} = \rho^{[r,s]} \,\dd t$. 
				\item The first moment is trivially extended to $\E(I^{[r,s]}) = I_0^{[r,s]}$. 
						However, the multivariate extension of the second moment 
						(\ref{eqt:VarI_lemma}) is non-trivial. The generalized result is 
					\begin{align*}
						\Var(I^{[r,s]}) 
							&= \sqrt{v^{[r,r]}v^{[s,s]}} \bigg\{
									\int_a^b \int_{t-c}^{t+d} f^2(t,t') \, \dd t' \, \dd t\nonumber \\
							&\qquad+ \rho^{[r,s]} 
								\int_a^b \int_{\max\{a,t-\min(c,d)\}}^{\min\{b,t+\min(c,d)\}} f(t,t') f(t',t) \, \dd t' \, \dd t \bigg\}.
					\end{align*}
					The proof is similar to the proof of Lemma~\ref{lem:momentsA} 
					except the following two points.
					First, when we apply It\^{o} isometry on the squared terms 
					$\E(I^{[r,s]}_+)^2$ and $\E(I^{[r,s]}_-)^2$, 
					we use $\dd\mathbb{B}_t^{[r]}\, \dd\mathbb{B}_t^{[r]} = \dd t$
					and $\dd\mathbb{B}_t^{[s]}\, \dd\mathbb{B}_t^{[s]} = \dd t$. 
					Second, when we apply It\^{o} isometry on the cross term $\E(I^{[r,s]}_+I^{[r,s]}_-)$, 
					we use $\dd \mathbb{B}_{t}^{[r]} \,\dd \mathbb{B}_{t}^{[s]} = \rho^{[r,s]} \,\dd t$. 
			\end{enumerate}
			After that we applied the generalized Lemma~\ref{lem:momentsA}
			to obtain  
			the generalized (\ref{eqt:varLemmaStatement}):
\begin{align}\label{eqt:varLemmaStatementGeneral}
	\Var_0(\widehat{v}'^{[r,s]}) \sim \frac{4w^{[r,s]}\ell}{n} \left\{ \int_0^1 K^2(u)\, \dd u \right\} \left\{ \sum_{|s|\leq m} \delta_s^2 \right\}.
\end{align}
			The proof of (\ref{eqt:varLemmaStatementGeneral}) follows by doing the following  
			trivial modifications.
			For simplicity, we denote $R_n = O_{\mathcal{L}^2}(r_n)$ if $\|R_n\|_2 =O(r_n)$.
			Using this notation, 
			we have $\widehat{v}'^{[r,s]} = \widehat{V}^{[r,s]} + O_{\mathcal{L}^2}(\ell^2/n)$, 
			where 
\begin{align*}\label{eqt:hatV_form2_Gen}
	\widehat{V}^{[r,s]}	= \frac{\ell}{n} \sum_{b=b_0}^{b_1} \widehat{I}_b^{[r,s]} 
	\quad \text{and}\quad
	\widehat{I}_b^{[r,s]} = 
					\sum_{i = (b-1)\ell+1}^{b\ell} 
					\sum_{i'=i-\ell}^{i+\ell} K\left( \frac{i-i'}{\ell}\right) \frac{D_i^{[r]} D_{i'}^{[s]}}{\ell}.
\end{align*}
			After some straightforward generalization, 
			(\ref{eqt:varLemma_inter1}) becomes 
\[
	\Var\left( \sum_{b=1}^{b'} \widehat{I}_b^{[r,s]}  \right) 
		\rightarrow b' C_1^{[r,s]} 
			+ C_0^{[r,s]}, 
\]
where $C_0^{[r,s]}$ is a constant that does not depends on $b'$, and 
\[
	C_1^{[r,s]} = 4 w^{[r,s]}\left\{ \int_{0}^{1} K^2(u)\, \dd u  \right\} 
			\left\{ \sum_{|k|\leq m} \delta_{|k|}^2 \right\}.
\]
		Following the remaining steps in Lemma~\ref{lem:blocking_vHatp}, 
		we obtain $\Var(\widehat{V}^{[r,s]}) \sim C_1^{[r,s]}\ell/n$. 
		Then we have the desired (\ref{eqt:varLemmaStatementGeneral}). 
		Finally, following the remaining steps in Corollary~\ref{coro:vpHat}, we obtain 
		\[	 
			\Var_0(\widehat{v}_p^{[r,s]})
			=  \frac{4A_p\ell^{1+2p} w^{[r,s]} \Delta_m}{n} + \left(r_{p,\se}^{[r,s]}\right)^2, 
		\]
		where 
		$\left(r_{p,\se}^{[r,s]}\right)^2 = o(\ell^{1+2p}/n)+ o(\|\widehat{v}_p^{[r,s]}-v_p^{[r,s]}\|^2)$.
	\item (Robustness) All results follow after trivial multivariate generalization. 
\end{itemize}

\section{Auxiliary results}\label{sec:auxillary}
\subsection{Lemma~\ref{lem:rough_var}}\label{sec:proof_rough_var}
\begin{lemma}\label{lem:rough_var}
Let $X_{1:n}$ be a stationary time series satisfying Assumption~\ref{ass:weakDep}.
Suppose that 
$1/\ell + (\ell+mh)/n \rightarrow 0$, and 
Assumption \ref{ass:summability_dj} holds.
(a) Then $\Var(\widehat{v}_{\diff}) \lesssim (\ell+mh)/n$.
(b) Suppose, in addition, one of the following conditions is satisfied:
\begin{enumerate}
	\item $m<\infty$ and $h/\ell\rightarrow\lambda_{\infty}\in[1, \infty]$;
	\item $m\rightarrow \infty$ and $h/\ell\rightarrow\infty$;
	\item $m\rightarrow \infty$, $h/\ell \rightarrow \lambda_{\infty}\in[1, \infty)$, and Assumption \ref{ass:d_uncorrelatedDiff} holds.
\end{enumerate}
$\Var(\widehat{v}_{\diff}) \lesssim \ell/n$,
where the upper bound is achievable in all three cases.
\end{lemma}

\begin{proof}
Note that $\ell=\ell_n$, $h=h_n$ and $m=m_n$ may depend on $n$, 
but we drop the subscript $n$ for  notational simplicity.
Note that $Z_i$ is $\mathcal{F}_i$-measurable, so we can represent it as 
$Z_i = \sum_{j=-\infty}^i \mathcal{P}_{j} Z_i$, 
where $\mathcal{F}_i$ is the $\sigma$-algebra generated by $\{\ldots, \varepsilon_{i-1}, \varepsilon_i\}$.
So,  
\begin{align}\label{eqt:CovSampleACVF}
	&\Cov(\widehat{\gamma}_k, \widehat{\gamma}_{k'}) \nonumber\\
	&\sim \frac{1}{n^2}\sum_{i,i'=1}^n \Cov(X_iX_{i-k}, X_{i'}X_{i'-k'}) \nonumber \\
	&=\frac{1}{n^2}\sum_{i,i'=1}^n \E\left\{ X_iX_{i-k} (X_{i'}X_{i'-k'}-\gamma_{k'}) \right\}  \nonumber \\
	&= \frac{1}{n^2}\sum_{i,i'=1}^n \E\left\{ 
			\left(\sum_{j=-\infty}^i \mathcal{P}_{j}X_i\right) 
			\left( \sum_{j'=-\infty}^i \mathcal{P}_{j'}X_{i-k} \right)
			\left(\sum_{j''=-\infty}^{i'} \mathcal{P}_{j''} (X_{i'}X_{i'-k'}) \right) \right\}. \nonumber \\ 
\end{align}						
Since, for any random variables $A$ and $B$, we know that 
$\mathcal{P}_j A$ and $\mathcal{P}_{j'} B$ are uncorrelated for $j\neq j'$, 
we can simplify (\ref{eqt:CovSampleACVF}) as 
\[
	\Cov(\widehat{\gamma}_k, \widehat{\gamma}_{k'}) 
	\sim \frac{1}{n^2}\sum_{i,i'=1}^n \sum_{j=-\infty}^{\max(i,i')}
			\E\left[ \left(\mathcal{P}_{j}X_i\right) \left( \mathcal{P}_{j}X_{i-k} \right) \left\{\mathcal{P}_{j} (X_{i'}X_{i'-k'}) \right\} \right] , 
\]
thus 
\begin{align}
	|\Cov(\widehat{\gamma}_k, \widehat{\gamma}_{k'})|
	\leq \frac{1}{n^2}\sum_{i,i'=1}^n \sum_{j=-\infty}^{\max(i,i')}
			\| \mathcal{P}_{j}X_i \|_4 
			\| \mathcal{P}_{j}X_{i-k} \|_4 
			\| \mathcal{P}_{j} (X_{i'}X_{i'-k'}) \|_2 , \label{eq:covACVF_depMea}
\end{align}
where the last line follows from the generalized H\"{o}lder's inequality 
for a product of three random variables.
Define $\omega_{p,i} = \| \mathcal{P}_0X_i \|_p$ for each $i$.
Note that for any $j,a,b$, we have, by stationarity, that $\| \mathcal{P}_jX_a \|_4 = \omega_{4,a-j}$ and 
\begin{align*}
	\| \mathcal{P}_j X_a X_b \|_2 
	&= \| \E(X_aX_b \mid \mathcal{F}_j) - \E(X_aX_b \mid \mathcal{F}_{j-1})\|_2\\
	&= \| \E(X_aX_b - X_{a,\{j\}} X_{b,\{j\}} \mid \mathcal{F}_j) \|_2 \\
	&= \left\| \E\left\{ X_a(X_b -X_{b,\{j\}}) + X_{b,\{j\}}(X_a- X_{a,\{j\}}) \mid \mathcal{F}_j \right\} \right\|_2 \\
	&\leq \| X_a \|_4 \| X_b -X_{b,\{j\}} \|_4 +  \|X_{b,\{j\}} \|_4 \|X_a- X_{a,\{j\}}\|_4 \\
	&= C_2 \left\{ \omega_{4,b-j} + \omega_{4,a-j} \right\},
\end{align*}
for some finite $C_2>0$, 
where $X_{a,\{j\}}$ is a coupled version of $X_a=g(\ldots, \varepsilon_{a-1},\varepsilon_a)$ 
with the $\varepsilon_j$ in $X_a$ being replaced by an i.i.d. copy $\varepsilon_j'$; 
see Section~\ref{sec:notations} for the detailed definition. 
Then (\ref{eq:covACVF_depMea}) can be bounded as follows:
\begin{eqnarray}\label{eqt:covACVF_final}
	|\Cov(\widehat{\gamma}_k, \widehat{\gamma}_{k'}) | 
	 \lesssim \frac{1}{n^2}\sum_{i,i'=1}^n \sum_{j=-\infty}^{\max(i,i')} 
			\omega_{4,i-j} \omega_{4,i-k-j} ( \omega_{4,i'-j} + \omega_{4,i'-k'-j}) . 
\end{eqnarray}						
By Theorem 1 of \cite{wu2005} and the assumption that $\Theta_4<\infty$, we have 
\begin{align}
	\sum_{i\in\mathbb{Z}} \omega_{4,i} 
	\leq \sum_{i\in\mathbb{Z}} \theta_{4,i} 
	= \sum_{i=0}^{\infty} \theta_{4,i} 
	= \Theta_{4} <\infty. \label{eqt:sumOmega}
\end{align}	
Now, let $T$ be an integer to be specified. 
Using (\ref{eqt:covACVF_final}) and (\ref{eqt:sumOmega}), 
we have   
\begin{align}
	& \sum_{-T\leq k,k' \leq T} \left\vert 
			\Cov(\widehat{\gamma}_k, \widehat{\gamma}_{k'}) \right\vert \nonumber \\
		&\qquad\lesssim  \frac{1}{n^2}\sum_{-T\leq k,k' \leq T} \sum_{i,i'=1}^n \sum_{j=-\infty}^{\max(i,i')} 
			\omega_{4,i-j} \omega_{4,i-k-j} ( \omega_{4,i'-j} + \omega_{4,i'-k'-j}) \nonumber\\
		&\qquad\leq \frac{1}{n^2}\sum_{|k'|\leq T} \sum_{i=1}^n 
				 \sum_{j\in\mathbb{Z}} \omega_{4,i-j}
				 \left( \sum_{k\in\mathbb{Z}} \omega_{4,i-k-j} \right)
				 \left[ \sum_{i'\in\mathbb{Z}} ( \omega_{4,i'-j} + \omega_{4,i'-k'-j} )  \right] \nonumber\\
		&\qquad= \frac{1}{n^2}\sum_{|k'|\leq T} \sum_{i=1}^n 
				 \left\{\sum_{j=-\infty}^{\infty} \omega_{4,i-j}
				 \left( \sum_{k=-\infty}^{\infty} \omega_{4,k} \right)
				 \left( 2\sum_{i'=-\infty}^{\infty} \omega_{4,i'}  \right) \right\} \nonumber\\
		&\qquad\leq \frac{1}{n^2}\sum_{|k|\leq T} \sum_{i=1}^n 
				 \left\{\sum_{j\in\mathbb{Z}} \omega_{4,j}
				 \left( \sum_{k\in\mathbb{Z}} \omega_{4,k} \right)
				 \left( 2\sum_{i'\in\mathbb{Z}} \omega_{4,i'} \right) \right\} \nonumber\\
		&\qquad= \frac{1}{n^2}\sum_{|k|\leq T} \sum_{i=1}^n  O(1)  = O\left( T/n \right), \label{eqt:sumAbsCov}
\end{align}	
i.e., 
$\sum_{-T\leq k,k' \leq T} \left\vert \Cov(\widehat{\gamma}_k, \widehat{\gamma}_{k'}) \right\vert \lesssim T/n$.
For (a), let $L = \ell+mh$. Since $K$ is bounded, we have 
\begin{align}
	\Var(\widehat{v}_{\diff})
		&= \sum_{-L\leq k,k' \leq L} K_{\diff}\left(\frac{k}{\ell} 
				\right)K_{\diff}\left(\frac{k'}{\ell} \right) 
				\Cov(\widehat{\gamma}_k, \widehat{\gamma}_{k'}) \label{eqt:varvhatDiff_general}\\
		&=  O(1) \sum_{-\ell\leq k,k' \leq \ell} 
										\left\vert \Cov(\widehat{\gamma}_k, \widehat{\gamma}_{k'}) \right\vert \nonumber \\
		&= O(L/n), \nonumber
\end{align}	
where the last line is obtained by using (\ref{eqt:sumAbsCov}) with $T=L$. 
For (b), we consider the following three cases. 
\begin{enumerate}
	\item Case 1: $m<\infty$. 
			We further divide it into two sub-cases.  
			\begin{itemize}
				\item Case 1.1: $h/\ell\rightarrow\infty$. For a large enough $n$, we have $K_{\diff}(t) = K(t)$, 
						which is zero when $|t|>1$. 
						Using (\ref{eqt:varvhatDiff_general}) and (\ref{eqt:sumAbsCov}) with $T=\ell$, we have 
						$\Var(\widehat{v}_{\diff})=O(\ell/n)$.
				\item Case 1.2: $h/\ell\rightarrow\lambda_{\infty}\in[1, \infty)$.
						Using (\ref{eqt:varvhatDiff_general}) and (\ref{eqt:sumAbsCov}) with $T=L$, we have 
						$\Var(\widehat{v}_{\diff})=O(L/n)=O(\ell/n)$.
			\end{itemize}
	\item Case 2: $m\rightarrow \infty$ and $h/\ell\rightarrow\infty$. 
			For a large enough $n$, we have $K_{\diff}(t) = K(t)$.
			Hence, similar to Case 1.1, we have  $\Var(\widehat{v}_{\diff}) = O(\ell/n)$.
	\item Case 3: $m\rightarrow \infty$, $h/\ell \rightarrow \lambda_{\infty}\in[1, \infty)$, and Assumption \ref{ass:d_uncorrelatedDiff} holds.
			For $|k|>\ell$, using Assumption \ref{ass:d_uncorrelatedDiff}, 
			we know that $K_{\diff}(k/\ell) = O(1/m)$.
			Then, using (\ref{eqt:sumAbsCov}) with $T=\ell$ and $T=L$, we have 
			\begin{align*}
				\Var(\widehat{v}_{\diff})
					&\leq  \left\{ O(1) \sum_{|k|,|k'| \leq \ell} 
							\left\vert \Cov(\widehat{\gamma}_k, \widehat{\gamma}_{k'}) \right\vert\right\}
							+ \left\{ O(1/m^2) \sum_{\ell< |k|,|k'| \leq L} 
							\left\vert \Cov(\widehat{\gamma}_k, \widehat{\gamma}_{k'}) \right\vert \right\}\\
					&=O\left( \frac{\ell}{n} \right) + O\left( \frac{L}{nm^2} \right) 
					= O\left( \frac{\ell}{n} \right).
			\end{align*}	
\end{enumerate}
Finally, we show that 
the upper bound of the variance is achievable.
In this lemma, we assume $h/\ell\rightarrow\lambda_{\infty}\in[1,\infty]$. 
By Proposition \ref{prop:effKernel}, we know that, for $|k|\leq \ell$, 
the value of $K_{\diff}(k/\ell)$ can only be 
$K(k/\ell)$ or $K(k/\ell)+\delta_1\{K((k+h)/\ell) + K((k-h)/\ell)\}$. 
By continuity of $K$, we know that $\asymp\ell$ elements of 
$\{ K_{\diff}(k/\ell) : |k|\leq \ell\}$ are bounded away from zero. 
Consider $X_i = \theta \varepsilon_{i-1} + \varepsilon_{i}$, where $\varepsilon_i$ follow $\Normal(0,1)$ independently 
and $\theta>0$. 
By Proposition 7.3.4 of \citet{brockwellDavis1991}, 
we know that 
$n\Cov(\widehat{\gamma}_k, \widehat{\gamma}_{k'}) 
	\rightarrow C_{k,k'}$, 
where $C_{k,k'}>0$ is bounded away from zero for $||k|-|k'||\leq 2$, 
and $C_{k,k'}=0$ for $||k|-|k'||> 2$.
Consequently, $\Var(\widehat{v}_{\diff}) \gtrsim \ell/n$. 
Thus, we obtain the desired results. 
\end{proof}

\subsection{Lemma~\ref{lem:momentsA}}
\begin{lemma}\label{lem:momentsA}
Let $\{X_{i}\}$ be a stationary time series satisfying Assumption~\ref{ass:weakDep}.
Let $A,B,C,D\in\mathbb{Z}$ be deterministic numbers 
such that $n=(B+D)-(A-C)\rightarrow \infty$, $A/n\rightarrow a$, $B/n\rightarrow b$, $C/n\rightarrow c$, $D/n\rightarrow d$ and the limits satisfies $a\leq b$ and $c,d\geq 0$.  
Also let $f:\mathbb{R}^2 \mapsto \mathbb{R}$ be a deterministic function.  
Define  
\begin{eqnarray}
	\widehat{I} &=&  \frac{1}{n} \sum_{i=A}^{B} \sum_{i'=i-C}^{i+D} f(i/n, i'/n) X_i X_{i'}, 
	\label{eqt:def_I_Ihat_proofLemma1}\\
	I &=& v\int_{a}^b \int_{t-c}^{t+d} f(t,t') \,\dd \mathbb{B}_{t'} \,\dd \mathbb{B}_t, 
	\label{eqt:def_I_Ihat_proofLemma2}
\end{eqnarray}
where $\{\mathbb{B}_t : t\in[a-c, b+d]\}$ is a standard Brownian motion.  
\begin{enumerate}
	\item \label{lem:momentsA_limit} (Weak limit) 
			The double sum $\widehat{I}$ can be asymptotically 
			represented as the stochastic integral $I$ in the following sense: 
			\begin{eqnarray}\label{eqt:Lemma_I_IHat_weakConv}
				\widehat{I} \inD I.
			\end{eqnarray}
	\item \label{lem:momentsA_represent} (Adapted representation) 
		The stochastic integral can be decomposed into three parts: $I = I_0 + I_- + I_+$, 
		where 
		\begin{align*}
			I_0 &= v\int_{a}^b f(t,t) \, \dd t, \\
			I_- &= v\int_a^b \int_{t-c}^{t} f(t,t') \,\dd \mathbb{B}_{t'} \,\dd \mathbb{B}_t, \\
			I_+ &= v\int_a^{b+d} \int_{\max(a,t'-d)}^{\min(b,t')} f(t,t') \,\dd \mathbb{B}_{t} \,\dd \mathbb{B}_{t'} .
		\end{align*}
	\item \label{lem:momentsA_moments} (Moments) The first two moments of $I$ are given by 
		\begin{align}
			\E(I) &= I_0 , \nonumber \\
			\Var(I) &= v^2\int_a^b \int_{t-c}^{t+d} f^2(t,t') \, \dd t' \, \dd t \nonumber\\
					& \qquad+ v^2\int_a^b \int_{\max\{a,t-\min(c,d)\}}^{\min\{b,t+\min(c,d)\}} f(t,t') f(t',t) \, \dd t' \, \dd t. \label{eqt:VarI_lemma}
		\end{align}
			
\end{enumerate}
\end{lemma}

\begin{proof}[Proof of Lemma~\ref{lem:momentsA}]
(\ref{lem:momentsA_limit}) 
By Assumption~\ref{ass:weakDep}, we have 
the functional central limit theorem (Theorem 3 of \cite{wu2005}):
\[
	\left\{ \frac{1}{\sqrt{n}}
	\sum_{a-c \leq i/n \leq t} X_i  \right\}_{t\in[a-c,b+d]}
	\inD 
	\left\{ \sqrt{v} \mathbb{B}_t \right\}_{t\in[a-c,b+d]}.
\]
Then, the result follows from the continuous mapping theorem.
(\ref{lem:momentsA_represent}) Let 
$a-c := t_{A-C} < t_{A-C+1} < \ldots < t_{B+D} =: b+d$, 
where $A,B,C,D\in\mathbb{Z}$. 
The set $\mathcal{T} = \{ t_i : i=A-C, \ldots, B+D\}$ partitions $[a-c, b+d)$ into $n = (B+D) - (A-C)$
disjoint intervals of equal width, i.e.,  
\[
	[a-c, b+d) = \bigcup_{i=A-C}^{B+D-1} [t_i, t_{i+1}). 
\]
The mesh size $\|\mathcal{T}\| = \inf_{A-C \leq i < B+D} (t_{i+1}-t_i) = (b+d-a+c)/n$
is strictly decreasing with $n$. 
Then the stochastic integral $I$ can be defined as the $\mathcal{L}^2$ limit as follows: 
\begin{align*}
	I 	&= v\lim_{\|\mathcal{T}\|\rightarrow 0} \sum_{i=A}^B\sum_{i'=i-C}^{i+D} f(t_i, t_{i'})
			\left( \mathbb{B}_{t_{i'}} - \mathbb{B}_{t_{i'-1}}\right)
			\left( \mathbb{B}_{t_i} - \mathbb{B}_{t_{i-1}}\right)  \nonumber \\
		&= v\lim_{\|\mathcal{T}\|\rightarrow 0} \sum_{i=A}^B 
			\bigg\{
				f(t_i, t_i) \left( \mathbb{B}_{t_{i}} - \mathbb{B}_{t_{i-1}}\right)\nonumber \\
				&\qquad+
				\left(\sum_{i'=i-C}^{i-1} + \sum_{i'=i+1}^{i+D}\right)  
				f(t_i, t_{i'})
				\left( \mathbb{B}_{t_{i'}} - \mathbb{B}_{t_{i'-1}}\right)
			\bigg\}
			\left( \mathbb{B}_{t_i} - \mathbb{B}_{t_{i-1}}\right) \nonumber \\
		&= I_0 + I_- + I_+,  \label{eqt:proof_intergal_L2_limit}  
\end{align*}
where 
\begin{align*}
	I_0 &= v\left\{
				\lim_{\|\mathcal{T}\|\rightarrow 0} 
				\sum_{i=A}^B f(t_i, t_i) (t_i-t_{i-1})
			\right\},\\
	I_- &= v\left\{
				\lim_{\|\mathcal{T}\|\rightarrow 0} 
				\sum_{i=A}^B
				\sum_{i'=i-C}^{i-1} f(t_i, t_{i'})
				\left( \mathbb{B}_{t_{i'}} - \mathbb{B}_{t_{i'-1}}\right)
				\left( \mathbb{B}_{t_i} - \mathbb{B}_{t_{i-1}}\right) 
			\right\} ,\nonumber
			\\
	I_+ &= v\left\{
				\lim_{\|\mathcal{T}\|\rightarrow 0} 
				\sum_{i'=A+1}^{B+D}
				\sum_{i= \max(A,i'-D)}^{\min(B,i'-1)}
				f(t_i, t_{i'})
				\left( \mathbb{B}_{t_i} - \mathbb{B}_{t_{i-1}}\right)
				\left( \mathbb{B}_{t_{i'}} - \mathbb{B}_{t_{i'-1}}\right) 
			\right\}   .
\end{align*}
Note that $I_+$ is obtained after swapping two summations.
Thus the desired results follow by the definition of It\^{o} integrals. 

For (\ref{lem:momentsA_moments}), we start with the first moment. 
The integrands of the outer integrals in $I_{\pm}$ are adapted to their respective outer integrals, i.e., 
\[
	\int_{t-c}^{t} f(t,t') \,\dd \mathbb{B}_{t'} \in \mathcal{F}_{t}
	\qquad \text{and} \qquad
	\int_{\max(a,t'-d)}^{\min(b,t')} f(t,t') \,\dd \mathbb{B}_{t}  \in \mathcal{F}_{t'},
\] 
where $\mathcal{F}_{u}$ is the filtration of $\{\mathbb{B}_s : s\leq u\}$ up to time $u$. 
So, $\E(I_+) = \E(I_-) = 0$. 
Consequently, $\E(I) = I_0$ (because $I_0$ is non-stochastic). 
Next, we work on the variance. Since $\E(I_-+ I_+)=0$, 
we have $\E(I^2) = \E(I_+^2) + \E(I_-^2) + 2\E(I_- I_+) + I_0^2$, 
which implies that 
\begin{eqnarray}\label{eqt:varA_expanded}
	\Var(I) = \E(I_+^2) + \E(I_-^2) + 2\E(I_- I_+).
\end{eqnarray}
Applying It\^{o} isometry twice and then swapping the integrals, we have 
\begin{eqnarray}
	\E(I_+^2) 
			&=& v^2\int_a^{b+d} \E\left\{  \int_{\max(a,t'-d)}^{\min(b,t')} f(t,t') \, \dd \mathbb{B}_t\right\}^2\, \dd t'\nonumber \\
			&=& v^2\int_a^{b+d}  \int_{\max(a,t'-d)}^{\min(b,t')} f^2(t,t') \, \dd t \, \dd t' \nonumber\\
			&=& v^2\int_a^{b}  \int_{t}^{t+d} f^2(t,t') \, \dd t' \, \dd t. \label{eqt:itoIso_1}
\end{eqnarray}
Similarly, we have $\E(I_-^2) = v^2\int_a^{b}  \int_{t-c}^{t} f^2(t,t') \, \dd t' \, \dd t$.
Summing $\E(I_+^2)$ and $\E(I_-^2)$, we have 
\begin{eqnarray}\label{eqt:varA_expanded_squares}
	\E(I_+^2)+\E(I_-^2) = v^2\int_a^{b}  \int_{t-c}^{t+d} f^2(t,t') \, \dd t' \, \dd t.
\end{eqnarray}
To calculate the cross product term $\E(I_- I_+)$, we first re-express $I_-$ and $I_+$ 
by inserting an indicator in the integrand of $I_-$ and changing dummy indexes in $I_+$, respectively: 
\begin{align*}
	I_- &= v\int_a^{b+d} \int_{t-c}^{t} f(t,t')\mathbb{1}(t< b) \,\dd \mathbb{B}_{t'} \,\dd \mathbb{B}_t, \\
	I_+ &= v\int_a^{b+d} \int_{\max(a,t-d)}^{\min(b,t)} f(t',t) \,\dd \mathbb{B}_{t'} \,\dd \mathbb{B}_{t} .
\end{align*}
Then, applying It\^{o} isometry again, we have 
\begin{align*}
	\E(I_+ I_-) 
			&= v^2\int_a^{b+d} \E\left[
					\left\{\int_{t-c}^{t} f(t,t')\mathbb{1}(t< b) \,\dd \mathbb{B}_{t'}\right\}
					\left\{\int_{\max(a,t-d)}^{\min(b,t)} f(t',t) \,\dd \mathbb{B}_{t'}\right\}
			\right]\, \dd t \\
			&= v^2\int_a^{b} \int_{\max(a,t-d,t-c)}^{t} f(t',t)f(t,t') \,\dd t' \, \dd t .
\end{align*}
By exchanging integrals, we obtain 
\begin{align}
	2\E(I_+ I_-) 
			&= v^2\int_a^{b} \int_{\max(a,t-d,t-c)}^{t} f(t',t)f(t,t') \,\dd t' \, \dd t \nonumber\\
			& \qquad\qquad+v^2\int_a^{b} \int_{t}^{\min(b,t-d,t-c)} f(t',t)f(t,t') \,\dd t' \, \dd t  \nonumber\\
			&= v^2\int_a^b \int_{\max\{a,t-\min(c,d)\}}^{\min\{b,t+\min(c,d)\}} f(t,t') f(t',t) \, \dd t' \, \dd t. \label{eqt:varA_expanded_cross}
\end{align}
Substituting (\ref{eqt:varA_expanded_squares}) and (\ref{eqt:varA_expanded_cross}) into (\ref{eqt:varA_expanded}), 
we obtain the desired result. 
\end{proof}

\subsection{Proposition~\ref{prop:acvf_lagk}}   
\begin{proposition}\label{prop:acvf_lagk}
For any $x_1, \ldots, x_n$, 
define $\widehat{\gamma}_k^x = \sum_{i=1+|k|}^n (x_i - \bar{x})(x_{i-|k|}- \bar{x})/n$, where $\bar{x} = \sum_{i=1}^n x_i/n$. 
Then the following identity holds for all $k\in\{0, \pm 1, \ldots, \pm (n-1)\}$:
\[
	\widehat{\gamma}_{k}^x
	\equiv \widehat{\gamma}_0^x  - \frac{1}{2n} \sum_{j=1+|k|}^n (x_j - x_{j-|k|})^2 - \frac{1}{2n}\left\{ \sum_{j=1}^{|k|} (x_j-\bar{x})^2 + 
						\sum_{j=n-|k|+1}^n (x_j-\bar{x})^2\right\}.
\]
\end{proposition}

\begin{proof}[Proof of Proposition~\ref{prop:acvf_lagk}]
Let $\overline{x^2} = \sum_{i=1}^n x_i^2/n$. For $k\geq 0$, we have  
\begin{align}\label{eqt:gammaExpansion1}
	&\widehat{\gamma}_0^x  - \frac{1}{2n} \sum_{j=1+k}^n (x_j - x_{j-k})^2 \nonumber\\
		&\qquad=  \frac{1}{n} \sum_{i=1}^n (x_i^2 + \bar{x}^2 - 2x_i \bar{x})  
				- \frac{1}{2n} \sum_{j=1+k}^n (x_j^2 + x_{j-k}^2 + 2x_j x_{j-k}) \nonumber\\
		&\qquad= \overline{x^2} - \bar{x}^2 
			- \left[ \overline{x^2} 
				- \frac{1}{2n}\left( \sum_{j=1}^k x_j^2 + \sum_{j=n-k+1}^n x_j^2 \right)
				- \frac{1}{n}\sum_{j=1+k}^n x_j x_{j-k}\right] \nonumber\\
		&\qquad= \frac{1}{n}\sum_{j=1+k}^n x_j x_{j-k} - \bar{x}^2 
				+\frac{1}{2n}\left( \sum_{j=1}^k x_j^2 + \sum_{j=n-k+1}^n x_j^2 \right) . \nonumber\\
\end{align}
Note that 
\begin{align}\label{eqt:gammaExpansion2}
	\frac{1}{n}\sum_{j=1+k}^n x_j x_{j-k}
	= \widehat{\gamma}_{k}^x 
				+ \bar{x}^2\left(\frac{n-k}{n} \right)
				- \frac{\bar{x}}{n} \left( \sum_{j=1}^k x_j + \sum_{j=n-k+1}^n x_j \right).
\end{align}
So, putting (\ref{eqt:gammaExpansion2}) into (\ref{eqt:gammaExpansion1}), we have 
\begin{align*}
	\widehat{\gamma}_0^x  - \frac{1}{2n} \sum_{j=1+k}^n (x_j - x_{j-k})^2 
		= \widehat{\gamma}_{k}^x 
				+ \frac{1}{2n}\left\{ \sum_{j=1}^k (x_j-\bar{x})^2 + 
						\sum_{j=n-k+1}^n (x_j-\bar{x})^2\right\}.
\end{align*}
Similar results can be obtained for $k<0$.
Rearranging the terms, we obtain the desired result. 
\end{proof}

\subsection{Lemma~\ref{lem:size_vHatpmumu}}  
\begin{lemma}\label{lem:size_vHatpmumu}
Assume all conditions in Theorem~\ref{thm:robustness}.
Define $\widehat{v}^{(\mu\mu)}$ as in (\ref{eqt:vhatp_ab}). 
Then
\begin{align}\label{eqt:size_vHatpmumu}
	\widehat{v}^{(\mu\mu)} 
		= 	\left\{ \kappa\ell\mathcal{V} + O\left( \frac{\ell}{n}\right) \right\} \mathbb{1}_{(m=0)}
			+ O\left[ \frac{\ell}{n} \left\{  (\ell\mathbb{1}_{m=0} +mh ) \mathcal{S}^2\mathcal{J}
					+  \frac{(\ell\mathbb{1}_{m=0} +mh )^2\mathcal{C}^2}{n}\right\} \right],
\end{align}
where the upper bound is achievable by some mean functions.
\end{lemma}

\begin{proof}[Proof of Lemma~\ref{lem:size_vHatpmumu}]
We begin with the case $m>0$. 
Recall that $\mu_i = c_i + s_i$, where  
$c_i := c(i/n)$ and 
$s_i := s(i/n) = \sum_{j=0}^{\mathcal{J}} \xi_{j} \mathbb{1}(T_j \leq i < T_{j+1})$ for each $i$. 
Similar to (\ref{eqt:DDmuZ}), we define 
\[
	D_i^{(s)} = \sum_{j=0}^m d_j s_{i-jh}
	\qquad \text{and} \qquad
	D_i^{(c)}=\sum_{j=0}^m d_j c_{i-jh}
\]
for each $i$.
Using the inequality $|xy| \leq x^2 + y^2$, we have 
\begin{align*}
	\left\vert D_i^{(\mu)} D_{i-|k|}^{(\mu)} \right\vert 
		&=\left\vert D_i^{(c)} D_{i-|k|}^{(c)} 
					+D_i^{(c)} D_{i-|k|}^{(s)}
					+D_i^{(s)} D_{i-|k|}^{(c)}
					+D_i^{(s)} D_{i-|k|}^{(s)}
					\right\vert\\
		&\leq 3\left( 
					\left\vert D_i^{(c)} D_{i}^{(c)} \right\vert
					+\left\vert D_i^{(s)} D_{i}^{(s)}\right\vert
					+\left\vert D_{i-k}^{(c)} D_{i-k}^{(c)} \right\vert
					+\left\vert D_{i-k}^{(s)} D_{i-k}^{(s)}\right\vert
					\right).
\end{align*}
So, we have 
\begin{align}
	|\widehat{v}^{(\mu \mu)}|
		&\leq \frac{1}{n}\sum_{|k| < \ell} \left\vert K\left(\frac{k}{\ell}\right)\right\vert 
				\sum_{i=mh+|k|+1}^n
					\left\vert D_i^{(\mu)} D_{i-|k|}^{(\mu)} \right\vert  \nonumber \\
		&\leq \frac{6}{n}\sum_{|k| < \ell} \left\vert K\left(\frac{k}{\ell}\right)\right\vert 
				\sum_{i=mh+1}^n 
				\left( 
					\left\vert D_i^{(c)} \right\vert^2
					+\left\vert D_i^{(s)} \right\vert^2
					\right).	\label{eqt:vHatmumu12}
\end{align}
Note that, for any $i$, 
if there is no change point within the time interval $i-mh, \ldots, i$, 
then $D_i^{(s)}=0$. 
Since there are $\mathcal{J}$ change points, 
at most $(mh+1)\mathcal{J}$ elements in $\{D_i^{(s)} \}_{i=mh+1}^n$ are non-zero. 
By assumption, the non-zero $D_i^{(s)}$ satisfies 
$|D_i^{(s)}| \leq \sum_{j=0}^{m}|d_j| \times O(\mathcal{S}) = O(\mathcal{S})$
as $\sup_{n\in\mathbb{N}}\sum_{j=0}^m|d_j|<\infty$ and $\mathcal{G} \gtrsim \ell+mh$.
So,
\begin{align}
	\frac{6}{n}\sum_{|k| < \ell} \left\vert K\left(\frac{k}{\ell}\right)\right\vert 
				\sum_{i=mh+|k|+1}^n
				\left\vert D_i^{(s)} \right\vert^2
	\leq O\left( \frac{mh\ell \mathcal{S}^2\mathcal{J}}{n}\right). 
	\label{eqt:vHatmumu1}
\end{align}
Similarly, for any $i$, we have $|D_i^{(c)}| = O( \mathcal{C}mh/n )$.
Hence, 
\begin{align}
	\frac{6}{n}\sum_{|k| < \ell} \left\vert K\left(\frac{k}{\ell}\right)\right\vert 
				\sum_{i=mh+|k|+1}^n
				\left\vert D_i^{(c)}\right\vert^2
	\leq  O\left( \frac{m^2h^2\ell\mathcal{C}^2}{n^2}\right) .	
	\label{eqt:vHatmumu2}
\end{align}
Putting (\ref{eqt:vHatmumu1}) and (\ref{eqt:vHatmumu2}) into (\ref{eqt:vHatmumu12}), 
we obtain 
\[
	|\widehat{v}^{(\mu\mu)}| 
		\leq O\left( \frac{mh\ell \mathcal{S}^2\mathcal{J}}{n}\right)
				+O\left( \frac{m^2h^2\ell\mathcal{C}^2}{n^2}\right).
\]

When $m=0$, we have $D_i^{(\mu)} = \mu_i- \overline{\mu}$ for all $i$, and 
$K_{\diff}(\cdot) \equiv K(\cdot)$ by Proposition~\ref{prop:effKernel} (\ref{item:KdiffProperty}). 
Using Proposition~\ref{prop:acvf_lagk}, we have
\begin{align}
	\widehat{v}^{(\mu\mu)} 
		&= \sum_{|k| < \ell} K\left(\frac{k}{\ell}\right) 
				\frac{1}{n} \sum_{i=|k|+1}^n (\mu_i-\overline{\mu})( \mu_{i-|k|} -\overline{\mu})\nonumber\\
		&= \sum_{|k| < \ell} K\left(\frac{k}{\ell}\right) 
				\left\{
				\frac{1}{n}\sum_{i=1}^n(\mu_i-\overline{\mu})^2
				- \frac{1}{2n} \sum_{j=1+|k|}^n  (\mu_j - \mu_{j-|k|})^2
				- r_k
				\right\} , \label{eqt:vHatmumuCase2}
\end{align}
where $r_0=0$, $r_{k}= r_{-k}$ and, for $k=1,2,\ldots,\ell-1$, 
\[
	r_k =  \frac{1}{2n}\left\{ \sum_{j=1}^k (\mu_j-\bar{\mu})^2 + 
					\sum_{j=n-k+1}^n (\mu_j-\bar{\mu})^2\right\} = O(\ell\mathcal{V}/n).
\]
Similar to (\ref{eqt:vHatmumu1}) and (\ref{eqt:vHatmumu2}), we have 
\begin{align*}
	\frac{1}{2n} \sum_{j=1+|k|}^n  (\mu_j - \mu_{j-|k|})^2
	&\leq \frac{1}{2n} \sum_{j=1+|k|}^n \left\{  2(s_j - s_{j-|k|})^2+2(c_j - c_{j-|k|})^2 \right\}\\
	&\leq O\left( \frac{\mathcal{S}^2\mathcal{J}\ell}{n}\right)
			+ O\left( \frac{\mathcal{C}^2\ell^2}{n^2}\right).
\end{align*}
Together with the Riemann approximations 
$\sum_{|k| < \ell} K(k/\ell) = \ell\{\kappa+O(1/\ell)\}$
and $\sum_{i=1}^n(\mu_i-\overline{\mu})^2/n = \mathcal{V} + O(1/n)$, 
we can simplfy (\ref{eqt:vHatmumuCase2}) by $\kappa \ell \mathcal{V}$ in the following sense:
\begin{align*}
	\widehat{v}^{(\mu\mu)} 
		&= \ell\left\{ \kappa + O(1/\ell)\right\}\left\{ \mathcal{V} + O(1/n)\right\}\\
		&\qquad
			+ O(\ell)\left\{ 
			O\left( \frac{\mathcal{S}^2\mathcal{J}\ell}{n}\right)
			+ O\left( \frac{\mathcal{C}^2\ell^2}{n^2}\right)
			+ O(\ell\mathcal{V}/n)
			\right\} \\
		&= \kappa\ell\mathcal{V} + O\left( \frac{\ell}{n}\right) +
			O\left( \frac{\mathcal{S}^2\mathcal{J}\ell^2}{n}\right)
			+ O\left( \frac{\mathcal{C}^2\ell^3}{n^2}\right), 
\end{align*}
where the last line follows from $\ell/n = o(1)$ and $\kappa \neq 0$.
Finally, it is easy to check that the upper bound is achievable.
Hence, the desired result follows. 
\end{proof}

\subsection{Lemma~\ref{lem:Var_vHatp_muZ}} 
\begin{lemma}\label{lem:Var_vHatp_muZ}
Assume all conditions in Theorem~\ref{thm:robustness}.
Define $\widehat{v}^{(\mu Z)}$ as in (\ref{eqt:vhatp_ab}). 
Then
\begin{align*} 
	\left\| \widehat{v}^{(\mu Z)} \right\| 
		= O\left\{ \frac{\ell}{\sqrt{n}} \left( \mathcal{C} +\mathcal{S} \mathcal{J}  \right)\right\} \mathbb{1}_{(m=0)}
			+  O\left\{ \frac{\ell}{\sqrt{n}} \left( \frac{mh\mathcal{C}}{n} + \mathcal{S} \left(\frac{mh\mathcal{J}}{n}\right)^{1/2} \right)\right\} .
\end{align*}
\end{lemma}

\begin{proof}
Similar to Lemma~\ref{lem:size_vHatpmumu}, 
we let $c_i = c(i/n)$ and $s_i = s(i/n)$ for $i=1, \ldots, n$. 
Similar to (\ref{eqt:vHat_quad}), 
we have 
\begin{align*}
	\widehat{v}^{(\mu Z)} 
		&= \sum_{i=1}^n\sum_{i'=1}^n
			\left\{ \frac{1}{n} K\left( \frac{i-i'}{\ell} \right) \mathbb{1}_{(\min(i,i')\geq mh+1)} 
				\left(D_{i'}^{(c)}+D_{i'}^{(s)}\right) D_{i}^{(Z)} \right\}.
\end{align*}
So, by Minkowski's inequality, 
\begin{align}\label{eqt:normvHatmuZ12}
	\left\|  \widehat{v}^{(\mu Z)}  \right\| 
		&\leq \left\| \sum_{i=1}^n \alpha_{c,i} D_{i}^{(Z)}\right\| 
				+ \left\| \sum_{i=1}^n \alpha_{s,i} D_{i}^{(Z)}\right\|,
\end{align}
where 
\begin{align*}
	\alpha_{c,i} &=\sum_{i'=1}^n
			\left\{ \frac{1}{n} K\left( \frac{i-i'}{\ell} \right) \mathbb{1}_{(\min(i,i')\geq mh+1)} 
				D_{i'}^{(c)} \right\} ,\\
	\alpha_{s,i} &=\sum_{i'=1}^n
			\left\{ \frac{1}{n} K\left( \frac{i-i'}{\ell} \right) \mathbb{1}_{(\min(i,i')\geq mh+1)}
				D_{i'}^{(s)} \right\} .
\end{align*}
We consider two cases: $m>0$ and $m=0$.

Suppose $m>0$.
Note that $|\alpha_{c,i}|\leq  O(\mathcal{C}mh\ell/n^2)$ for $i=1, \ldots, n$, 
Using (\ref{eqt:stableD}) and Lemma 1 of \citet{wu2010}, we have 
\begin{align}\label{eqt:normvHatmuZ1}
	\left\| \sum_{i=1}^n \alpha_{c,i} D_{i}^{(Z)}\right\| 
	\leq \sqrt{n \cdot  O\left(  \frac{ \mathcal{C}mh\ell}{n^2}  \right)^2}
		=  O\left( \frac{mh\ell\mathcal{C}}{n^{3/2}} \right).
\end{align}
Note also that at most $O(\ell\mathcal{J})$ elements in $\{\alpha_{s,i}\}_{i=mh+1}^n$ are non-zero. 
Those non-zero values satisfy 
$|\alpha_{s,i}| = O(\ell\mathcal{S}/n )$.
By Lemma 1 of \citet{wu2010} again, we have 
\begin{align}\label{eqt:normvHatmuZ2}
	\left\| \sum_{i=1}^n \alpha_{s,i} D_{i}^{(Z)}\right\| 
		\leq \sqrt{O( mh\mathcal{J} ) \cdot O \left( \frac{\ell\mathcal{S}}{n}  \right)^2}
		=  O\left( \frac{m^{1/2}h^{1/2}\ell\mathcal{S}\mathcal{J}^{1/2}}{n} \right) .
\end{align}
Putting (\ref{eqt:normvHatmuZ1}) and (\ref{eqt:normvHatmuZ2}) into (\ref{eqt:normvHatmuZ12}), 
we have 
\[
	\left\|  \widehat{v}^{(\mu Z)}  \right\| 
	=  O\left\{ \frac{\ell}{\sqrt{n}} \left( \frac{mh\mathcal{C}}{n} + \mathcal{S} \left(\frac{mh\mathcal{J}}{n}\right)^{1/2} \right)\right\} .
\]

Suppose $m=0$. We have $\alpha_{c,i}= O(\mathcal{C}\ell/n)$ 
and $\alpha_{s,i}= O(\mathcal{JS}\ell/n)$ 
for $i=1, \ldots, n$. 
Similarly, by Lemma 1 of \citet{wu2010} again, we have 
\[
	\left\|  \widehat{v}^{(\mu Z)}  \right\| 
	 = O\left\{ \frac{\ell}{\sqrt{n}} \left( \mathcal{C} +\mathcal{S} \mathcal{J}  \right)\right\}.
\] 
Thus, the desired result follows.
\end{proof}

\subsection{Lemma~\ref{prop:varSum_bApC}}
\begin{proposition}\label{prop:varSum_bApC}
Let $T_0, T_2 \ldots\in\mathcal{L}^2$ be a stationary time series. 
If there exists $C_0,C_1\in\mathbb{R}$ such that, for any $k=0,1,\ldots$, 
\begin{eqnarray}\label{eqt:varSum_cond}
	\Var\left( \sum_{i=0}^k T_i \right) = (k+1) C_1 + C_0,
\end{eqnarray}
then 
\begin{eqnarray}\label{eqt:varSum_result}
	\Cov(T_0, T_k) 
		= \left\{ \begin{array}{ll} 
			C_1+C_0  & \text{if $k=0$;}	\\
			-C_0/2 & \text{if $k=1$;}	\\
			0 & \text{if $k\geq 2$.}
		\end{array}\right.
\end{eqnarray} 
\end{proposition}

\begin{proof}[Proof of Proposition~\ref{prop:varSum_bApC}]
It is not hard to see that 
\begin{align}
	\Var\left( \sum_{i=0}^k T_i \right)
		&= \sum_{|h|\leq k} (k+1 -|h|)\Cov(T_0, T_h) \nonumber \\ 
		&= (k+1)\Cov(T_0, T_h) + 2\sum_{h=1}^k (k+1 -h)\Cov(T_0, T_h). \label{eqt:varSum}
\end{align}
The result is trivial for $k=0$. 
For $k=1$, applying the identity (\ref{eqt:varSum}) and the assumption (\ref{eqt:varSum_cond}), we have 
\begin{align*}
	2 C_1 + C_0 = \Var\left( T_0 + T_1 \right) 
		&= \Var(T_0) + \Var(T_1) + 2\Cov(T_0,T_1) \\
		&= 2(C_1+C_0) + 2\Cov(T_0,T_1).
\end{align*}
Rearranging the terms, we have $\Cov(T_0,T_1) = -C/2$. 
Similarly, for $k=2$, we have 
\begin{align*}
	3 C_1 + C_0 = \Var\left( T_0 + T_1 +T_2\right) 
		&= \sum_{|h|\leq 2} (3-|h|)\Cov(T_0, T_h)\\
		&= 3(C_1+C_0) + 2\left\{ 2(-C_0/2) + \Cov(T_0, T_2)\right\}.
\end{align*}
Hence, $\Cov(T_0, T_2)=0$. 
Assume that there exists $k'\geq 2$ such that the (\ref{eqt:varSum_result}) is true for all $k\leq k'$. 
Then, we use (\ref{eqt:varSum}) and (\ref{eqt:varSum_cond}) to obtain 
\begin{align*}
	&(k'+2) C_1 + C_0 \\
		&\qquad= \Var\left( \sum_{i=0}^{k'+1}T_i \right) \\
		&\qquad= \sum_{|h|\leq k'+1} (k'+1-|h|)\Cov(T_0, T_h)\\
		&\qquad= (k'+2)(C_1+C_0) + 2\left\{ 2(-C_0/2) + 0 + \ldots + 0 + \Cov(T_0, T_{k'+1})\right\}.
\end{align*}
It implies that $\Cov(T_0, T_{k'+1}) = 0$. 
By mathematical indication, (\ref{eqt:varSum_result}) is also true for all $k\geq 2$.
\end{proof}

\subsection{Lemma~\ref{lem:blocking_vHatp}}\label{sec:proof_blocking_vHatp}
\begin{lemma}\label{lem:blocking_vHatp}
Let $\widehat{v}'$ be defined as (\ref{eqt:Stype_vHat}).
Suppose that the assumptions in Theorem~\ref{thm:var} are satisfied. 
Then 
\begin{align}\label{eqt:varLemmaStatement}
	\Var(\widehat{v}') \sim \frac{4v^2\ell}{n} \left\{ \int_0^1 K^2(t)\, \dd t \right\} \left\{ \sum_{|s|\leq m} \delta_s^2 \right\}.
\end{align}
\end{lemma}

\begin{proof}[Proof of Lemma~\ref{lem:blocking_vHatp}]
For simplicity, we denote $R_n = O_{\mathcal{L}^2}(r_n)$ if $\|R_n\|_2 =O(r_n)$.
Let $b_0$ be the smallest integer such that $(b_0-1)\ell \geq mh+\ell$;
and $b_1$ be the largest integer such that $b_1\ell \leq n$. 
Then 
\begin{eqnarray*}
	\widehat{v}' 
		&=& \frac{1}{n} \sum_{i=(b_0-1)\ell+1}^{b_1\ell} \sum_{|k|\leq \ell} K\left( \frac{k}{\ell} \right) D_i D_{i-k} + O_{\mathcal{L}^2}\left( \frac{\ell^2}{n} \right) \\
		&=& \widehat{V} + O_{\mathcal{L}^2}\left( \frac{\ell^2}{n} \right) ,
\end{eqnarray*}
where 
\[
	\widehat{V} = \frac{1}{n} \sum_{b=b_0}^{b_1} \sum_{i=(b-1)\ell+1}^{b\ell} \sum_{|k|\leq \ell} K\left( \frac{k}{\ell} \right) D_i D_{i-k}.
\]
Hence, we know that 
$\| \widehat{v}' - \E(\widehat{v}') \| =  
		\| \widehat{V}-\E(\widehat{V}) \|  + O(\ell^2/n)$.
Then, it suffices to find $\Var(\widehat{V}) = \| \widehat{V}-\E(\widehat{V}) \|^2$. 
Now, write 
\begin{align}\label{eqt:hatV_form2}
	\widehat{V}	= \frac{\ell}{n} \sum_{b=b_0}^{b_1} \widehat{I}_b,
	\quad \text{where} \quad
	\widehat{I}_b = \sum_{i = (b-1)\ell+1}^{b\ell} 
					\sum_{i'=i-\ell}^{i+\ell} K\left( \frac{i-i'}{\ell}\right) \frac{D_i D_{i'}}{\ell}.
\end{align}
For any $b'\in\mathbb{N}$, we consider 
\begin{align*}
	\sum_{b=1}^{b'} \widehat{I}_b	
		&= \frac{1}{\ell} \sum_{i = 1}^{b'\ell} 
					\sum_{i'=i-\ell}^{i+\ell} K\left( \frac{i-i'}{\ell}\right) 
					\sum_{j,j'=0}^m d_j d_{j'} X_{i'-j'h}X_{i-jh} \\
		&= \frac{1}{\ell} \sum_{j,j'=0}^m \sum_{i = 1-jh}^{b'\ell-jh} 
					\sum_{i'=i+jh-\ell- j'h}^{i+jh+\ell- j'h} K\left( \frac{i-i'+ h(j-j')}{\ell}\right) 
					d_j d_{j'}X_iX_{i'} \\
		&= \frac{1}{\ell} 
					\sum_{i = 1-mh-(\ell+mh)}^{b'\ell}  
					\sum_{i'=i-\ell-mh}^{i+\ell+mh} 
					\Bigg\{\sum_{j,j'=0}^m 
						d_j d_{j'}K\left( \frac{i-i'+h(j-j')}{\ell} \right) 
						\\&\qquad
					\times\mathbb{1}_{\left(1-jh\leq i \leq b'\ell-jh\right)} 
					\mathbb{1}_{\left(h(j-j')-\ell \leq i'-i \leq h(j-j')+\ell \right)}
					\Bigg\}
					X_iX_{i'}.
\end{align*}
Using Lemma~\ref{lem:momentsA}, we know that 
\begin{align}\label{eqt:A1A2_var}
	\Var\left( \sum_{b=1}^{b'} \widehat{I}_b  \right)\nonumber
		&\rightarrow \Var\left\{ \int_{-m\lambda-(1+m\lambda)}^{b'} \int_{t-(1+m\lambda)}^{t+(1+m\lambda)} \psi(t,t') \, \dd\mathbb{B}_{t'}\, \dd\mathbb{B}_{t} \right\} \\
		&= \Psi_1 + \Psi_2,
\end{align}
where
\begin{align*}
	\psi(t,t') &= \sum_{j,j'=0}^m 
						d_j d_{j'} K\left( t-t' + \lambda (j-j')\right) \\
			& \qquad\qquad\times \mathbb{1}_{\left(-\lambda j \leq t \leq b'-\lambda j\right)}
					\mathbb{1}_{\left(\lambda (j-j')-1 \leq t'-t\leq  \lambda (j-j')+1\right)},
\end{align*}
and
\begin{eqnarray*}
	\Psi_1 &=& \int_{-m\lambda-(1+m\lambda)}^{b'} \int_{t-(1+m\lambda)}^{t+(1+m\lambda)} \psi^2(t,t') \, \dd t' \, \dd t ,\\
	\Psi_2 &=& \int_{-m\lambda-(1+m\lambda)}^{b'} \int_{t-(1+m\lambda)}^{t+(1+m\lambda)} \psi(t,t')\psi(t',t) \, \dd t' \, \dd t.
\end{eqnarray*}
Next, we compute $\Psi_1$ and $\Psi_2$ one by one. 
First, let $u = t'-t$, then we can express $\psi(t,t')$ as 
\begin{align*}
	\psi(t,t') &= \sum_{j,j'=0}^m 
						d_j d_{j'}K\left( -u + \lambda (j-j')\right) 
					\mathbb{1}_{\left(-\lambda j \leq t \leq b'-\lambda j\right)}
					\mathbb{1}_{\left(\lambda (j-j')-1\leq u \leq \lambda (j-j')+1\right)} \\
			&= \sum_{|k|\leq m} K\left( -u + \lambda k\right) \chi_k(t) \mathbb{1}_{(\lambda k - 1\leq u \leq\lambda k+1)},  
\end{align*}
where  
\[
	\chi_k(t) = \sum_{j=\max(0,k)}^{m+\min(0,k)} d_j d_{j-k} 
				\mathbb{1}_{\{-\lambda j \leq t \leq b'-\lambda j\}}.
\]
Since $\lambda\geq 2$, 
we have
\begin{align*}
	\Psi_1 &= \int_{-m\lambda-(1+m\lambda)}^{b'} \int_{-(1+m\lambda)}^{1+m\lambda} 
				\left\{ \sum_{|k|\leq m} K(-u+\lambda k) \chi_k(t) \mathbb{1}_{(\lambda k - 1\leq u \leq \lambda k +1)} \right\}^2 \, \dd u \, \dd t \\
		&= \int_{-m\lambda-(1+m\lambda)}^{b'} \int_{-(1+m\lambda)}^{1+m\lambda} 
				\sum_{|k|\leq m} K^2(-u+\lambda k) \chi_k^2(t) \mathbb{1}_{(\lambda k - 1\leq u \leq \lambda k +1)} \, \dd u \, \dd t \\
		&= \sum_{|k|\leq m} \left\{ \int_{k\lambda-1}^{k\lambda+1} K^2(-u+\lambda k)  \, \dd u \right\} \\
		& \qquad \qquad\times
			 \int_{-m\lambda-(1+m\lambda)}^{b'} \left\{\sum_{j=\max(0,k)}^{m+\min(0,k)} d_j d_{j-k} 
				\mathbb{1}_{\{-\lambda j \leq t \leq b'-\lambda j\}} \right\}^2 \, \dd t \\
		&= \sum_{|k|\leq m} \left\{ \int_{-1}^{1} K^2(u)\, \dd u  \right\} 
			\left\{ \mathop{\sum\sum}_{j,j'=\max(0,k)}^{m+\min(0,k)}
					\int_{-\lambda \min(j,j')}^{b'-\lambda \max(j,j')} 
							d_j d_{j-k} d_{j'} d_{j'-k} 
					\, \dd t \right\} \\
		&= \sum_{|k|\leq m} \left\{ 2\int_{0}^{1} K^2(u)\, \dd u  \right\} 
			\left\{ \mathop{\sum\sum}_{j,j'=\max(0,k)}^{m+\min(0,k)} (b'-|j-j'|)d_j d_{j-k} d_{j'} d_{j'-k} \right\} \\
		&= b' \left\{ 2\int_{0}^{1} K^2(u)\, \dd u  \right\} 
			\left\{ \sum_{|k|\leq m} \delta_{|k|}^2 \right\} +C_0',
\end{align*}
where $C_0'$ does not depend on $b'$.
Similarly, for $\Psi_2$, we also let $u=t'-t$. Using $\lambda\geq 2$, we have 
\begin{align}\label{eqt:crossProd_psi}
	\psi(t,t') \psi(t',t)
		&= \psi(t,u+t) \psi(u+t,t) \nonumber \\
		&= \mathop{\sum\sum}_{-m\leq k,k'\leq m} 
			K(-u+\lambda k)K(u+\lambda k')
			\chi_k(t) \chi_{k'}(u+t)  \nonumber \\
		&\qquad\qquad\qquad\times	\mathbb{1}_{(\lambda k - 1\leq u \leq\lambda k+1)}
			\mathbb{1}_{(\lambda k' - 1\leq -u\leq\lambda k'+1)}  \nonumber \\
		&= \sum_{|k|\leq m} 
			K(-u+\lambda k)K(u-\lambda k)
			\chi_k(t) \chi_{-k}(u+t) 
			\mathbb{1}_{(\lambda k - 1\leq u \leq\lambda k+1)}\nonumber \\
\end{align}
for all $t$ and $u$ except countably many values. 
Using (\ref{eqt:crossProd_psi}), we have 
\begin{align*}
	\Psi_2 &= \int_{-m\lambda-(1+m\lambda)}^{b'}
			\int_{-(1+m\lambda)}^{1+m\lambda} 
				\psi(t,u+t) \psi(u+t,t) \, \dd u\, \dd t \\
		&= \sum_{|k|\leq m} \int_{-m\lambda-(1+m\lambda)}^{b'}
			\int_{\lambda k -1}^{\lambda k+1} 
			K(-u+\lambda k)K(u-\lambda k)
			\chi_k(t) \chi_{-k}(u+t) \, \dd u\, \dd t \\
		&= \sum_{|k|\leq m} 
			\int_{\lambda k -1}^{\lambda k+1} 
			K(-u+\lambda k)K(u-\lambda k) \sum_{j=\max(0,k)}^{m+\min(0,k)}
			\sum_{j'=\max(0,-k)}^{m+\min(0,-k)}
			\\&\qquad\qquad	d_jd_{j-k}d_{j'}d_{j'+k}
			\int_{-m\lambda-(1+m\lambda)}^{b'}\mathbb{1}_{(-\lambda j \leq u \leq b'-\lambda j - u)}  \, \dd t\, \dd u \\
		&= \sum_{|k|\leq m} 
			\delta_k^2
			\int_{\lambda k -1}^{\lambda k+1} 
			K(-u+\lambda k)K(u-\lambda k) \cdot (b'-u) \, \dd u \\
		&= b' \left\{ \sum_{|k|\leq m} 
			\delta_k^2
			\int_{-1}^{1} 
			K(-u)K(u)  \, \dd u \right\} + C_0'' \\
		&= b' \left\{ 2\int_{0}^{1} K^2(u)\, \dd u  \right\} 
			\left\{ \sum_{|k|\leq m} \delta_{|k|}^2 \right\} + C_0'',
\end{align*}
where $C_0''$ does not depends on $b'$.
Consequently, (\ref{eqt:A1A2_var}) can be simplified to 
\begin{align}\label{eqt:varLemma_inter1}
	\Var\left( \sum_{b=1}^{b'} \widehat{I}_b  \right) 
		\rightarrow b' C_1 
			+ C_0, 
\end{align}
where $C_0 = C_0' + C_0''$ and 
\[
	C_1 = 4\left\{ \int_{0}^{1} K^2(u)\, \dd u  \right\} 
			\left\{ \sum_{|k|\leq m} \delta_{|k|}^2 \right\}.
\]
By Proposition~\ref{prop:varSum_bApC}, we know that, for any finite $|b-b'|$,  
\[
	\Cov\left( \widehat{I}_b, \widehat{I}_{b'}\right) 
		\inD \left\{ \begin{array}{ll} 
			C_1+C_0  & \text{if $b-b'=0$;}	\\
			-C_0/2 & \text{if $|b-b'|=1$;}	\\
			0 & \text{if $|b-b'|\geq 2$.}
		\end{array}\right.
\]
We are going to show that $\Cov\left( \widehat{I}_b, \widehat{I}_{b'}\right) \rightarrow 0$
is also true for large $|b-b'|$. 
Without loss of generality, assume that $b<b'$. 
Note that the latest observation involved in $\widehat{I}_b$ is $X_{b\ell}$, 
whereas the earliest observation involved in $\widehat{I}_{b'}$ is $X_{(b'-1)\ell+1-\ell-mh}$. 
Hence, if $b'-b>2+\lambda m$, the observations used in $\widehat{I}_b$ and in $\widehat{I}_{b'}$
are fully non-overlapping. 
Using Lemma~\ref{lem:momentsA}, we know that there exists a functional $G$ and 
two non-overlapping intervals $\Lambda_b, \Lambda_{b'}$ such that 
\[
	\begin{bmatrix} \widehat{I}_b \\ \widehat{I}_{b'} \end{bmatrix}
		\inD \begin{bmatrix} G(\{\mathbb{B}_{t}: t\in \Lambda_b\})\\ G(\{\mathbb{B}_{t}: t\in \Lambda_{b'}\}) \end{bmatrix},
\]
where $\{\mathbb{B}_t : t\in\mathbb{R}\}$ is a standard Brownian motion on $\mathbb{R}$. 
Since $\{\mathbb{B}_{t}: t\in \Lambda_b\}$ and $\{\mathbb{B}_{t}: t\in \Lambda_{b'}\}$
are independent, we know that $\Cov(\widehat{I}_b, \widehat{I}_{b'})\rightarrow 0$. 

In view of (\ref{eqt:hatV_form2}) and $\ell(b_1-b_0)\sim n$, we have 
\begin{eqnarray*}
	\Var(\widehat{V}) 
		&=& \frac{\ell^2}{n^2} \sum_{|k|\leq b_1-b_0} (b_1-b_0+1-|k|) \Cov(\widehat{I}_0, \widehat{I}_k) \\
		&\sim& \frac{\ell^2}{n^2} \left\{ (b_1-b_0+1)(C_1+C_0) + 2(b_1-b_0)(-C_0/2)  + o(b_1-b_0) \right\} \\
		&\sim& \frac{\ell}{n} \left\{(C_1+C_0) + 2(-C_0/2)  + o(1) \right\} \\
		&=& \frac{C_1\ell}{n}  .
\end{eqnarray*}
Hence, 
\[
	\Var(\widehat{v}') 
		= \frac{4\ell}{n}\left\{ \int_{0}^{1} K^2(u)\, \dd u  \right\} 
			\left\{ \sum_{|k|\leq m} \delta_{|k|}^2 \right\} + o(\ell/n).
\]
\end{proof}

\subsection{Derivation of (\ref{eqt:compareMSEexisting})}\label{sec:derivationOfCompareMSEexisting}
Assume $K=K_{\Bart}$.
By equation (19) of \citet{chan2020} with their suggested parameter $c_1=1$, we know that 
\[
	n^{2/3} \MSE\{\widehat{v}_{(\Chan)}\} /v^2
		\rightarrow  48^{1/3}  (v_1/v_0)^{2/3}.
\]
According to the discussion in Section 4.2 of \citet{chan2020}, 
we know that $\MSE\{\widehat{v}_{(\Wu)}\} \sim (-v_1/\ell)^2 + 7 v^2 \ell /(2n)$. 
Minimizing the MSE with respect to $\ell$, we obtain 
\[
	n^{2/3} \MSE\{\widehat{v}_{(\Wu)}\} /v^2
		\rightarrow (1323/16)^{1/3} (v_1/v_0)^{2/3}.
\]
From (\ref{eqt:optMSE}), we have 
\[
	n^{2/3} \MSE\{\widehat{v}_{(3)}\} /v^2
		\rightarrow (49/3)^{1/3}  (v_1/v_0)^{2/3}.
\]
Hence, we obtain 
\begin{eqnarray*}
	\frac{\MSE\{\widehat{v}_{(\Chan)}\}}{\MSE\{\widehat{v}_{(3)}\}} 
		\rightarrow (12/7)^{2/3} \approx 1.43
	\quad \text{and} \quad
	\frac{\MSE\{\widehat{v}_{(\Wu)}\}}{\MSE\{\widehat{v}_{(3)}\}} 
		\rightarrow (3/2)^{4/3} \approx 1.71 . 
\end{eqnarray*}

\section{Additional plots}\label{sec:additional_plots}
Additional simulation results for Section~\ref{sec:sim_MSE} are shown below. 
Figure~\ref{fig:dTAVC_A032_n400} and Figure~\ref{fig:dTAVC_A032_n200_full}
show the results for $n=200$ and $n=400$, respectively.

\begin{figure}[h]
\begin{center}
\includegraphics[width=\textwidth]{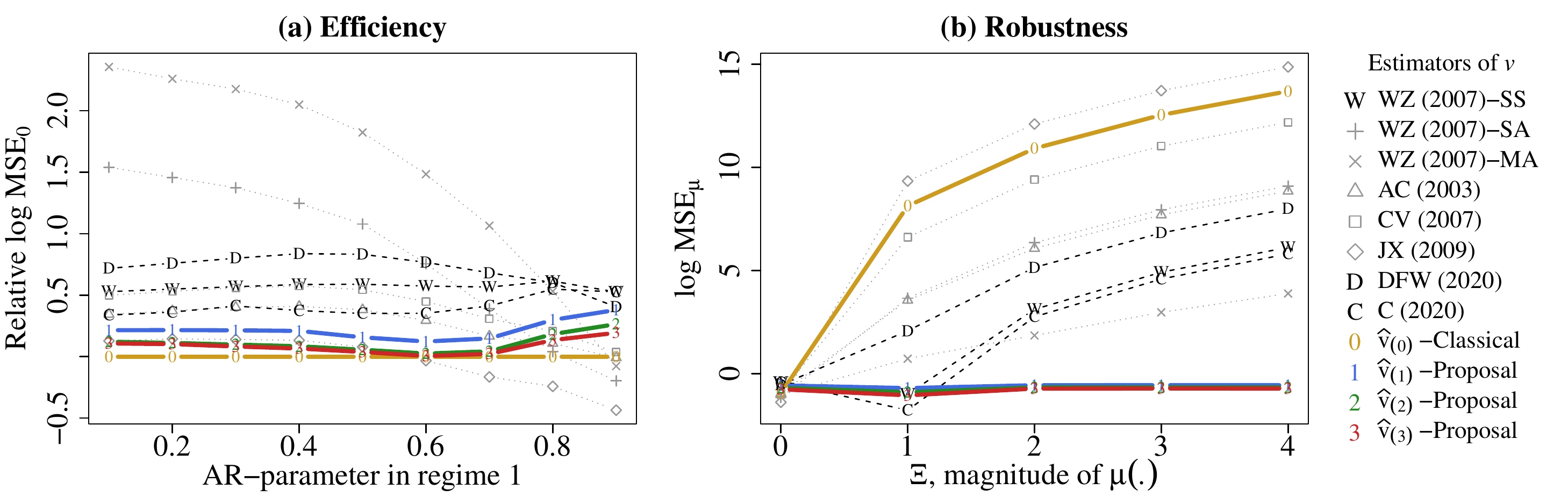} 
\end{center} 
\caption{The caption in Figure~\ref{fig:dTAVC_A032_n200} also applies here. 
		The results above are based on time series of size $n=200$,
		where SS, SA and MA in WZ (2007) refer to 
		the estimators that use sum of squared differences, sum of absolute differences, 
		and median of absolute differences, respectively.}
		\label{fig:dTAVC_A032_n200_full}
\end{figure}

\begin{figure}[h]
\begin{center}
\includegraphics[width=\textwidth]{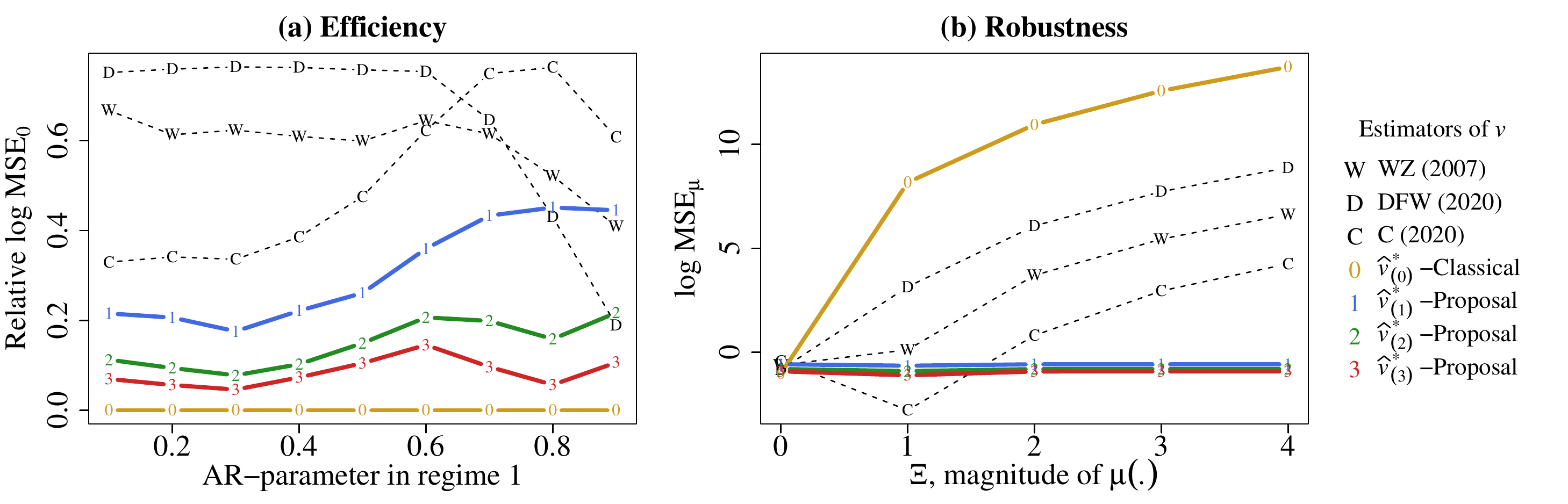} 
\end{center} 
\caption{The caption in Figure~\ref{fig:dTAVC_A032_n200} also applies here. 
		The results above are based on time series of size $n=400$ (instead of $n=200$).}
		\label{fig:dTAVC_A032_n400}
\end{figure}

\section{Additional simulation experiments}\label{sec:robustness_experiment_additional}
Additional simulation experiments are performed in this section. 
We consider the following three mean functions:
\begin{itemize}
	\item[(a)] $\mu(t) = \mathcal{C} t$;  
	\item[(b)] $\mu(t) = \mathcal{S}\left\{ \mathbb{1}(t>1/4)+\mathbb{1}(t>2/4)+\mathbb{1}(t>3/4) \right\}$; and  
	\item[(c)] $\mu(t) = \mathcal{C} t + \mathcal{S}\left\{\mathbb{1}(t>1/4)+\mathbb{1}(t>2/4)+\mathbb{1}(t>3/4)\right\}$. 
\end{itemize}
Case (a) considers a continuous mean function with a Lipschitz constant $\mathcal{C}$. 
Case (b) considers a step discontinuous mean function with a maximal step size $\mathcal{S}$.
Case (c) is a superposition of the mean functions in cases (a) and (b). 
The same noise sequence is used as in part (b) of the experiment in Section~\ref{sec:sim_MSE}.
The MSEs of different estimators in cases (a)--(c) when $n\in\{200,400\}$ are plotted in Figure~\ref{fig:dTAVC_K015_full}.

The proposed estimators (i.e., $\widehat{v}_{(0)}^{\star},\ldots, \widehat{v}_{(3)}^{\star}$)
are the most robust estimators in all three cases. 
When $n=200$, we notice that the MSEs of the proposed estimators are slightly inflated when $\mathcal{S}$ 
is small but non-zero. 
It is because the the proposed estimators are only robust asymptotically. 
When the sample size increases to $n=400$, this phenomenon does not exist.
Moreover, the inflated MSEs quickly decline back to the null values when $\mathcal{S}$ further increases
because the proposed rough centering procedure enhances the finite-sample robustness.

\begin{figure}[h]
\begin{center}
\includegraphics[width=\textwidth]{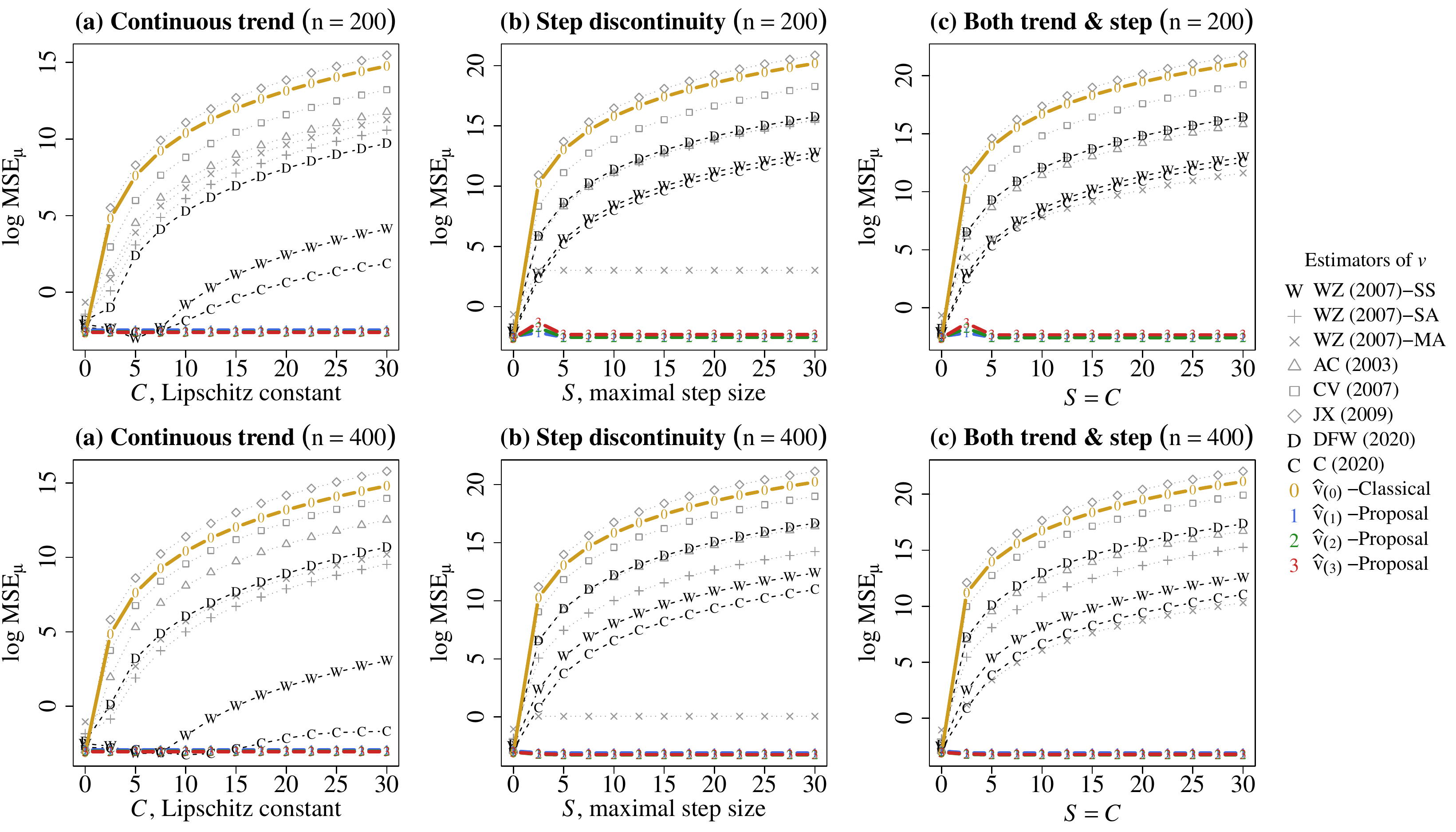} 
\end{center} 
\caption{The values of $\log \MSE_\mu(\cdot)$ against $\mathcal{C}$ or $\mathcal{S}$ 
under the mean functions defined in Section~\ref{sec:robustness_experiment_additional}.}
		\label{fig:dTAVC_K015_full}
\end{figure}

\end{document}